\documentclass[11pt, showpacs,onecolumn,aps,pra,longbibliography,superscriptaddress,notitlepage]{revtex4-2}
\usepackage{qcircuit}
\usepackage[dvips]{graphicx}
\usepackage{amsmath,amssymb,amsthm,mathrsfs,amsfonts,dsfont}
\usepackage{subfigure, epsfig}
\usepackage{braket}
\usepackage{svg}
\usepackage{bm}
\usepackage{enumerate}
\usepackage{wrapfig}
\usepackage{algpseudocode}
\usepackage{diagbox}
\usepackage{physics}
\usepackage{color}
\usepackage{multirow}
\usepackage[marginal]{footmisc}
\usepackage{comment}
\usepackage{makecell}
\usepackage[colorlinks = true]{hyperref}
\usepackage[nameinlink,capitalize]{cleveref}
\usepackage[ruled,vlined]{algorithm2e}

\usepackage{booktabs}

\newcommand{\calO}{\mathcal{O}}
\newcommand{\calH}{\mathcal{H}}

\newtheorem{theorem}{Theorem}
\newtheorem{lemma}{Lemma}
\newtheorem{corollary}{Corollary}
\newtheorem{proposition}{Proposition}
\newtheorem{remark}{Remark}

\newtheorem{problem}{Problem}

\newcommand{\lc}[1]{{\color{orange}\textbf{Chang: #1}}}

\begin{document}

\title{Accelerating Quantum Dense-Output Simulation through Locality}

\begin{abstract}
Estimating time-accumulated observables, or dense outputs, is central to extracting classical information from quantum simulations. Whereas most quantum-simulation algorithms aim to reproduce a final state or a single-time expectation value, many applications require the integrated signal $J_H(O)=\int_0^T \langle \psi(t)|O(t)|\psi(t)\rangle\,dt$.
Previous Hamiltonian-simulation-based dense-output algorithms have query costs that depend on the norm of the full Hamiltonian, which can grow with system size. However, one can naturally optimize the simulation algorithm when the system structure is known. Consequently, we develop a locality-based truncation framework for the dense-output problem on lattice Hamiltonians. By truncating the Hamiltonian and tightening the resulting error with Lieb-Robinson bounds, we show that a local simulation can recover the dense output. We bound the quantum simulation error from spatial truncation and translate it into local-simulation query bounds. Our analysis includes free and interacting systems for both fermions and bosons, as well as open quantum systems governed by Lindbladians. The improvement from our truncation algorithm varies from constant to unbounded as the system size grows. We also give example applications and conduct relevant numerical experiments to support our findings.
\end{abstract}

\author{Songqinghao Yang}
\affiliation{Cavendish Lab., Department of Physics, University of Cambridge}

\author{Chang Liu}
\affiliation{Graduate School of China Academy of Engineering Physics}

\author{Junkai Wang}
\affiliation{Yau Mathematical Sciences Center, Tsinghua University}

\author{Jingxuan Zhang}
\affiliation{Yau Mathematical Sciences Center, Tsinghua University}

\author{Jin-Peng Liu}
\email{liujinpeng@tsinghua.edu.cn}
\affiliation{Yau Mathematical Sciences Center, Tsinghua University}
\affiliation{Institute for Applied Mathematics, Tsinghua University}
\affiliation{Yanqi Lake Beijing Institute of Mathematical Sciences and Applications}

\maketitle
\section{Introduction}
\label{sec:intro}

Simulating quantum many-body systems is widely regarded as one of the primary motivations for building quantum computers~\cite{feynman2018simulating}. Feynman's foundational observation—that a quantum computer can simulate quantum physics with polynomial rather than exponential resources—set in motion decades of algorithmic development. The first explicit quantum algorithm for Hamiltonian simulation was given by Lloyd~\cite{lloyd1996universal} using product formulae, an approach that has since been refined and extended in many directions: Lie-Trotter-Suzuki decompositions~\cite{suzuki1991general,childs2021theory}, qubitization and quantum signal processing~\cite{low2019hamiltonian,low2017optimal},
linear combination of unitaries~\cite{berry2014exponential}, randomized simulation~\cite{childs2019faster,campbell2018random}, and Dyson-series based methods~\cite{berry2015simulating,berry2020time,berry2024quantum}.
Together, these algorithms achieve near-optimal complexity in all
relevant parameters such as the evolution time $T$, the error $\epsilon$, and the spectral norm of the Hamiltonian. On the application side, quantum simulation has been targeted at quantum chemistry~\cite{cao2019quantum,bauer2020quantum,mcArdle2020quantum},
quantum field theory~\cite{jordan2012quantum,preskill2018simulating},
and condensed matter physics~\cite{babbush2018low,childs2022quantum}.

In parallel, extracting classical information from quantum states has attracted intense interest. Rather than fully characterizing a quantum state, which usually requires exponentially many measurements, the aim is to predict specific properties of interest. The classical shadow formalism of Huang \textit{et al.}~\cite{huang2020predicting} provides an elegant solution. By applying random unitaries and measuring in a fixed basis, one constructs a compact classical description of the state from which many observables can be predicted simultaneously. The sample complexity scales with the complexity of the observable class rather than with the Hilbert space dimension. Classical shadows have been extended to handle fermionic systems~\cite{zhao2021fermionic,koh2022classical}, shallow circuits~\cite{bertoni2024shallow,bu2024classical}, and have found applications in quantum chemistry~\cite{hadfield2022measurements},
entanglement verification~\cite{elben2020cross}, and Hamiltonian learning~\cite{huang2022provably}. Optimized protocols exist for estimating spatially local observables~\cite{huang2020predicting,akhtar2023scalable}.

This paper addresses a problem at the intersection of quantum simulation and observable estimation. We consider the dense-output problem, which was recently formalized by Liu and Lin~\cite{liu2024dense}. In many applications, one is interested not in the quantum state at a single final time $T$, but in a time-accumulated observable
\begin{equation}
J_H(O) = \int_0^T \langle \psi(t) | O(t) | \psi(t) \rangle \dd t,
\label{eq:dense_output_intro}
\end{equation}
where the state $\ket{\psi(t)}$ evolves under a time-dependent Hamiltonian $H(t)$, and $O(t)$ is a time-dependent observable.
Such quantities arise ubiquitously: in quantum control~\cite{d2021introduction,brif2010control},
where \eqref{eq:dense_output_intro} is the cost functional of an optimal control problem~\cite{werschnik2007quantum,roloff2009optimal};
in spectroscopic computation, where $O(t)=e^{i\omega t} O$ gives the linear absorption and fluorescence spectra of molecular
aggregates~\cite{ren2018time}; and in fidelity estimation, where $O(t) = \ket{\phi(t)}\bra{\phi(t)}$ and
$J_H(O)$ measures the time-accumulated overlap between the driven state and a desired trajectory~\cite{huang2020predicting,huang2023learning}.

The naive strategy for computing \eqref{eq:dense_output_intro} is to restart the simulation from $t = 0$ at each of $N_t$ discrete time points and measure the observable independently at each step.
Because quantum measurements collapse the state, there is no way to combine information from a single trajectory; one must genuinely restart. This cost grows multiplicatively with the number of time samples and the per-sample simulation cost.
Liu and Lin~\cite{liu2024dense} showed that more sophisticated approaches—using unbiased amplitude estimation, history state and quantum linear ODE solvers, and Carleman linearization—can substantially reduce this overhead for general or low-rank observables.

The algorithms of~\cite{liu2024dense} treat the Hamiltonian as a black box and do not exploit any spatial structure. In practice, however, the Hamiltonians of physical interest are generally geometrically local. For example, in condensed matter systems, each interaction term typically acts on a bounded number of nearby sites on a $D$-dimensional lattice. Recently, Haah \textit{et al.}~\cite{haah2022optimal} showed that a geometrically local Hamiltonian
on $n$ qubits invokes complexity of $\calO(\log n / \epsilon^2)$, to precision $\epsilon$. And local observables are predominantly sensitive to local features of the Hamiltonian~\cite{yu2023robust,anshu2021sample,brandao2021models}: one does not need to learn the full $n$-qubit dynamics to predict a single-site expectation value. Motivated by this observation, we raise a fundamental question:

\begin{center}
\emph{Does the spatial locality allow one to compute the dense-output problem more efficiently?}
\end{center}

The answer to this question is positive. The reason is deeply connected to the Lieb-Robinson bound (LRB)~\cite{lieb1972finite,chen2023speed,kliesch2014lieb,nachtergaele2019quasi}, which shows that information cannot propagate effectively faster than a finite velocity $v_{\rm LR}$~\cite{burrell2007bounds,van2023topological}. Denote by $A(t):= e^{itH} A e^{-itH}$, given operators $A$ and $B$ supported on the regions separated by distance $r$, the classical LRB asserts that the commutator $|[A(t), B]|$ is exponentially small, $r-v_{\rm LR}t$. The effective velocity can be measured in experiments~\cite{richerme2014non,cheneau2012light,jurcevic2014quasiparticle,cheneau2022experimental}. The original LRB was proved for quantum spin systems~\cite {lieb1972finite}.
Subsequent decades have seen the bound refined for power-law
interactions~\cite{tran2020hierarchy,tran2021optimal,defenu2023long,foss2015nearly,kuwahara2020strictly,tran2021lieb,TranGuoSuChildsGorshkov2020LongRange}, fermionic systems~\cite{wang2020tightening,toniolo2024stability},
and Lindbladian dynamics~\cite{poulin2010lieb,barthel2012quasilocality},
and applied to prove fundamental results in quantum information theory, including the area law for
entanglement~\cite{hastings2007area,brandao2015exponential,epstein2017quantum,eisert2006general},
the clustering of correlations in gapped
systems~\cite{hastings2004locality,nachtergaele2006lieb,landau2015polynomial},
and the efficiency of classical simulation
algorithms~\cite{osborne2006efficient,bravyi2006lieb}. A great amount of development was conducted to make the interaction bound tighter~\cite{nachtergaele2006lieb,hamma2009lieb,schuch2011information,roberts2016lieb,NachtergaeleSims2006Clustering,nachtergaele2010lieb,wang2020tightening}. There are also a lot of work dedicated to derive the bound for specific system~\cite{damanik2014new,hinrichs2025lieb,gebert2022lieb,gebert2016polynomial,nachtergaele2007lieb,lemm2023information,woods2016dynamical} or generalize it to thermalization settings~\cite{gogolin2016equilibration,kuwahara2021absence,nachtergaele2011lieb}.

\vspace{5mm}
\noindent \textbf{Main contributions of this paper.}
We exploit the LRB to give the first systematic analysis of spatially truncated dense-output computation for quantum lattice systems. Our contributions are summarized as follows.

\begin{enumerate}

\item \textbf{Truncation error bound for the dense-output problem.}
We consider a $D$-dimensional lattice of free fermion systems, free boson systems, interacting fermionic systems, as well as long-range interacting bosonic systems. We define a spatially truncated dense output by restricting the Hamiltonian to a region around the observable's support to a local regime. We obtain improvement in terms of the query complexity over the black-box dense-output baseline~\cite{else2020improved,chen2019finite} as the system size grows and at a reasonable simulation time. The truncation bounds are given in Theorems~\ref{thm:flc-dense-output}, \ref{thm:free_fermion_main}, \ref{thm:fermion-dense-output}, and~\ref{thm:bosonic-dense-output}.

\item \textbf{Extension to open quantum systems.}
We extend the framework to open quantum systems~\cite{lindblad1976generators,gorini1976completely}. When the Lindbladian possesses a spectral gap $\Delta > 0$, the information propagation speed of the observable in the Heisenberg picture decays exponentially. The decay rate is proportional to the spectral gap, making the truncation scales even better when the gap grows. The corresponding truncation bound is stated in Theorem~\ref{thm:lindblad_main}.

\item \textbf{Numerical verification.}
We numerically investigate the locality mechanism in free and interacting fermionic and bosonic lattice models. For free fermions, we examine how the required local region depends on simulation time and target precision, including the role of temporal cancellation in the dense-output error (Fig.~\ref{fig:num-free-fast}). In the slow-decay regime, we further study the system-size dependence and find that the selected local region varies only weakly with the full system size while the norm of the full Hamiltonian continues to grow, illustrating the resource mechanism underlying our query-complexity improvement (Fig.~\ref{fig:num-free-slow}). We also compare fixed-state and operator-level error criteria, illustrating that a specific dense-output task can require a substantially smaller region than uniform observable control (Fig.~\ref{fig:num-hierarchy}). Finally, we demonstrate local reconstruction in interacting fermionic and bosonic models (Figs.~\ref{fig:num-fermions} and~\ref{fig:num-bosons}), supplemented by finite-size, grid-refinement, and error-criterion checks. Taken together, these results show how the analytically identified locality structure manifests in finite systems and clarify when dense-output estimation can require substantially less spatial information than stronger forms of local reconstruction.

\end{enumerate}

\noindent \textbf{Connections to Hamiltonian learning.} The locality structure exploited throughout this paper suggests a
natural connection to Hamiltonian learning. Rather than reconstructing
the full microscopic Hamiltonian, the dense-output problem only
requires predicting the part of the dynamics visible through a fixed
observable. Writing the dense-output signal in the energy eigenbasis, the decomposed signal is effectively sparse whenever the initial state has support on few energy eigenstates, the Hamiltonian and observable share a conserved symmetry, or the observable is low-rank. Such sparsity would, in principle, allow estimation with only a sublinear number of time samples via optimized sampling techniques like compressed sensing. In other words, one can perform temporal truncation to simplify the estimation. Making this connection rigorous requires establishing the relevant restricted-isometry property for the sampling ensemble~\cite{baraniuk2007compressive}. We view this as a promising bridge between dense-output simulation and Hamiltonian learning, and we leave a full treatment, including the sample- and query-complexity analysis of the associated
sparse-recovery algorithm, to future work.

\vspace{5mm}

\noindent \textbf{Related work.} Our work builds directly on the dense output framework of Liu and Lin~\cite{liu2024dense}, and we complement their complexity analysis with spatial locality bounds. The use of LRB to truncate observables and bound errors in local simulations has a long history, starting with the original paper~\cite{lieb1972finite} and continuing through studies of the classical simulability of local dynamics~\cite{osborne2006efficient,bravyi2006lieb}.
Truncating a Hamiltonian to a neighbourhood and bounding the resulting error appear in the context of tensor-network methods~\cite{vidal2007classical,white2004real}. Similar methods also arise, more recently, in the analysis of quantum simulation
algorithms~\cite{haah2022optimal,yu2023robust}. The Lindbladian extension uses the LRB for open systems established in~\cite{poulin2010lieb,barthel2012quasilocality,bardet2023rapid}
and the spectral-gap methods of~\cite{kastoryano2016quantum,cubitt2015stability}.
The connection to Hamiltonian learning draws on the optimal
algorithms of~\cite{haah2022optimal,yu2023robust} and the shadow process tomography framework of~\cite{kunjummen2023shadow}.

\section{Setup}
\label{sec:background}

We first give a self-contained account of the dense output problem,
following~\cite{liu2024dense} but adapted to the lattice setting that is
the focus of this paper.

Let $\calH$ be the Hilbert space of a quantum system and let
$H(t)$ be a time-dependent Hermitian Hamiltonian.
The Schr\"{o}dinger’s equation
\begin{equation}
i\frac{\dd}{\dd t}\ket{\psi(t)} = H(t)\ket{\psi(t)},
\qquad \ket{\psi(0)} = \ket{\psi_{0}},
\label{eq:schrodinger}
\end{equation}
has the formal solution $\ket{\psi(t)} = U(t)\ket{\psi_{0}}$, where
\begin{equation}
U(t) = \mathcal{T}\exp\left(-i\int_0^t H(s)\dd s\right)
\label{eq:propagator}
\end{equation}
is the time-ordered propagator and $\mathcal{T}$ denotes the time-ordering operator.

\begin{problem}[Quantum dense output~\cite{liu2024dense}]
\label{prob:dense_output}
Let $O(t)$ be a family of bounded Hermitian observables with $\norm{O(t)} \le 1$, and let $T > 0$. Given access to block encoding of $H(t)$, and a state-preparation oracle for $\ket{\psi_{0}}$, estimate the time-accumulated observable
\begin{equation}
J_H(O) = \int_0^T \langle\psi(t)|O(t)|\psi(t)\rangle\dd t
\label{eq:dense_output}
\end{equation}
to additive precision $\epsilon$ with success probability at least $1-\delta$.
\end{problem}
The quantity $J_H(O)$ is called the dense output by analogy with the
classical numerical-analysis concept~\cite{hairer1990dense,landry2009solving}.
One can ask a classical ODE solver to produce the solution
on a dense grid of time points rather than only at the final time. Similarly, here we wish to extract time-integrated information from the quantum trajectory rather than a single snapshot.

The most na\"{i}ve strategy for estimating $J_H(O)$ is to discretize the
integral with a quadrature rule on $N_t$ nodes ${t_1, \ldots, t_{N_t}}$
and weights ${\omega_1, \ldots, \omega_{N_t}}$,
\begin{equation}
J_H(O) \approx\sum_{k=1}^{N_t} \omega_k
\langle\psi(t_k)|O(t_k)|\psi(t_k)\rangle,
\label{eq:quadrature}
\end{equation}
and then estimate each expectation value $\langle\psi(t_k)|O(t_k)|\psi(t_k)\rangle$ by preparing $\ket{\psi(t_k)}$ and measuring. In \cite{liu2024dense}, to obtain the simulated dense output, Liu and Lin introduced a hierarchy of quantum algorithms that progressively reduce the simulation overhead. In this paper, we use their Theorem 3.3 as our baseline and optimize query complexity with structured Hamiltonians, providing an explicit construction for query access. To this end, we define the parameters and setup we use throughout the paper.

\subsection{The query model}
\label{sec:query-model}

Before moving to our optimization strategy, we emphasize that we will consistently state improvements over the original paper~\cite{liu2024dense} using the query-complexity model standard in the quantum algorithms literature. We assume a quantum algorithm with a block encoding of $H(t)$, \textit{i.e.,} a unitary $U_H$ acting on the system register together with $m_H$ ancilla qubits, such that
\begin{equation}
\label{eq:block-encoding}
\big(\langle 0^{m_H}|\otimes I\big)U_H\big(|0^{m_H}\rangle\otimes I\big)
=\frac{H(t)}{\alpha_H},
\end{equation}
for some normalizing constant $\alpha_H\geq \|H(t)\|$. Consequently, a single `query to the matrix oracle’ means one call to $U_H$ or its
controlled/adjoint variants; the cost of an algorithm is the number of
such calls. For time-dependent $H(t)$, the oracle carries an additional index register that selects among a dense grid of times ${t_\ell}\subset[0, T]$:
\begin{equation}
\label{eq:ham-t}
\big(\langle 0^{m_H}|\otimes I\big)
\mathrm{HAM\text{-}T}
\big(|0^{m_H}\rangle\otimes I\big)=\frac{1}{\alpha_H}
\sum_\ell
|\ell\rangle\langle\ell|\otimes H(t_\ell).
\end{equation}
We also assume the access to a state-preparation oracle $O_{\psi_{0}}$ with $O_{\psi_{0}}\ket{0}=\ket{\psi_{0}}$. A quantum algorithm for the dense-output problem is thus scaled by the number of queries to the matrix oracle for $H$ and the number of queries to $O_{\psi_{0}}$. These scale with $T$, $\delta$ and $\epsilon$, and the Hamiltonian complexity grows with dependence on $\|H\|$ as well, where $\norm{H}:=\sup_{t\in[0,T]}\norm{H(t)}$. All algorithms in \cite{liu2024dense} treat the Hamiltonian $H$ as a black box, characterized only by $\|H\|$ and the block-encoding structure. However, in physically relevant settings, one does not need to access the full Hamiltonian in time; the part inside the information light cone is enough. This intuition is formalized by physical theorems such as LRB, which is of great interest to the community as a tool for establishing rigorous physical bounds~\cite{bravyi2010topological, bravyi2006lieb,
bachmann2012automorphic}. We make this intuition precise at the level of the dense-output functional itself. In this work, spatial locality lets us improve the $\|H\|$-dependence. Since the state-preparation oracle has a relatively trivial dependence on the norm, when we refer to the query complexity, we will just assume that it is the number of queries made to the Hamiltonian. For the rest of this work, we will use Theorem 3.3 in \cite{liu2024dense} as our baseline for the query number,
\begin{equation}
\text{query}_{\rm base}\bigl(\norm{H},T,\epsilon\bigr) = \calO\left(\frac{\norm{H}T^3\log(1/\delta)}{\epsilon}\right).
\label{eq:liu_cost}
\end{equation}
Here, we implement the dense-output algorithm via a Hamiltonian-simulation protocol, which assumes the standard optimal dependence for Hamiltonian simulation. That is, the complexity scales linearly in time, and we apply amplitude amplification to get the $1/\epsilon$ scaling. We can also turn to the Hadamard-based protocol in \cite{liu2024dense}, which achieves $1/\epsilon^2$ for early fault-tolerant architectures. On the other hand, we can further improve the time scaling quadratically by using the history state. However, this requires a more subtle construction of access to the block encoding of the diagonalizable observable query unitary, as well as to the Hamiltonian itself. It is also unclear how to bound the truncation robustly. We leave this as an open problem for future work. Fix a target accuracy $\epsilon>0$ and failure probability $\delta\in(0,1)$. The main question, making analogy to Problem~\ref{prob:dense_output}, we ask is then,
\begin{problem}[Truncated quantum dense output]
\label{prob:truncated_dense_output}
Let $O(t)$ be a family of bounded Hermitian observables with $\norm{O(t)} \le 1$, and let $T > 0$. Define the complete Hamiltonian to be $H(t) = \sum_{i\in \Lambda} h_i$, given access to block encoding of $H_R(t) = \sum_{i\in X_R} h_i$, where $X_R \subset\Lambda$, and a state-preparation oracle for the initial state $\ket{\psi(0)}$, estimate the time-accumulated observable
\begin{equation}
J_{H_R}(O) = \int_0^T \langle\psi_R(t)|O(t)|\psi_R(t)\rangle\dd t
\label{eq:trun_dense_output}
\end{equation}
to additive precision $\epsilon$ with success probability at least $1-\delta$, where $\ket{\psi_R(t)}=e^{-iH_R}\ket{\psi(0)}$.
\end{problem}

\section{Results}
\label{sec:results}
In this section, we present our main theoretical results. We will first study the dense-output problem for free systems in section~\ref{subsec:free_sys}, for both fermions and bosons. Then we will move on to the interacting fermionic system in section~\ref{subsec:int_ferm}, followed by the long-range interacting bosonic system in section~\ref{subsec:int_boson}. Finally, we will briefly touch on the possible extension to open quantum systems in section~\ref{subsec:open}.
We leave the relevant application of our algorithm in the appendix~\ref{app:application}.

\subsection{Free systems}
\label{subsec:free_sys}

We work on a lattice $\Lambda$ with $N=|\Lambda|$ sites, distance metric $d(\cdot,\cdot)$, and spatial dimension $D$. We first state the assumptions of this section explicitly, since they differ in an important way from the interacting-system assumptions used elsewhere.

\begin{enumerate}
\item[(A1)] \textbf{Quadratic (non-interacting) Hamiltonian.} The dynamics are generated by
\begin{equation}
\label{eq:flc-ham}
H(t) = \sum_{i,j\in\Lambda} h_{ij}(t) c_i^\dagger c_j ,
\end{equation}
where $h(t)$ is Hermitian for every $t$, and $c_i,c_i^\dagger$ are either fermionic operators obeying the canonical anticommutation
relations $\{c_i,c_j^\dagger\}=\delta_{ij}$, or bosonic operators obeying the canonical commutation relations $[c_i,c_j^\dagger]=\delta_{ij}$. All now dynamics is controlled by the particle propagator $G(t)$, solving $\dot G(t)=-iH(t)G(t)$ with $G(0)=I$, via
\begin{equation}
c_i(t) = U(t)^\dagger c_i U(t) = \sum_{j\in\Lambda} G_{ij}(t) c_j .
\end{equation}
Also, because $H(t)$ is a sum of
number-conserving bilinears, $[H(t),\mathcal N]=0$ for the total number operator $\mathcal N=\sum_i c_i^\dagger c_i$, so the full propagator preserves every $n$-particle sector $\mathcal H_n\subset\mathcal H$ and, in particular, every truncated subspace $\mathcal H_{\le n}:=\bigoplus_{m=0}^n\mathcal H_m$. We denote $\mathcal H$ as the space that the Hamiltonian lives in and $\mathcal H_{\le n}$ as the truncated ones. Then the initial state satisfies $|\psi_{0}\rangle\in\mathcal H_{\le n}$ for some fixed $n\ge 1$, independent of $N$.

\item[(A2)] \textbf{Power-law decay.} There is a constant $J_0>0$ and exponent $\alpha>D+1$ such that
\begin{equation}
\label{eq:flc-decay}
|h_{ij}(t)| \le J_0 d(i,j)^{-\alpha}, \quad t\in[0,T],\ i\neq j.
\end{equation}

\item[(A3)] \textbf{Local, single-body observable.} $O$ is the single-site number operator $O=c_q^\dagger c_q$, for some $q \in \Lambda$.
\end{enumerate}
Given a truncation radius $R>0$, let $X_R={i\in\Lambda: d(i,X_0)\le R}$ and let $H_R(t)$ denote the restriction of $H(t)$ to $X_R$ (\textit{i.e.,} drop every term $h_{ij}(t)c_i^\dagger c_j$
with $i\notin X_R$ or $j\notin X_R$), with propagator $U_R(t)$ and Heisenberg-evolved operators $c_i(t)$ replaced by $\tilde c_i(t):=U_R(t)^\dagger c_i U_R(t)$. We set out to minimize the difference between the pair of dense outputs \eqref{eq:dense_output} and \eqref{eq:trun_dense_output} for a given initial state $\ket{\psi_0}$. Specifically, given $T, \epsilon >0$, we seek the radius $R_\epsilon(T)$ needed to make $|J_H(O)-J_{H_R}(O)|<\epsilon$. We also quantify the resulting complexity gain compared to the baseline algorithm~\cite{liu2024dense}. Here, we utilize the LRB derived in~\cite{tran2020hierarchy},

\begin{corollary}[Local simulation of a free particle, Cor.10 \cite{tran2020hierarchy}]
\label{cor:flc-trunc}
Under (A1)-(A3), for any $\eta>0$ there exist constants $0<C,C_0<\infty$ such that, for all times
\begin{equation}
\label{eq:flc-validity}
t < \frac{C_0}{n} R^{\min\left(1,(\alpha-D-\eta)/3\right)},
\end{equation}
the single-site creation operator obeys, on $\mathcal H_{\le n}$,
\begin{equation}
\label{eq:flc-cor10}
\big| c_q^\dagger(t) - \tilde c_q^\dagger(t) \big|
\le C n^{3/2} \left( \frac{t}{R^{\alpha-d}} + \frac{t^{3/2}}{R^{(\alpha-d-\eta)/2}} \right).
\end{equation}
\end{corollary}

\begin{theorem}[Free-light-cone dense-output truncation]
\label{thm:flc-dense-output}
Under (A1)-(A3), let $O=c_q^\dagger c_q$ and let $\eta>0$. There exist constants $0<C,C_0<\infty$ such that, for all $R$ satisfying
\begin{equation}
\label{eq:flc-Tvalid}
T < \frac{C_0}{n} R^{\min\left(1,(\alpha-D-\eta)/3\right)},
\end{equation}
\begin{equation}
\label{eq:flc-main-bound}
\big|J_H(O) - J_{H_R}(O)\big|
\le C n^2 \left( \frac{T^2}{R^{\alpha-D}} + \frac{4}{5}, \frac{T^{5/2}}{R^{(\alpha-D-\eta)/2}} \right).
\end{equation}
\end{theorem}

\begin{proof}
Write $O(t)=c_q^\dagger(t)c_q(t)$ and $O_R(t)=\tilde c_q^\dagger(t)\tilde c_q(t)$. Adding and subtracting a cross term gives
\begin{equation}
O(t)-O_R(t) = c_q^\dagger(t)\big(c_q(t)-\tilde c_q(t)\big) + \big(c_q^\dagger(t)-\tilde c_q^\dagger(t)\big)\tilde c_q(t).
\end{equation}
Hence, by the triangle inequality and submultiplicativity of the operator norm on $\mathcal H_{\le n}$,
\begin{equation}
\big|O(t)-O_R(t)\big| \le \big|c_q^\dagger(t)\big|\big|c_q(t)-\tilde c_q(t)\big|
+ \big|c_q^\dagger(t)-\tilde c_q^\dagger(t)\big|\big|\tilde c_q(t)\big| .
\end{equation}
Since $U(t)$ and $U_R(t)$ are unitary and preserve $\mathcal H_{\le n}$, both $c_q(t)$ and $\tilde c_q(t)$ are unitarily conjugates of the annihilation operator restricted to $\mathcal H_{\le n}$. Hence $|c_q(t)|,|\tilde c_q(t)|\le\sqrt{n+1}=O(\sqrt n)$. Moreover,
$|c_q(t)-\tilde c_q(t)| = |c_q^\dagger(t)-\tilde c_q^\dagger(t)|$, applying \eqref{eq:flc-cor10} to both terms gives, for every $t$ satisfying \eqref{eq:flc-validity},
\begin{equation}
\label{eq:flc-Ot-bound}
\big|O(t)-O_R(t)\big|
\le 2\sqrt{n}\cdot C n^{3/2}\left(\frac{t}{R^{\alpha-D}}+\frac{t^{3/2}}{R^{(\alpha-D-\eta)/2}}\right)
= 2Cn^2\left(\frac{t}{R^{\alpha-D}}+\frac{t^{3/2}}{R^{(\alpha-D-\eta)/2}}\right),
\end{equation}
If $R$ satisfies \eqref{eq:flc-Tvalid}, then \eqref{eq:flc-Ot-bound} holds for all $t\in[0,T]$. Since
$|\psi_{0}\rangle\in\mathcal H_{\le n}$ and $|J_H(O)-J_{H_R}(O)|\le\int_0^T|O(t)-O_R(t)|,dt$, integrating \eqref{eq:flc-Ot-bound} gives \eqref{eq:flc-main-bound}, with $C=2C_0$.
\end{proof}

\begin{remark}
For a general finite observable region $X_0$ with $|X_0|=O(1)$ and $O=\sum_{p,q\in X_0}M_{pq}c_p^\dagger c_q$, $|M|\le1$, the same theorem applies up to some multiplicative constant. One needs to sum~\eqref{eq:flc-Ot-bound} over the $O(1)$ pairs $(p,q)$ where $p\in X_0$ and $p\in X_0$.
\end{remark}

To keep the dense-output estimation error down, two competing effects determine the required truncation radius: the light-cone constraint \eqref{eq:flc-Tvalid}, needed for the validity of Theorem~\ref{thm:flc-dense-output}, and the quantitative decay of the right-hand side of \eqref{eq:flc-main-bound}, required to push the truncation error below the target precision $\epsilon$.
\begin{corollary}[Truncation radius]
\label{cor:flc-radius}
Let the assumptions of Theorem~\ref{thm:flc-dense-output} hold. Then, for any $T, \epsilon>0$, we have provided $R\ge R_\epsilon(T)$ to guarantee $|J_H(O)-J_{H_R}(O)|\le \epsilon$ with
\begin{equation}
\label{eq:flc-radius}
R_\epsilon(T) \;=\; \Theta\!\left( \max\left\{\, (nT)^{\max(1,\,3/(\alpha-D-\eta))},\ \ \left(\frac{n^2T^{5/2}}{\epsilon}\right)^{2/(\alpha-D-\eta)} \right\} \right).
\end{equation}
In particular, for fixed $n$ and $\alpha-D>5$ (so the exponent $5/(\alpha-D)<1$), the ballistic term eventually dominates for large $T$, and $R_\epsilon(T)=\Theta(T)$.
\end{corollary}
Since $\alpha-D > (\alpha-D-\eta)/2$ whenever $\alpha-D>-\eta$, the second term of \eqref{eq:flc-main-bound} has the smaller exponent in $R$ and therefore dominates for large $R$; requiring it alone to be at most $\epsilon$ gives
\begin{equation}
R_{\epsilon}(T) = \Theta\left( \left(\frac{n^2 T^{5/2}}{\epsilon}\right)^{2/(\alpha-D-\eta)} \right).
\end{equation}
Requiring \eqref{eq:flc-Tvalid} to hold at $t=T$ gives
\begin{equation}
R_\epsilon(T) = \Theta\left( (nT)^{\max\left(1,\ 3/(\alpha-D-\eta)\right)} \right).
\end{equation}
\begin{remark}
For $\alpha-D-\eta\ge 5$ this is the ballistic linear light cone,
$R_\epsilon(T)=\Theta(nT)$; for $D<\alpha-D-\eta<5$ the exponent exceeds $1$, \textit{i.e.,} the required radius grows faster than linearly in $T$. In the latter case, the light cone is present but its `velocity’ is not uniformly bounded, reflecting the weaker decay.
\end{remark}
We now ask what improvement \eqref{eq:flc-radius} brings over the original query complexity in section~\ref{sec:query-model}. To answer this question, we must first evaluate the operator norm $\|H(t)\|$. It is the ratio between the norm of the full lattice $\|H\|$ and the truncation norm $\|H_R\|$ that gives us the advantage. Analogously, we can also use the block-encoding normalization $\alpha_{H_R}$ and $\alpha_H$ to show more explicitly the complexity difference, see appendix~\ref{app:do-resources}.

\begin{lemma}[Norm bound via shell-sum integral]
\label{lem:flc-shell-sum}
Let $\Lambda$ be a $D$-dimensional lattice and fix $i\in\Lambda$. Under (A2), the number of sites at distance in $[r,r+dr)$ from $i$ scales as $\omega_D\, r^{D-1}\,dr$ for a dimensional constant $\omega_D$, so for any cutoff radius $L\ge 1$,
\begin{equation}
\label{eq:flc-shell-integral}
\|H\| \;:=\; \sum_{\substack{j\in\Lambda\\ 1\le d(i,j)\le L}} |h_{ij}(t)|\approx J_0\,\omega_D \int_1^{L} r^{D-1-\alpha}\,dr .
\end{equation}
Evaluating the integral gives three regimes,
\begin{equation}
\label{eq:flc-three-regime}
\|H\|=
\begin{cases}
\dfrac{J_0\omega_D}{\alpha-D}\Big(1-L^{-(\alpha-D)}\Big) \;=\; S_\infty\Big(1-L^{-(\alpha-D)}\Big), & \alpha>D,\\[2ex]
J_0\omega_D \log L, & \alpha=D,\\[1ex]
\dfrac{J_0\omega_D}{D-\alpha}\Big(L^{D-\alpha}-1\Big) \;=\; \Theta\big(L^{D-\alpha}\big), & \alpha<D,
\end{cases}
\end{equation}
where $S_\infty := J_0\omega_D/(\alpha-D)$ is a finite constant independent of $L$ in the $\alpha>D$ case.
\end{lemma}
Given assumption (A2) we write $L_N:=\mathrm{diam}(\Lambda)=\Theta(N^{1/D})$ for the linear extent of the full $N$-site lattice,
\begin{align}
\|H\| &\approx  S_\infty\Big(1-\Theta\big(N^{-(\alpha-D)/D}\big)\Big), \label{eq:flc-full-norm}\\
\|H_R\| &\approx  S_\infty\Big(1-\Theta\big(R^{-(\alpha-D)}\big)\Big). \label{eq:flc-trunc-norm}
\end{align}
Both bounds are $\Theta(1)$ at leading order and converge to the same constant $S_\infty$ as their respective cutoffs grow. The advantage lies therefore not in the leading order asymptotically, but in how fast each bound approaches $S_\infty$. Suppose we truncate at exactly $\frac{1}{M}$ of the linear extent of the full lattice, $R=L_N/M$. Then one obtains
\begin{equation}
|H_R| =
S_\infty\left(
1-\Theta\!\left(
M^{\alpha-D}N^{-(\alpha-D)/D}
\right)
\right).
\end{equation}
where the rate of convergence for the truncated algorithm is diluted by a factor of $M^{\alpha-D}$. For example, halving the truncation radius multiplies the residual norm by exactly $2^{\alpha-D}$, regardless of $N$. It is not hard to observe that one can make the comparison more naturally by writing the system size parameter $L_N$ in terms of the simulation time $T$ with Corollary~\ref{cor:flc-radius}. Consequently, we can find the simulation time within which the truncation algorithm gives lower query number.

\begin{theorem}[Asymptotic query complexity for free systems]
\label{thm:flc-complexity-implications}
Under (A1)-(A3), let $R_\epsilon(T)$ be the truncation radius of
Corollary~\ref{cor:flc-radius}, and let $L_N=\Theta(N^{1/D})$ be the linear extent of the full lattice. Then:
\begin{enumerate}
\item[(i)] \textbf{Leading order.} $\|H\|=\Theta(1)$ and $|H_R|=\Theta(1)$ in $N$ and $R$ (for $R\gg1$), so substituting either into the generic complexity formula \eqref{eq:liu_cost} gives
\begin{equation}
\text{query}_{\rm base}(N,T,\epsilon) = \text{query}_{\rm trunc}(T,\epsilon) = \mathcal{O}\big(T^3\log(1/\delta)/\epsilon\big).
\end{equation}

\item[(ii)] \textbf{Finite-size reduction.}
The example above shows a finite-size reduction in the normalization when $R<L_N$. In the algorithmic setting, for fixed $T$ and $\epsilon$, $R_\epsilon(T)$ is independent of $N$; hence, as the full system grows, the truncated normalization can remain smaller than the full-system normalization by an $N$-independent factor, yielding a corresponding constant-factor reduction in the query upper bound.

\end{enumerate}
\end{theorem}

One important result is that, in the above setting, the benefit of query complexity over the asymptotic order is only a constant factor. For the complementary regime $\alpha < D+1$, however, one can actually obtain unbounded advantage. We now change the assumption (2) to $\alpha < D+1$. We also restrict ourselves to the free-fermion model. We leave the full derivation to appendix~\ref{app:slow} for interested readers.

\begin{theorem}[Free-fermion dense-output truncation error]
\label{thm:free_fermion_main}
Let $H(t)$ be the free-fermion Hamiltonian~\eqref{eq:flc-ham}
with power-law hopping $|J_{ij}(t)|\le J_0r_{ij}^{-\alpha}$,
assuming a slow decay constant $D >  \alpha > D/2$. Let $O$ be the observable supported on $X_0$.
Then for any $R > 0$ and any initial state $|\psi_0\rangle$,
\begin{equation}
|J_H(O) - J_{H_R)}(O)| \le \frac{T^2}{R^{\alpha-D/2}}
\label{eq:free_fermion_main}
\end{equation}
for total simulation time T. To achieve $|J_H(O) - J_{H_R)}(O)|\le\epsilon$ it suffices to choose the radius in order of $\mathcal{O}\left((\frac{T^2}{\epsilon})^{1/(\alpha-D/2)}\right)$. For fixed $\epsilon$ this grows as $R\sim T^{2/(\alpha-D/2)}$.
\end{theorem}

\begin{proposition}[Query complexity before truncation]
\label{prop:baseline_free}
Applying Theorem 3.3 of~\cite{liu2024dense} directly to the full $N$-site Hamiltonian and using Lemma~\ref{lem:flc-shell-sum} one gets
\begin{equation}
\text{query}_{\rm base}(N,T,\epsilon)=\calO\left(
\frac{N^{1-\frac{\alpha}{D}}T^3\log(1/\delta)}{\epsilon}
\right).
\label{eq:Qbase_free}
\end{equation}
\end{proposition}

\begin{theorem}[Asymptotic query complexity for slow-decay free fermion]
\label{thm:advantage_free}
Fix $D>\alpha>D/2$, and let
\begin{equation}
R_\epsilon(T) := \Theta\bigl((2T^2/\epsilon)^{1/(\alpha-D/2)}\bigr)
\label{eq:R_eps_free_def}
\end{equation}
be the radius of Theorem~\ref{thm:free_fermion_main} guaranteeing
$|J_H(O)-J_{H_R}(O)|\le\epsilon$. Applying Theorem 3.3 of~\cite{liu2024dense} to the truncated Hamiltonian $H_R(t)$, requires
\begin{equation}
\text{query}_{\rm trunc}(T,\epsilon)=\calO\left(
\Bigl(\tfrac{T^2}{\epsilon}\Bigr)^{\frac{1-\frac{\alpha}{D}}{\alpha-D/2}}
\frac{T^3\log(1/\delta)}{\epsilon}\right)
\label{eq:Qtrunc_free}
\end{equation}
queries, with no dependence on $N$. Consequently, for fixed $T,\epsilon$,
\begin{equation}
\frac{\text{query}_{\rm base}(N,T,\epsilon)}{\text{query}_{\rm trunc}(T,\epsilon)}=\Theta\left(\frac{N^{1-\alpha/D}}{ \Bigl(\tfrac{T^2}{\epsilon}\Bigr)^{\frac{1-\frac{\alpha}{D}}{\alpha-D/2}}}\right),
\label{eq:advantage_free}
\end{equation}
which diverges as $N\to\infty$ whenever $D/2<\alpha < D$.
\end{theorem}

\begin{theorem}[Truncated complexity with $\alpha=D$]
\label{thm:advantage_D}
For $\alpha=D$, from Lemma \ref{lem:flc-shell-sum} we get a weaker dependence of the system size $N$ for the baseline query complexity,
\begin{equation}
\frac{\text{query}_{\rm base}(N,T,\epsilon)}{\text{query}_{\rm trunc}(T,\epsilon)}=\Theta\left(\frac{\text{log}N}{\text{log}T}\right).
\end{equation}
which means the advantage is now only suppressed exponentially at a larger power-law decay rate.
\end{theorem}
Theorem~\ref{thm:advantage_free} and Theorem~\ref{thm:advantage_D} identify the advantage of our truncated dense-output framework in the regimes when the decay rate lies in $D/2<\alpha < D+1$. The black-box algorithm of~\cite{liu2024dense}, applied without any locality structure, does not admit an $N$-independent query complexity at fixed target accuracy $\epsilon$: the query count in~\eqref{eq:Qbase_free} diverges as $N\to\infty$ purely because $\norm{H}$ is unbounded. The Lieb-Robinson-truncated algorithm, by contrast, achieves a query complexity that is a fixed function of $(T,\epsilon,\alpha,D)$ which leads to an unbounded advantage in the thermodynamic limit in short times. Note that the advantage of the decay constant being smaller than the dimension is not as strong (linear dependence compared to logarithmic scaling) as the case when $\alpha=D$. This does not mean the light cone is largest at $\alpha=D$. On the contrary, since the decay rate is larger, the light cone should be tighter, leading to less information propagation and further improvements in the query number, as we see in the time dependence. But in this case, the overall bound of the Hamiltonian norm will grow much smaller.

\begin{remark}[On the lower bound]
Naturally, one would ask for the lower bound of the algorithm because it then demonstrates the gap between the simulation efficiency and the practicality of the algorithm. We stress that Liu and Lin~\cite{liu2024dense} themselves already qualitatively made the argument: they invoke the Heisenberg limit for the $\epsilon$-dependence and the intuition of no fast-forwarding for the $T$-dependence, arriving at a multiplicative $\mathcal{O}(T/\epsilon)$. And to get the integration, one typically needs to do discretization where another factor of $T$ comes in. By adding LRB does not help here with the lower bound, since the information propagation typically gives upper bound with the speed of the system. This means that the norm term in the scaling could be assumed to be as low as a constant factor scaling.
\end{remark}

\subsection{Interacting fermions}
\label{subsec:int_ferm}
In this section, we consider the interacting fermions in the following setup:

\begin{itemize}
\item[(B1)] \textbf{Even interaction.} The dynamics is generated by assigning an interaction $\Phi$ to each finite $X\subset\Lambda$. The self-adjoint operator $\Phi(X,t)$ supported on $X$ is even under the fermion-parity automorphism in~\cite{nachtergaele2019quasi}.

\item[(B2)] \textbf{Decay interaction.} There is an $F$-function $F:[0,\infty)\to(0,\infty)$ that is non-increasing, uniformly integrable, and satisfies the convolution condition $C:=\sup_{x,y}\sum_z \frac{F(d(x,z))F(d(z,y))}{F(d(x,y))}<\infty$ such that, with $G(x,y):=C^{-1}F(d(x,y))$ as the kernel function,
\begin{equation}
\label{eq:fermion-decay}
\sum_{\substack{Z\in\{X_0\subset \Lambda:|X_0|<\infty\}\\ x,y\in Z}} \|\Phi(Z,t)\| \;\le\; \|\Phi\|_G(t)\, G(x,y),\quad \forall\, x,y\in\Lambda,\ t\in[0,T].
\end{equation}
We treat the time-independent case explicitly below, writing $\|\Phi\|_G(t)\equiv g$ constant; the time-dependent case only requires replacing $g t$ by $\int_0^t\|\Phi\|_G(r)\,dr$ throughout.

\end{itemize}
Consider the standard choice of the power-law profile $F(r)=(1+r)^{-(D+\eta)}$ ($\eta>0$) as an example. The kernel $G$ is $G(x,y) = \Theta\big(d(x,y)^{-(D+\eta)}\big)$. With $|\Phi|G(t)= g$, one has,
\begin{equation}
\sum_{\substack{x,y \in X_0}} |\Phi(X,t)| \le g, G(x,y) = \Theta\big(g d(x,y)^{-D}\big), \qquad \forall, x,y, \ t\in[0,T].
\end{equation}
The following result is obtained in Theorem 3.1 in \cite{nachtergaele2019quasi},
\begin{theorem}[Lieb-Robinson bound for interacting fermions]
\label{thm:nsy-lrb}
Under (B1)-(B2), let $A,B$ be two operators living on disjoint support. Then for all $0 \le t\le T$,
\begin{equation}
\label{eq:nsy-lrb}
\big\|[A(t),B]\big\| \;\le\; 2\|A\|\|B\|\left(\exp\!\Big(2\!\int_0^t \|\Phi\|_G(r)\,dr\Big)-1\right) \sum_{x\in\partial_\Phi X}\sum_{y\in Y} G(x,y),
\end{equation}
where $\partial_\Phi X=\{x\in X: \exists, Z\in B_\Lambda(X),\ \Phi(Z,\cdot)\not\equiv 0,\ x\in Z\}$ is the $\Phi$-boundary of $X$.
\end{theorem}
Here, $B_\Lambda(X)$ is the collection of all finite subsets $X$ of $\Lambda$. Physically, the first sum in \eqref{eq:nsy-lrb} is over the active boundary sites $B_\Lambda(X)$ of lattice $X$ that are coupled to sites outside $X$. We now use Theorem~\ref{thm:nsy-lrb} together with the localization lemma~\cite{nachtergaele2019quasi}, which converts the commutator norm into the operator norm, to get the dense-output formulation.
\begin{lemma}[Local approximation]
\label{lem:cond-exp}
Let $A$ be an operator for the fermionic algebra and $X\subset\Lambda$. If $\|[A,B]\| \le \epsilon|B|$ for all $B$ on disjoint support of $A$, then there exists an operator $A' \in A_X$ that satisfies $|A-A'|\le \epsilon$.
\end{lemma}
To put Lemma~\ref{lem:cond-exp} in words, for every truncation radius $R>0$, a strictly local operator $O_R(t)$ approximating $O(t)$ can be found.
\begin{theorem}[Fermion dense-output truncation]
\label{thm:fermion-dense-output}
Under (B1)-(B2), define, for each $t\in[0,T]$,
\begin{equation}
\label{eq:fermion-main-bound}
\big|J_H(O)-J_{H_R}(O)\big|
\le \int_0^T \big|O(t)-O_R(t)\big|dt
\le 2\|O\| \Sigma_R(X_0) \int_0^T \Big(e^{2g t}-1\Big)dt ,
\end{equation}
where
\begin{equation}
\label{eq:sigma-R}
\Sigma_R(X_0) := \sum_{x\in\partial_\Phi X_0} \ \sum_{y\notin X_R} G(x,y)
\end{equation}
is a finite, $|\partial_\Phi X_0|\le|X_0|=O(1)$, tail sum that decays monotonically to $0$ as $R\to\infty$.
\end{theorem}

\begin{proof}
Apply Theorem~\ref{thm:nsy-lrb} with $X=X_0$, $A=O$ and for time-independent $\Phi$ for some operator $B$,
\begin{equation}
\big\|[\tau_t(O),B]\big\| \le 2\|O\||B|\big(e^{2g t}-1\big)\Sigma_R(X_0).
\end{equation}
Lemma~\ref{lem:cond-exp} with $\epsilon = 2\|O\|(e^{2g t}-1)\Sigma_R(X_0)$ then gives
$|\tau_t(O)-O_R(t)|\le 2\|O\|(e^{2g t}-1)\Sigma_R(X_0)$ for every fixed $t$. Since $|\psi_{0}\rangle$ is normalized,
$|J_H(O)-J_{H_R}(O)|\le\int_0^T|O(t)-O_R(t)|,dt$, and substituting the pointwise bound and integrating in $t$ gives \eqref{eq:fermion-main-bound}.
\end{proof}
Evaluating $\int_0^T(e^{2g t}-1),dt = \tfrac{1}{2g}(e^{2g T}-1)-T$, the bound \eqref{eq:fermion-main-bound} becomes, for $g T\gg 1$,
\begin{equation}
\label{eq:fermion-error-final}
\big|J_H(O)-J_{H_R}(O)\big| = \mathcal{O}\left(\frac{\|O\|}{g}\Sigma_R(X_0) e^{2g T}\right).
\end{equation}
This exponential-in-$T$ growth is present regardless of the decay class of $\Phi$. What does depend on the decay class is $\Sigma_R(X_0)$, and hence the shape of the required truncation radius $R_\epsilon(T)$. For $F$ to be a power-law decay, a standard shell-sum estimate gives, with fixed $x\in X_0$,
\begin{equation}
\sum_{y: d(y,X_0)>R} G(x,y) = \Theta(\text{poly}(R)),
\end{equation}
hence, $\Sigma_R(X_0) = \Theta(\text{poly}(R))$. Requiring \eqref{eq:fermion-error-final} $\le\epsilon$ then forces
\begin{equation}
\label{eq:fermion-Reps-poly}
R_\epsilon(T) = \Theta\left(e^{T}/{\text{poly}(\epsilon)}\right),
\end{equation}
\textit{i.e.,} exponential radius growth in $T$. However, for $F=e^{-ar}/(1+r)^{D+\eta}$, $\Sigma_R(X_0)=\Theta(e^{-aR})$, one instead gets
\begin{equation}
\label{eq:fermion-Reps-exp}
R_\epsilon(T) = \Theta\left(\frac{2g}{a}T + \frac{1}{a}\log\frac{1}{g\epsilon}\right),
\end{equation}
which is a linear light cone. Therefore, the type of interaction, not the presence of interactions per se, determines whether the required truncation radius grows linearly or exponentially with total simulation time. Many other specific bounds can model the shape of the light cone. For example, in \cite{lemm2025enhanced} the authors discuss the LRB for general interacting systems as well as those with commuting interacting Hamiltonians, and in \cite{guo2020signaling}, one can model the long-range fermionic model with strong decay $\alpha<D$ where we can again obtain the exponential decay in time and polynomial decay in the truncation radius.

Similar to the previous section~\ref{subsec:free_sys}, we can derive the complexity bound for the dense-output query model.
\begin{proposition}[Complexity for interacting fermions]
\label{prop:baseline_interacting}
For fast decay systems $\alpha>D$, the advantage over the baseline query count is a constant factor (Theorem~\ref{thm:flc-complexity-implications}), \textit{i.e.,} depending on the convergence rate. This results from the norm of the Hamiltonian scaling asymptotically as a constant. For slow decay models, the full interacting system requires
\begin{equation}
{\rm query}{\rm base}=\calO\left(\frac{\text{poly}(N)T^3\log(1/\delta)}{\epsilon}\right)
\label{eq:Qbase_interacting}
\end{equation}
queries to estimate the dense output at accuracy $\epsilon$. And one can truncate using Theorem~\ref{thm:fermion-dense-output} to reduce the query number to
\begin{equation}
{\rm query}_{\rm trunc}=\calO\left(
\frac{e^{T}\log(1/\delta)}{\epsilon}
\right).
\label{eq:Qtrunc_interacting}
\end{equation}
Now, we have
\begin{equation}
\frac{{\rm query}_{\rm base}}{{\rm query}_{\rm trunc}} = \calO\left(\frac{\text{poly}(N)}{e^{T}}\right),
\end{equation}
which means that the advantage is only attainable at the time scale of $\log N$.
\end{proposition}
As we can see, for a general choice of the observable, truncation does not provide asymptotic advantage at large evolution time. We demonstrate another layer of optimization in appendix~\ref{app:int_pic} by going into the interaction picture of some cleverly chosen observable. Another way to get a stronger query-number bound than the baseline is to model short-range interactions, where the exponential factor allows polynomial scaling in time, thus suppressing the divergence in the improvement factor.

\subsection{Long-range interacting bosons}
\label{subsec:int_boson}

In this section, we also extend our analysis to the long-range interacting bosonic systems~\cite{cevolani2015protected,cramer2008exact,cao2022interaction,faupin2022lieb,kuwahara2021lieb,kuwahara2024effective,yin2022finite,woods2015simulating}, where our LRB result is based on~\cite{lemm2025quantum}. Previously, the equivalent bound is hard to derive because of the unboundedness in the bosonic system~\cite{eisert2009supersonic}. Specifically, the nice propagation property of light-cone shape comes from both the short-range interaction and the locally bounded Hamiltonian. Quantum boson systems violate the second condition because their Hamiltonians are locally unbounded. But the propagation velocity is proportional to the norm of the local energy; thus the on-site energy and velocity can reach infinity~\cite{kuwahara2024enhanced}.

We now apply the same methodology as in the previous sections to a bosonic, long-range system.  Fix a finite lattice $\Lambda$ and work on the bosonic Fock space. Consider Bose-Hubbard-type Hamiltonians
\begin{equation}
H = \sum_{i,j\in\Lambda}\mathcal J_{ij}a_i^\dagger a_j +\frac12\sum_{i,j\in\Lambda}V_{ij}a_i^\dagger a_j^\dagger a_ja_i,
\end{equation}
where $\mathcal J$ Hermitian, $V$ real-symmetric, both allowed to be long-ranged:
\begin{equation}
|\mathcal J_{ij}|\le C_{\mathcal J,\alpha}|i-j|^{-\alpha},\qquad |V_{ij}|\le C_{V,\alpha}|i-j|^{-\alpha},
\end{equation}
for a fixed power-law exponent $\alpha>0$ and constants $C_{\mathcal J,\alpha},C_{V,\alpha}$.
\begin{enumerate}

\item[(C1) ]\textbf{Density bound.} For some $\lambda>0$,
\begin{equation}
\bra{\psi_0} N_{X_R}^q \ket{\psi_0} \le (\beta r^D)^q\qquad (q=1,2,\ X_R \subset \Lambda,\ R\ge1).
\end{equation}
Denote $\beta$ as the particles per unit volume. The assumption states that nowhere in any ball of any radius $r$ does the initial state carry more particles on average and in its fluctuations than a fixed density  $\beta$ would predict. It rules out states where bosons can pile up arbitrarily on a few sites. For example, this could be a Mott-insulator state at fixed filling with $\beta=1$.

\item[(C2)]\textbf{Particle-free shell.} For $X \subset \Lambda$ and radius $R\ge1$, we have
\begin{equation}
N_{X_{2R}\setminus X}\psi_0 = 0,
\end{equation}
which means that there is an annular vacuum of width $2R$, surrounding the region $X$ we care about, with zero probability of finding a particle there.
\end{enumerate}
\begin{theorem}[Theorem~2.2~\cite{lemm2025quantum}]
\label{thm:LRZ}
Assume $\alpha>3D+1$, set $\beta:=\lfloor\alpha-3D-1\rfloor$, and assume (C1)-(C2). Then for
every $v_{\rm LR}>2\sup_{x\in\Lambda}\sum_{y\in\Lambda}|\mathcal J_{xy}||x-y|$, there is a constant $C$ such that, for all $R\ge\max(2,\mathrm{diam}X)$ and all $0\le t<R/v_{\rm LR}$,
\begin{equation}
\big|\langle\psi_0|O(t)-O_R(t)|\psi_0\rangle\big| \le C\|O\|tR^{-\beta}.
\label{eq:LRZ-pointwise}
\end{equation}
\end{theorem}
Now we can integrate and use the theorem above for our dense-output problem.
\begin{theorem}[Dense-output truncation for long-range Bose-Hubbard systems]
\label{thm:bosonic-dense-output}
Under the setup $\|O\| \le 1$, for $R\ge v_{\rm LR}T$,
\begin{equation}
\big|J_{H}(O)-J_{H_R}(O)\big| \le \calO\big(\frac{T^2}{R^{\beta}}\big).
\label{eq:bosonic-dense-bound}
\end{equation}
To achieve $|J_{H}(O)-J_{H_R}(O)|\le\epsilon$ it suffices to take
\begin{equation}
R_\epsilon(T) = \Theta\left(v_{\rm LR}T + \left(\frac{T^2}{\epsilon}\right)^{1/\beta}\right).
\end{equation}
\end{theorem}

We can also formalize the dense-output problem scaling for near-term bosonic hardware, as in~\cite{kuwahara2024effective}. Unlike the long-range Bose model, they consider a finite-range Bose-Hubbard model in which the boson-boson interaction couples only sites within a fixed distance of one another. For a $D$-dimension system, \cite{kuwahara2024effective} shows the scaling, $R_\epsilon(T)=\Theta\bigl(T^D\mathrm{polylog}(T)\bigr)$ at fixed target accuracy $\epsilon$ for the truncation radius. Hence, the number of sites and gates that must be simulated grows polynomially in $T$, with an exponent that increases without bound as $D$ increases. Based on similar methods as in the previous section, it is not hard to get the time dependence of the query complexity for the dense-output problem as $\mathcal O\bigl(\mathrm{poly}(T)\mathrm{polylog}(T/\epsilon)/\epsilon\bigr)$. For details, see~\cite {kuwahara2024effective}.

\subsection{Open quantum systems}
\label{subsec:open}

The analysis of the preceding sections assumes the system is closed. In practice, lattice systems could couple to the environment, such as photonic reservoirs or phonon baths. The corresponding framework is open quantum dynamics governed by a Lindbladian master equation. In this section, we extend the dense-output problem and its spatial truncation analysis to this setting.

One of the most studied ways to model local Markovian dynamics in open quantum systems is to use the Lindbladian equations. Let $\rho(t)$ be the density matrix of the lattice system on $\Lambda$. Under the Born-Markov approximation~\cite{breuer2002theory, lindblad1976generators,gorini1976completely}, $\rho(t)$ obeys the
Gorini-Kossakowski-Sudarshan-Lindblad (GKSL) master equation
\begin{equation}
\frac{\dd\rho}{\dd t}
= \mathcal{L}(t)[\rho]:=
-i[H(t),\rho]+
\sum_a
\left(
L_a(t)\rho L_a(t)^\dagger
- \frac{1}{2}\{L_a(t)^\dagger L_a(t)\rho\}
\right),
\label{eq:GKSL}
\end{equation}
where $H(t)$ is the system Hamiltonian, $\{L_a(t)\}$ are jump operators, and $\mathcal{L}(t)$ is the Lindbladian superoperator. We work with time-independent generators throughout. We denote by $\mathcal{V}(t)$ the forward propagator satisfying $\rho(t) = \mathcal{V}(t)[\rho_0]$. By the GKSL theorem, $\mathcal{V}(t)$ is a completely positive trace-preserving (CPTP) map for all $t\ge 0$~\cite{lindblad1976generators}. It is convenient to view $\mathcal{B}(\mathcal{H})$ as a vector space and adopt Dirac notation for superoperators following~\cite{poulin2010lieb}. Write $|O\rangle\!\rangle$ for the operator $O$ viewed as a vector, with the Hilbert-Schmidt inner product $\langle\!\langle O | O'\rangle\!\rangle = \mathrm{Tr}[O^\dagger O']$. In this notation the Heisenberg-picture equation
$\dot{O} = \mathcal{L}^\dagger[O]$ becomes
$|\dot{O}\rangle\rangle = \mathcal{L}^\dagger |O\rangle\rangle$,
with formal solution
\begin{equation}
|O(t)\rangle\rangle
= e^{\mathcal{L}^\dagger t}|O(0)\rangle\rangle,
\end{equation}
where $\mathcal{L}^\dagger$ is the adjoint Lindbladian, defined by
duality $\mathrm{Tr}[A,\mathcal{L}[\rho]] = \mathrm{Tr}[\mathcal{L}^\dagger[A],\rho]$,
and acts on observables as
\begin{equation}
\mathcal{L}^\dagger[O]=
i[H,O]+
\sum_a
\left(
L_a^\dagger OL_a
- \frac{1}{2}\{L_a^\dagger L_a,O\}
\right).
\label{eq:adjoint_lindblad}
\end{equation}
We shall freely switch between superoperator and operator notation.
The norm on $|O\rangle\rangle$ is the operator norm $\||O\rangle\rangle \| = \|O\|$, and the induced superoperator
norm is $\|\mathcal{L}^\dagger\| = \max_O |\mathcal{L}^\dagger[O]|/\|O\|$. We assume throughout that both $H(t)$ and the jump operators are spatially local. Write $\mathcal{L} = \sum_{X \subset \Lambda} \mathcal{L}_X$, where each $\mathcal{L}_X$ has the Lindblad form~\eqref{eq:GKSL} in local space $H_X$ with $L_{X,a} \in \mathcal{B}(\mathcal{H}_X)$, and $\mathcal{L}_X = 0$ whenever the support of the jump operator is greater than some fixed interaction range $d^* > 0$. We normalize so that $\|\mathcal{L}_X\| \le 1$ for all $X$. We define the interaction strength constant
\begin{equation}
\lambda
:=
\max_{x\in\Lambda}
\sum_{x \in X} \|\mathcal{L}_X\|
< \infty.
\label{eq:lambda_def}
\end{equation}
Following \cite{poulin2010lieb}, we now extend the LRB to the Lindbladian setting,
\begin{lemma}[Theorem in \cite{poulin2010lieb}]
\label{lem:LRB_poulin}
Let $\mathcal{L} = \sum_{X}\mathcal{L}_X$ be a local Lindbladian
satisfying the locality and normalization assumptions, with interaction strength $\lambda$ as in~\eqref{eq:lambda_def}.
Let $O_A \in \mathcal{B}(\mathcal{H}A)$ and $O_B \in \mathcal{B}(\mathcal{H}B)$ be operators on disjoint
regions $A, B \subset \Lambda$ with $d(A,B) > 0$,
and let $O_B(t) = e^{\mathcal{L}^\dagger t}[O_B]$ denote the
Heisenberg-picture evolution of $O_B$. Then
\begin{equation}
|[O_B(t),O_A]|
\le
C|O_A||O_B|
\exp\left(-(d_{AB} - v_{\rm LR}t)\right),
\label{eq:LRB_lindblad}
\end{equation}
where the constants $C, v_{\rm LR}$ are system-dependent.
\end{lemma}

Let $O$ be an observable supported on $X_0 \subset \Lambda$
with $\|O\| \le 1$, and let $\rho_0$ be an initial density matrix. Note that if the observable is time-dependent, we assume $\|O(t)\| \le 1$. The Lindbladian dense output over $[0,T]$ is
\begin{equation}
J_H(O)=\int_0^T \mathrm{Tr}[O,\rho(t)],\dd t=\int_0^T \langle\langle\rho_0 | O(t)\rangle\rangle\dd t,
\label{eq:lindblad_dense_output}
\end{equation}
where $\rho(t) = e^{\mathcal{L}t}[\rho_0]$ and $O(t) = e^{\mathcal{L}^\dagger t}[O]$. We truncate the Lindbladian to the ball $X_R$ by retaining only those local terms $\mathcal{L}_X$ whose support lies entirely within $X_R$, $ \mathcal{L} _ R=\sum _ {X \subset X_R} \mathcal{L}_X$. This preserves the Lindblad structure: $e^{\mathcal{L}R t}$ is a CPTP map on $\mathcal{B}(\mathcal{H}_{X_R})$ for all $t \ge 0$. Terms outside the boundary $\partial X_R$ are discarded. Define the truncated Heisenberg observable $O_R(t) := e^{\mathcal{L}_R^\dagger t}[O]$ and the truncated dense output
\begin{equation}
J_{H_R}(O):=
\int_0^T \langle\langle\rho_0 | O_R(t)\rangle\rangle,\dd t.
\end{equation}
Both $O(t)$ and $O_R(t)$ satisfy first-order equations driven by
$\mathcal{L}^\dagger$ and $\mathcal{L}_R^\dagger$ respectively.
The variation formula gives
\begin{equation}
O(t) - O_R(t)
=
\int_0^t
e^{\mathcal{L}_R^\dagger(t-s)}
\left[
\left(\mathcal{L}^\dagger - \mathcal{L}_R^\dagger\right)
[O(s)]
\right]
\dd s,
\label{eq:duhamel}
\end{equation}
where the reference propagator $e^{\mathcal{L}_R^\dagger(t-s)}$ acts
on the full observable $O(s)$. The difference $\delta\mathcal{L}^\dagger := \mathcal{L}^\dagger-\mathcal{L}_R^\dagger = \sum_{X \not\subset X_R}\mathcal{L}_X^\dagger$ consists of all local Lindblad terms with support not entirely in $X_R$. Because $O$ is supported on $X_0 \subset X_R$, only those terms $\mathcal{L}_X^\dagger$ whose support is sufficiently close to $X_0$ can produce a small contribution to $\delta\mathcal{L}^\dagger[O(s)]$. Writing out the adjoint Lindblad action~\eqref{eq:adjoint_lindblad} for a single term $\mathcal{L}_X^\dagger$, with Hamiltonian part $H_X$ and jump operators ${L_{X,a}}$, gives
\begin{equation}
    \mathcal{L}_X^\dagger[O(s)]
    \;=\;
    i[H_X,\,O(s)]
    \;+\;
    \sum_a\!\left(
        L_{X,a}^\dagger\,O(s)\,L_{X,a}
        - \tfrac{1}{2}\{L_{X,a}^\dagger L_{X,a},\,O(s)\}
    \right).
    \label{eq:LX_action}
\end{equation}
The sandwich term $L_{X,a}^\dagger O(s) L_{X,a}$ and the anticommutator $\{L_{X,a}^\dagger L_{X,a}, O(s)\}$ can both be
expressed in terms of commutators with the jump operators.
Specifically, a direct algebraic identity gives
\begin{equation}
    L_{X,a}^\dagger O(s) L_{X,a}
    - \tfrac{1}{2}\{L_{X,a}^\dagger L_{X,a},\,O(s)\}
    \;=\;
    \tfrac{1}{2}L_{X,a}^\dagger[O(s),\,L_{X,a}]
    \;+\;
    \tfrac{1}{2}[L_{X,a}^\dagger,\,O(s)]L_{X,a}.
    \label{eq:dissipator_commutator_identity}
\end{equation}
Hence the dissipator part of $\mathcal{L}_X^\dagger[O(s)]$ is bounded by
\begin{equation}
    \left\|L_{X,a}^\dagger O(s) L_{X,a}
    - \tfrac{1}{2}\{L_{X,a}^\dagger L_{X,a},\,O(s)\}\right\|
    \;\le\;
    \left\|L_{X,a}\right\| \|[O(s),\,L_{X,a}]\|+
    \left\|L_{X,a}\right\| \|[L_{X,a}^\dagger,\,O(s)]\|.
    \label{eq:dissipator_commutator_bound}
\end{equation}
Now we can directly evaluate the error bound by the LRB of Lemma~\ref{lem:LRB_poulin}. For the Heisenberg evolution $O(s) = e^{\mathcal{L}^\dagger s}[O]$ of an observable initially on $X_0$, pick an operator $Q_X$ on region $X$ with $d(X_0, X) \ge R - d^*$,
\begin{equation}
    \|[O(s),\,Q_X]\|
    \;\le\;
    C\|O\|\,\|Q_X\|\,
    \exp\!\left(-((R - d^*) - v_{\rm LR}\,s)\right),
    \label{eq:LRB_applied}
\end{equation}
for $s \ge 0$, where the bound is exponentially small whenever
$v_{\rm LR}\,s < R - d^*$. Inserting~\eqref{eq:LRB_applied} into~\eqref{eq:dissipator_commutator_bound} and~\eqref{eq:LX_action}, and using $\|\mathcal{L}_X\| \le 1$
to absorb the operator norms of $H_X$ and $L_{X,a}$:
\begin{equation}
    \|\mathcal{L}_X^\dagger[O(s)]\|
    \;\le\;
    C'\|\mathcal{L}_X\|\,\|O\|\,
    \exp\!\left(-((R - d^*) - v_{\rm LR}\,s)\right).
    \label{eq:LX_LRB_bound}
\end{equation}
Summing over all boundary terms $\lambda_\partial := \sum_{X \not\subset X_R} \|\mathcal{L}_X\|$ and inserting into~\eqref{eq:duhamel}:
\begin{align}
    \|O(t) - O_R(t)\|
    &\;\le\;
    C'\lambda_\partial\,\|O\|
    \int_0^{\min(t,\,t^*)}
    \exp\!\left(-((R - d^*) - v_{\rm LR}\,s)\right)
    \dd s
    \label{eq:pointwise_bound}
\end{align}
where $t^* := (R - d^*)/v_{\rm LR}$ is the time at which the
Lieb-Robinson cone reaches the truncation boundary. Here the truncated propagator $e^{\mathcal{L}_R t}$ is unital CP hence contractive in operator norm: $\|e^{\mathcal{L}_R t}[\sigma]\|_1 \le \|\sigma\|_1$
for all $\sigma$. Simple integration gives us
\begin{equation}
    \|O(t) - O_R(t)\|
    \;\le\;
    \frac{C'\lambda_\partial\,\|O\|}{v_{\rm LR}}\,
    \exp\!\left(-((R - d^*) - v_{\rm LR}\,t)\right)
    \;\cdot\;
    \left[1 - e^{-v_{\rm LR}\,t}\right].
    \label{eq:pointwise_LR}
\end{equation}
Bounding $[1 - e^{-v_{\rm LR}t}] \le 1$,
\begin{equation}
    \|O(t) - O_R(t)\|
    \;\le\;
    \frac{C'\lambda_\partial\,\|O\|}{v_{\rm LR}}\,
    e^{-(R-d^*)}\,e^{v_{\rm LR}\,t},
    \qquad t \le t^*.
    \label{eq:pointwise_LR_simple}
\end{equation}
Now, integrating~\eqref{eq:pointwise_LR_simple} over $t \in [0,T]$,
assuming $T \le t^* = (R-d^*)/v_{\rm LR}$,
\begin{equation}
    |J_H(O) - J_{H_R}(O)|
    \;\le\;
    \frac{C'\lambda_\partial\,\|O\|}{v_{\rm LR}}\,
    e^{-(R-d^*)}
    \int_0^T e^{v_{\rm LR}\,t}\,\dd t
    \;=\;
    \frac{C'\lambda_\partial\,\|O\|}{v_{\rm LR}^2}\,
    \left[e^{v_{\rm LR}\,T} - 1\right]
    e^{-(R-d^*)}.
    \label{eq:error_no_gap}
\end{equation}
This grows exponentially in $T$. To achieve $|J_H(O) - J_{H_R}(O)| \le \epsilon$ one requires $R \sim v_{\rm LR}\,T + \log(\epsilon^{-1})$, growing linearly in $T$.

A central object in the analysis of open quantum systems is the spectral gap of the Lindbladian. The jump operators model incoherent processes such as spontaneous emission, dephasing and thermalization that continuously drive the system toward a preferred stationary state $\pi$. The eigenvalues $\lambda_j$ of the Lindbladian therefore acquire strictly negative real parts for $j > 0$, reflecting this  irreversible relaxation. The spectral gap
\begin{equation}
    \Delta \;:=\; -\mathrm{Re}(\lambda_1) \;>\; 0
    \label{eq:gap_def}
\end{equation}
is the smallest such decay rate: it measures how quickly the slowest-relaxing perturbation away from $\pi$ is suppressed, and it therefore sets the fundamental timescale of the dissipative dynamics. In more details, one can write the Jordan decomposition of $\mathcal{L}$ as $\mathcal{L} = \mathcal{S}\mathcal{J}\mathcal{S}^{-1}$,
where $\mathcal{J} = \bigoplus_{j\ge 0} \mathcal{J}_{d_j}(\lambda_j)$ is the Jordan normal form and the eigenvalues are ordered $\mathrm{Re}(\lambda_0) \ge \mathrm{Re}(\lambda_1) \ge \cdots$~\cite{terhal2000problem,blume2010information}. Note that trace preservation forces $\mathrm{Re}(\lambda_0) = 0$. In this case the propagator decomposes as
\begin{equation}
    e^{\mathcal{L}t}
    \;=\;
    |I\rangle\!\rangle\langle\!\langle\pi|
    \;+\;
    \mathcal{S}
    \!\left(\bigoplus_{j>0}e^{\lambda_j t}\mathcal{M}_{d_j}\right)
    \!\mathcal{S}^{-1},
    \label{eq:propagator_decomp}
\end{equation}
where $\mathcal{M}_d$ is the $d\times d$ upper-triangular matrix with $1/k!$ on its $k$-th upper diagonal~\cite{poulin2010lieb}.
Estimating the second term using $\mathrm{Re}(\lambda_j) \le -\Delta$ for all $j>0$ yields
\begin{equation}
    \left\|e^{\mathcal{L}t} - |I\rangle\!\rangle\langle\!\langle\pi|\right\|
    \;\le\;
    \|\mathcal{S}\|^2\,e^{-\Delta t + 1}.
    \label{eq:gap_bound_operator}
\end{equation}
Here $\|\mathcal{S}\|^2$ is the conditioning number of
$\mathcal{S}$, normalized so that $\|\mathcal{S}\| = \|\mathcal{S}^{-1}\|$.

We decompose $O = O^{(\pi)} + \overline{O}$ where $O^{(\pi)} := \mathrm{Tr}[O\,\pi]\cdot I$ is the projection onto the stationary component and $\overline{O} := O - O^{(\pi)}$. Since $\mathcal{L}^\dagger[I] = 0$, we have $\delta\mathcal{L}^\dagger[I] = 0$, so only $\overline{O}$ contributes:
\begin{equation}
    \delta\mathcal{L}^\dagger[O(s)]
    \;=\;
    \delta\mathcal{L}^\dagger\!\left[e^{\mathcal{L}^\dagger s}[\overline{O}]\right].
    \label{eq:Obar_only}
\end{equation}
The gap bound~\eqref{eq:gap_bound_operator} applied to $\overline{O}$ gives
\begin{equation}
    \left\|e^{\mathcal{L}^\dagger s}[\overline{O}]\right\|
    \;\le\;
    \left\|e^{\mathcal{L}^\dagger s} - |I\rangle\\rangle\langle\!\langle\pi|\right\|
    \cdot\|\overline{O}\|
    \;\le\;
    \|\mathcal{S}\|^2\,e^{-\Delta s + 1}\,\|\overline{O}\|
    \;\le\;
    2e\,\|\mathcal{S}\|^2\,\|O\|\,e^{-\Delta s}.
    \label{eq:Obar_op_norm_decay}
\end{equation}
Inserting~\eqref{eq:Obar_op_norm_decay} into~\eqref{eq:LRB_applied} 
and then into~\eqref{eq:LX_LRB_bound}:
\begin{equation}
    \left\|\delta\mathcal{L}^\dagger[O(s)]\right\|
    \;\le\;
    C\,\lambda_\partial\,\|\mathcal{S}\|^2\,\|O\|\,
    e^{-\Delta s}\,
    \exp\!\left(-((R-d^*) - v_{\rm LR}\,s)\right),
    \label{eq:delta_L_LRB_gap}
\end{equation}
where $C$ absorbs all constants. One can integrate again to get the dense output error bound,
\begin{align}
    |J_H(O) - J_{H_R}(O)|
    &\;\le\;
    C\,\lambda_\partial\,\|\mathcal{S}\|^2\,\|O\|\,
    e^{-(R-d^*)}
    \int_0^T (T-s)\,
    e^{-\Delta s}\,e^{v_{\rm LR}\,s}
    \,\dd s
    \nonumber\\
    &\;=\;
    C\,\lambda_\partial\,\|\mathcal{S}\|^2\,\|O\|\,
    e^{-(R-d^*)}
    \int_0^T (T-s)\,
    e^{-\mu s}
    \,\dd s,
    \label{eq:double_integral}
\end{align}
where $\mu := \Delta - v_{\rm LR}$. Define
\begin{equation}
    F(\mu, T)
    \;:=\;
    \int_0^T (T-s)\,e^{-\mu s}\,\dd s
    \;=\;
    \begin{cases}
        \dfrac{T}{\mu} - \dfrac{1-e^{-\mu T}}{\mu^2},
        & \mu \ne 0,\\[8pt]
        \dfrac{T^2}{2},
        & \mu = 0.
    \end{cases}
    \label{eq:F_function}
\end{equation}
Note that $F(\mu,T) \to T^2/2$ as $\mu\to 0$, and for $\mu < 0$ the integral grows as $e^{|\mu|T}/(|\mu|^2)$, recovering the no-gap behavior.
\begin{theorem}[Dense output error with Lindbladian gap]
\label{thm:lindblad_main}
Let $\mathcal{L} = \sum_X \mathcal{L}_X$ be a local Lindbladian
with interaction strength $\lambda$ as in~\eqref{eq:lambda_def}
and interaction range $d^*$. Suppose $\mathcal{L}$ has a unique stationary state $\pi$ and spectral gap $\Delta > 0$, with Jordan basis satisfying $\|\mathcal{S}\| = \|\mathcal{S}^{-1}\|$. Let $O$ be an observable on $X_0$ with $\|O\| \le 1$ and let $T > 0$. Then the truncation error satisfies
\begin{equation}
    |J_H(O) - J_{H_R}(O)|
    \;\le\;
    C\,\lambda_\partial\,\|\mathcal{S}\|^2\,
    e^{-(R-d^*)}\cdot F(\mu,T),
    \label{eq:lindblad_main_bound}
\end{equation}
where $C$ is a constant, and $v_{\rm LR}$ are the Lieb-Robinson constants,, $\mu = \Delta - v_{\rm LR}$, and
$F(\mu,T)$ is given by~\eqref{eq:F_function}. The three asymptotic regimes are:
\begin{enumerate}
    \item \textbf{No gap ($\Delta = 0$, $\mu = -v_{\rm LR} < 0$):} $ F(\mu,T) = \frac{1}{v_{\rm LR}^2} \left[e^{v_{\rm LR}T}-1-v_{\rm LR}T\right]$, growing as $e^{v_{\rm LR}T}$.
          The error bound in~\eqref{eq:lindblad_main_bound} grows as $e^{-(R-d^*)}\cdot e^{v_{\rm LR}T}$,
          and to keep this below $\epsilon$ one needs
          $R \sim v_{\rm LR}\,T + \log\epsilon^{-1}$,
          growing linearly in $T$.
    \item \textbf{Weak gap ($0 < \Delta < v_{\rm LR}$,
          $\mu < 0$):} The exponent $\mu < 0$ means the $e^{-\mu s}$ factor in~\eqref{eq:F_function} still grows with $s$, and $F(\mu,T)$ grows as $e^{|\mu|T}$. One can easily show the truncation error therefore obeys
        \begin{equation}
            |J_H(O) - J_{H_R}(O)|=\mathcal{O}\!\left(
                e^{-(R-d^*)}
                e^{(v_{\rm LR}-\Delta)T}\right).
        \end{equation}
    Thus the required truncation radius continues to grow linearly with
    $T$, but with a reduced effective velocity, $R=\mathcal{O}\!\left((v_{\rm LR}-\Delta)T+\log\epsilon^{-1}\right)$.
    \item \textbf{Strong gap ($\Delta > v_{\rm LR}$, $\mu > 0$):} The integrand $e^{-\mu s}$ is now decaying in $s$.
          Explicitly, as $T\to\infty$, $F(\mu,T) \sim T/\mu$ grows at most linearly in $T$.
          To achieve $|J-J_R|\le\epsilon$ one then chooses
          \begin{equation}
              R\sim d^* + \log\!\left(\frac{C\,|X_0|\,\lambda_\partial\,\|\mathcal{S}\|^2\,T}
              {\epsilon\,\mu}
              \right),
              \label{eq:radius_strong_gap}
          \end{equation}
          growing only logarithmically in $T$.
\end{enumerate}
\end{theorem}

\begin{proof}
The finite-time bound \eqref{eq:lindblad_main_bound} follows from
the integral estimate in \eqref{eq:double_integral} together
with some relaxation and Lieb-Robinson estimates summarized in
\eqref{eq:delta_L_LRB_gap}. In particular, the time dependence in the three regimes follow directly from \eqref{eq:F_function}. For $\mu>0$, $1-e^{-\mu T}\geq0$, and hence $F(\mu,T)=\frac{T}{\mu}-\frac{1-e^{-\mu T}}{\mu^2}\leq\frac{T}{\mu}$. This gives \eqref{eq:radius_strong_gap}.
\end{proof}

Under the stated error bound, the sufficient truncation radius grows logarithmically, rather than linearly, with $T$ when $\Delta>v_{\rm LR}$. It is not independent of $T$ for the unnormalized time-integrated observable. Relaxation suppresses transient contributions to the expectation value, but a nonzero stationary expectation continues to contribute to its time integral. In particular, if the full and truncated dynamics have different stationary expectations for $O$, their difference produces an accumulated error proportional to $T$. The spatial truncation must therefore control the stationary contribution as well as the transient dynamics.

\providecommand{\chang}[1]{\lc{#1}}

\section{Numerical experiments}
\label{sec:numerical}

We use finite-system calculations to examine three questions: how the local region required for a dense output depends on precision and time, how it changes with system size, and how local reconstruction behaves in interacting fermionic and bosonic models. Each comparison evolves the full finite Hamiltonian and its restricted counterpart independently, keeping all retained couplings unchanged. We set the lattice spacing and $J_0$ to one.

Writing $f(t)$ and $f_R(t)$ for the full and restricted expectation values, respectively, we measure
\begin{align}
E_{\mathrm{DO}}(T,R)&=\big|J_H(O)-J_{H_R}(O)\big|
 =\left|\int_0^T[f(t)-f_R(t)]\,\dd t\right|,
\label{eq:num_EDO}\\
E_{\mathrm{abs}}(T,R)&=\int_0^T|f(t)-f_R(t)|\,\dd t.
\label{eq:num_Enc}
\end{align}
The second quantity removes cancellation between different times and satisfies $E_{\mathrm{DO}}\le E_{\mathrm{abs}}$. All tolerances below refer to the absolute spatial error, without division by $T$. We denote by $\widehat R_\epsilon(T)$ the smallest tested radius for which the error remains below $\epsilon$ at every larger tested radius; the hat distinguishes this numerical choice from an analytical sufficient radius $R_\epsilon(T)$. Appendix~\ref{app:num_methods} gives the numerical methods, radius-selection rule, and interpolation used for the equal-error curves.

\subsection{Free systems}
\label{subsec:num_free}

We use the free model to isolate the effects of target precision, system size, and the level at which the truncation error is controlled. On an open chain we specialize the quadratic Hamiltonian~\eqref{eq:flc-ham} to
\begin{equation}
h_{ij}=-J_0|i-j|^{-\alpha}\quad(i\ne j),\qquad h_{ii}=0,
\label{eq:num_free_H_v3}
\end{equation}
and measure the central density $O=n_0=c_0^\dagger c_0$. The initial states are
\begin{equation}
|\psi_0^{(1)}\rangle=c_0^\dagger|\mathrm{vac}\rangle,
\qquad
|\psi_0^{(4)}\rangle=c_{-2}^\dagger c_{-1}^\dagger c_0^\dagger c_1^\dagger|\mathrm{vac}\rangle,
\label{eq:num_free_initials_v3}
\end{equation}
where the superscript is the fixed particle number $M$. The retained interval is $X_R=\{-R,\ldots,R\}$, with $R\ge2$ so that it contains every initially occupied site.

\subsubsection{Precision dependence and time cancellation}
\label{subsubsec:num_free_fast_exploration}

Figure~\ref{fig:num-free-fast} shows the single-particle dense-output error for $\alpha=3$ and $N=441$. At $\epsilon=10^{-3}$ and $10^{-5}$, the outermost equal-error boundary is approximately linear over the simulated time window. At $10^{-7}$, curvature and oscillations become more pronounced. The required region therefore depends on both elapsed time and target precision.

The comparison with $E_{\mathrm{abs}}$ shows the effect of time cancellation. At $10^{-3}$, the two boundaries nearly coincide, while at $10^{-7}$ the absolute-error criterion requires a larger region and eventually leaves the scanned range. A small dense-output error therefore need not require equally accurate reconstruction at every time. Reference-size and hopping-decay checks are given in Fig.~\ref{fig:num-app-free-fast}; the late-time $10^{-7}$ crossings are sensitive to the reference size.

\begin{figure}[htbp]
\centering
\includegraphics[width=\linewidth]{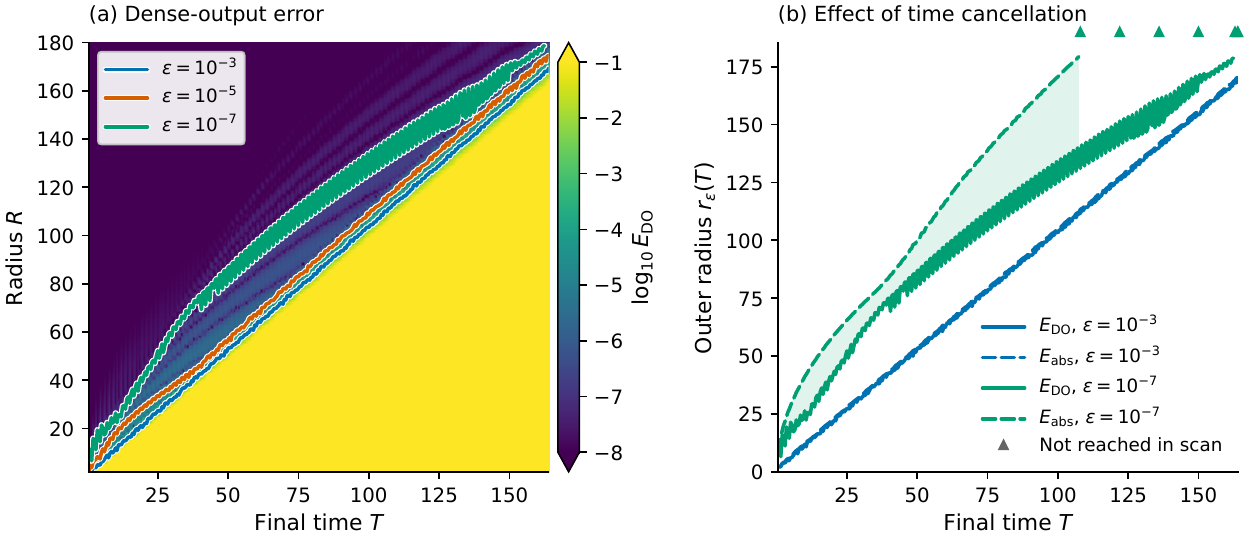}
\caption{Precision-dependent local reconstruction for a central single particle, with $N=441$ and $\alpha=3$. (a) Dense-output error and outermost equal-error boundaries at three tolerances. (b) Boundaries from $E_{\mathrm{DO}}$ and $E_{\mathrm{abs}}$; shading shows their separation at $10^{-7}$ only where both crossings are resolved. Triangles mark thresholds not reached within the scanned radii.}
\label{fig:num-free-fast}
\end{figure}

\subsubsection{Slow decay and system-size dependence}
\label{subsubsec:num_free_slow_exploration}

Figure~\ref{fig:num-free-slow} compares $\alpha=0.75$ and $1$ at fixed $T=4$ and $\epsilon=10^{-3}$ as the full chain grows from $N=513$ to $1025$ and $2049$, without rescaling the hopping. For both exponents, the single-particle task selects about $250$ sites at the two largest sizes. The error profiles also distinguish the integrated target from pointwise reconstruction: $E_{\mathrm{abs}}$ remains above $10^{-3}$ throughout $R\le192$, while $E_{\mathrm{DO}}$ passes this tolerance in the outer part of the scan. The radius-selection rule excludes isolated cancellation minima.

For $\alpha=0.75$, the full single-particle norm grows from about $24$ to $37$, while the selected local norm stays near $19$. At $\alpha=1$, the same separation is weaker but remains visible. These finite-size trends are consistent with the resource mechanism in Theorems~\ref{thm:advantage_free} and~\ref{thm:advantage_D}: the full problem changes with $N$ more strongly than the selected local problem. The plotted norms are finite-size diagnostics, not query counts. The four-particle comparison in Appendix~\ref{app:num_free_supp} shows a similar separation, with some remaining size dependence of the selected region.

\begin{figure}[htbp]
\centering
\includegraphics[width=0.90\linewidth]{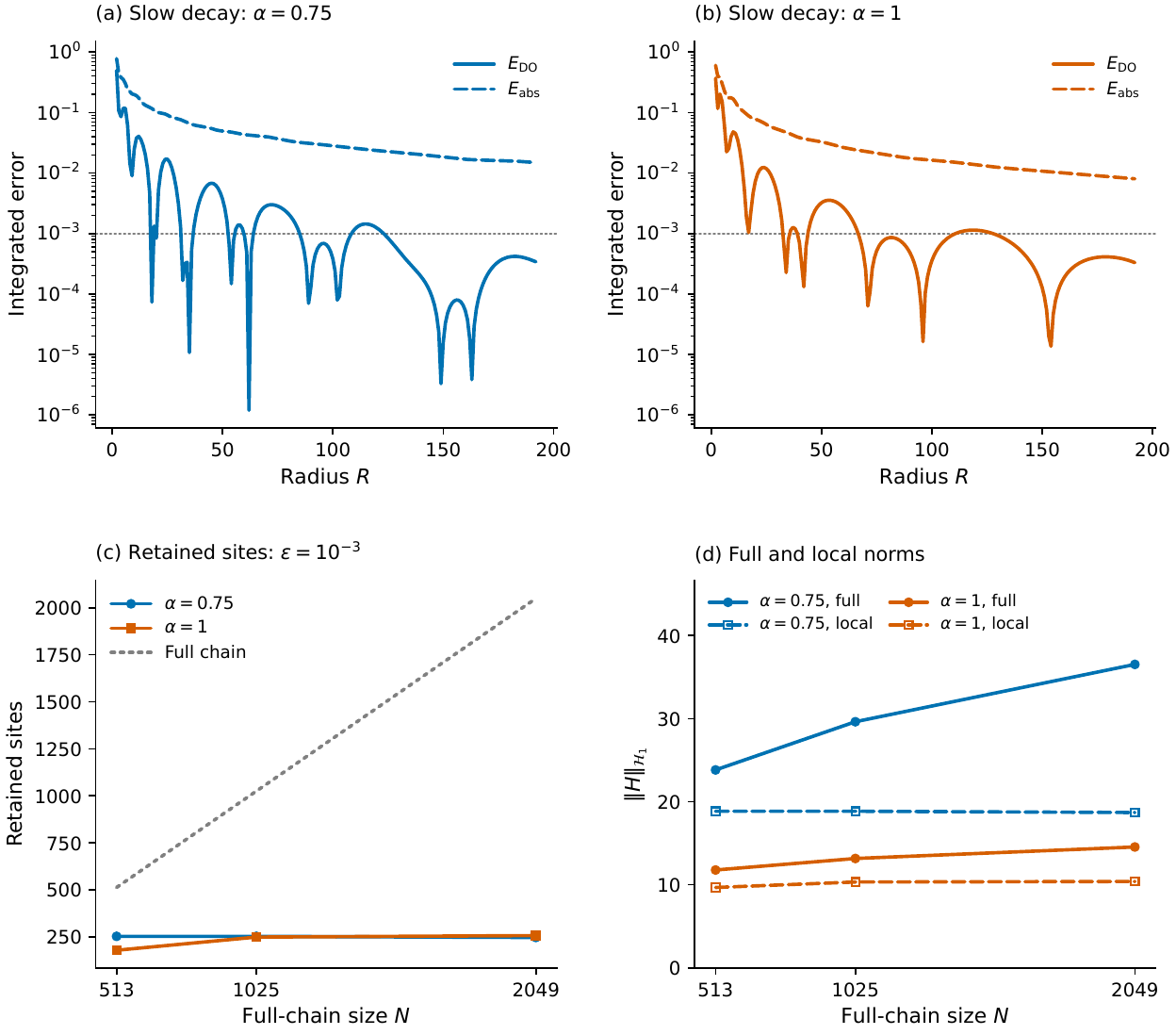}
\caption{Slow-decay single-particle reconstruction at $T=4$. (a,b) $E_{\mathrm{DO}}$ and $E_{\mathrm{abs}}$ for $N=2049$, with the dotted line marking $10^{-3}$. (c) Selected site counts at this tolerance for three full-chain sizes. (d) Spectral norms of the full and selected local Hamiltonians on $\mathcal H_1$.}
\label{fig:num-free-slow}
\end{figure}

\subsubsection{Fixed-state and operator-level requirements}
\label{subsec:num_free_hierarchy}

Local reconstruction also depends on the level at which the error is controlled. Figure~\ref{fig:num-hierarchy} keeps the Hamiltonian, observable, and restriction fixed, using $\alpha=3$, $N=301$, and the four-particle state in Eq.~\eqref{eq:num_free_initials_v3}. Alongside the two fixed-state errors, we compare the norm of the time-integrated operator difference and the time integral of its pointwise norm. These stronger criteria control all initial states in the full-chain $M$-particle sector; Appendix~\ref{app:num_hierarchy} gives their definitions and evaluation.

The fixed-state errors fall several orders of magnitude below the operator-level quantities in the large-radius tail. For four particles, $E_{\mathrm{DO}}$ and $E_{\mathrm{abs}}$ are close, so time cancellation is not the main source of this separation. At $T=12$ and $\epsilon=10^{-3}$, the selected radii are about $15$ for the fixed-state task, $50$ for the integrated-operator norm, and $74$ for the integral of pointwise norms. In this example, restricting the target to the chosen initial state gives the largest reduction, with a further smaller reduction from integrating before taking the operator norm.

\begin{figure}[htbp]
\centering
\includegraphics[width=0.90\linewidth]{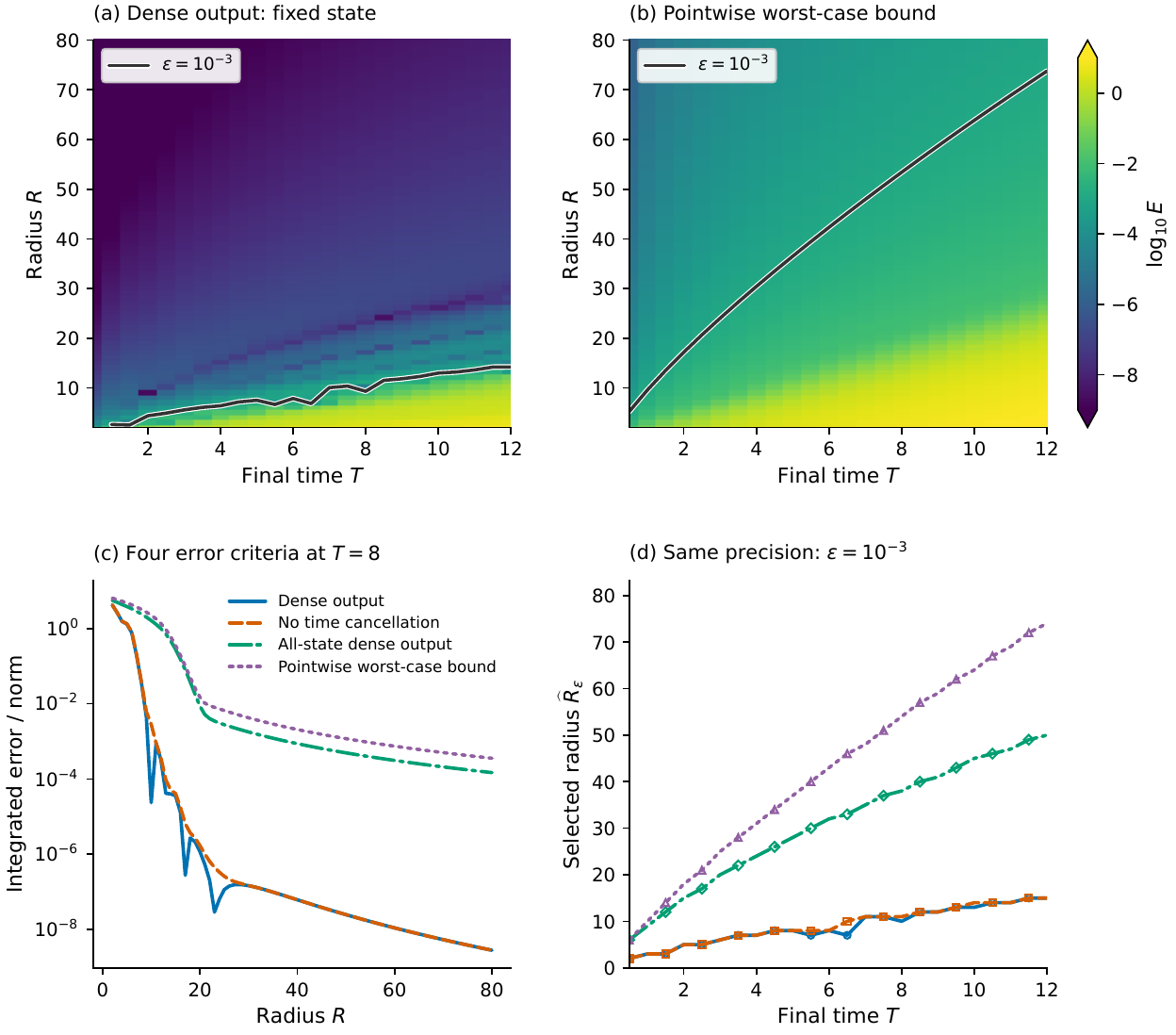}
\caption{Error requirements for the free-fermion model with $N=301$, $\alpha=3$, and $M=4$. (a,b) Fixed-state dense-output error and integrated pointwise operator norm on a common logarithmic scale. (c) All four error criteria at $T=8$. (d) Their selected integer radii at $\epsilon=10^{-3}$.}
\label{fig:num-hierarchy}
\end{figure}

\subsection{Interacting fermions}
\label{subsec:num_interacting}

For the interacting setting of Sec.~\ref{subsec:int_ferm}, we use nearest-neighbour hopping and density interactions,
\begin{equation}
H=-J_0\sum_i(c_i^\dagger c_{i+1}+c_{i+1}^\dagger c_i)
  +\sum_{i<j}W(|i-j|)n_i n_j.
\label{eq:num_tV}
\end{equation}
The initial state is $c_0^\dagger c_1^\dagger|\mathrm{vac}\rangle$ and the observable is $n_0$. We use $W(r)=V_0r^{-3}$ on $N=161$ sites.

Figure~\ref{fig:num-fermions}(a) fixes $V_0=2$ and compares the full dense output with restricted evolutions at $R=12,28,44$. Enlarging the region delays the visible departure from the full result; $R=44$, or $89$ sites, reaches the $10^{-3}$ spatial tolerance at all sampled times through $T=40$. Panel (b) compares the selected radii for $V_0=0,0.5,2$. They remain close over this window, so changing the interaction strength does not substantially change the required region in this example. Appendix~\ref{app:num_interacting_supp} examines the interaction tail at stricter precision.

\begin{figure}[htbp]
\centering
\includegraphics[width=\linewidth]{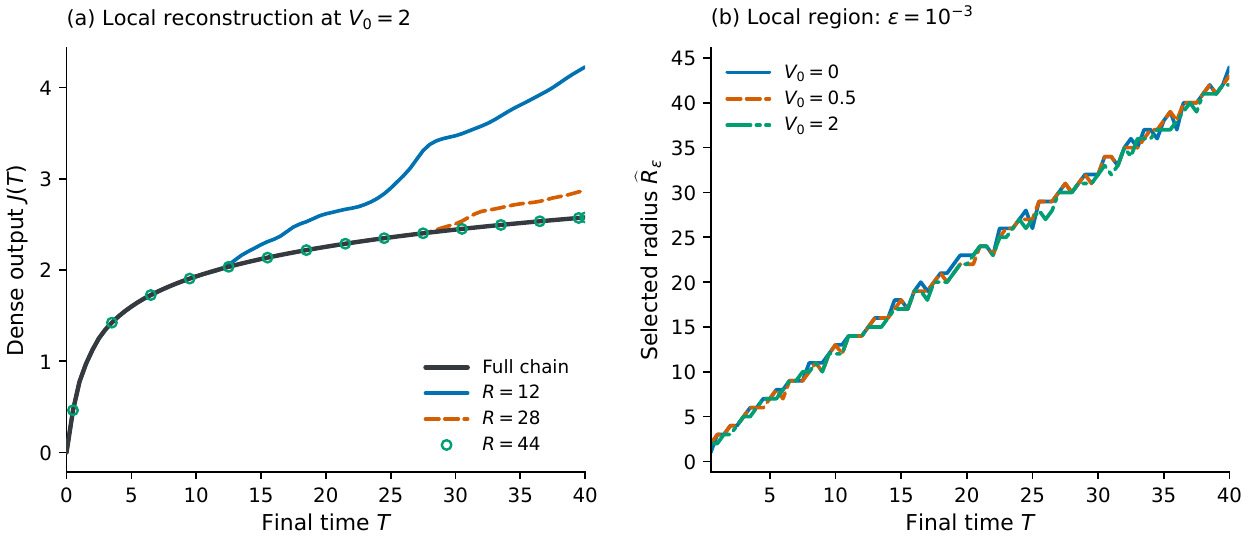}
\caption{Local reconstruction with density interactions for $N=161$ and $W(r)=V_0r^{-3}$. (a) At fixed $V_0=2$, the full-chain dense output is compared with $R=12,28$ (lines) and $R=44$ (open circles). (b) Selected integer radii at $\epsilon=10^{-3}$ for $V_0=0,0.5,2$.}
\label{fig:num-fermions}
\end{figure}

\subsection{Long-range interacting bosons}
\label{subsec:num_bosons}

For the bosonic setting of Sec.~\ref{subsec:int_boson}, we use
\begin{equation}
H=-J_0\sum_{i<j}\frac{a_i^\dagger a_j+a_j^\dagger a_i}{|i-j|^\alpha}
  +\frac U2\sum_i n_i(n_i-1)
  +V_0\sum_{i<j}\frac{n_i n_j}{|i-j|^\alpha}.
\label{eq:num_boson_H_v3}
\end{equation}
We compare a central pair $|2\rangle_0$ with the four-particle cloud $|1,2,1\rangle_{-1,0,1}$, with vacuum elsewhere. The observable $O=|2\rangle_0\langle2|\otimes I$ measures exactly two particles at the centre. Both calculations use the core $X_0=\{-1,0,1\}$, so $|X_R|=2R+3$.

For the configuration $N=25$, $\alpha=6$, $U=4$, and $V_0=1$, the two initial states give different dense outputs (see Fig.~\ref{fig:num-bosons}(a)). The pair integral levels off over the later part of the window, while the cloud integral continues to increase and overtakes it. Both are reproduced at the sampled times using the same fixed radius $R=6$, corresponding to $15$ sites. At $T=6$, the individually selected regions contain $13$ and $15$ sites at $10^{-3}$. Higher-occupation and decay checks are given in Appendix~\ref{app:num_boson_supp}.

\begin{figure}[htbp]
\centering
\includegraphics[width=\linewidth]{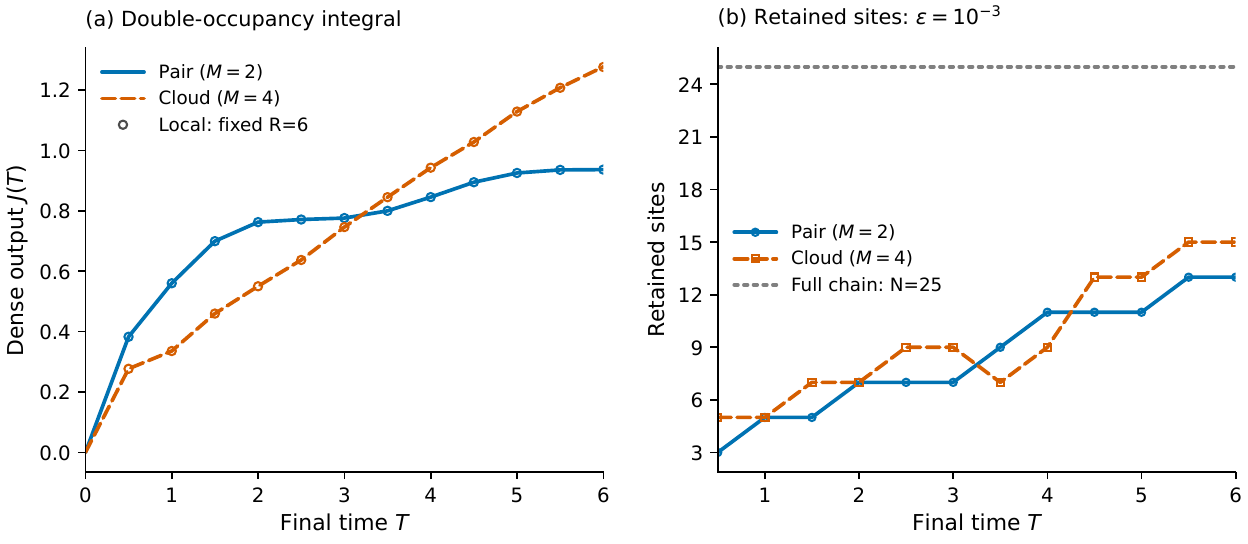}
\caption{Bosonic double-occupancy dense outputs for $N=25$, $\alpha=6$, $U=4$, and $V_0=1$. (a) Full-chain integrals for the pair and four-particle cloud, with local results at fixed $R=6$ shown as open circles. (b) Selected site counts at $\epsilon=10^{-3}$.}
\label{fig:num-bosons}
\end{figure}

\section{Conclusion}
\label{sec:conclusion}
 
We have developed a locality-based truncation framework for the quantum
dense-output problem on lattice Hamiltonians. Rather than simulating the
full many-body Hamiltonian, our approach restricts the dynamics to a
neighbourhood $X_R$ of the observable's support and controls the resulting error using LRB and localization arguments. This yields
explicit sufficient truncation radii for free fermionic and bosonic systems, interacting fermionic systems, long-range interacting bosons, and open quantum systems governed by local Lindbladians. Translating these radii into Hamiltonian-query bounds shows that the dense-output problem can be solved with a query count that is independent of the total system size in favourable regimes. For free systems, the improvement ranges from a constant factor for fast decay to an unbounded advantage in the thermodynamic limit for slow
decay. For interacting fermions, the decay profile of the interactions determines whether the required region grows linearly or exponentially with time. For long-range bosons, density and vacuum-shell assumptions give polynomial-in-$T$ truncation radii. For open systems, a sufficiently large Lindbladian spectral gap reduces the radius to logarithmic growth in $T$, while a weak or vanishing
gap leaves a linear light-cone contribution. Our numerical experiments on free, interacting fermionic, and bosonic lattice models support the underlying mechanism: at fixed time and precision, the selected local region varies only weakly with the full system size, whereas the full Hamiltonian norm continues to grow. The main conceptual implication is that spatial locality can be treated as a computational resource for extracting classical information from quantum
simulations. The dense-output signal is not necessarily a global property that requires global dynamics; it is controlled by the part of the system lying inside an effective information light cone around the observable. This decouples, at fixed $T$ and $\epsilon$, the cost of the dense-output task from the total number of lattice sites in regimes where the light cone is sufficiently narrow. 

Several directions remain open. First, it is important to establish lower bounds that quantify how much of the observed locality advantage is fundamental, and to determine whether the truncation radii and query bounds can be tightened. Moreover, it would be desirable to combine the locality truncation with history-state and quantum linear-ODE-solver approaches, which can improve the time scaling, but for which a robust truncation analysis is still lacking. Lastly, in the open-system setting, a systematic treatment of stationary-state mismatch, Liouvillian gaps, metastability, and non-Markovian effects would strengthen the theory.

\section*{Acknowledgement}

The authors thank Dong An, Chao Song, Haohua Wang, Jingyao Wang for helpful discussions. JPL acknowledges support from the Quantum Science and Technology National Science and Technology Major Project under Grant No.~2024ZD0300500, the Excellent Young Scientists Fund Program, start-up funding from Tsinghua University, and the Beijing Institute of Mathematical Sciences and Applications. CL acknowledges support from National Natural Science Foundation of China (Grant No. 12225507).

\section*{Use of AI-assisted tools}

Some of the examples and reference in the current literature on the LRB were supported by the search engine from Claude-5. Proofs and grammars were partly verified by the Claude-5 model as well and checked with Grammarly. The authors independently formulated the research questions, established the mathematical conventions as well as the physics background, performed and designed the numerical experiments and carefully verified every statement and proof included in the manuscript. The authors take full responsibility for the final text and all mathematical claims contained therein.

\bibliographystyle{unsrt}
\bibliography{references.bib}

\clearpage

\appendix
\begin{center}
    \textbf{Supplementary Materials}
\end{center}

\section{Examples}
\label{app:application}

In this section, we analyze several well-studied physical models and provide numerical illustrations of our observable-level dense-output framework.
We first consider a power-law interacting spin chain, where the interaction exponent $\alpha$ is experimentally tunable over a range of values.
On a 1D periodic chain ($D=1$) of $N$ sites, consider the long-range transverse-field Ising Hamiltonian
\begin{equation}
    H(t)=J_0\sum_{i<j}\frac{1}{r_{ij}^\alpha}\,Z_iZ_j+\sum_i h_i(t)X_i,
    \qquad
    \alpha\ge0,
    \label{eq:longrange_ising}
\end{equation}
with $\|h_i(t)\|\le1$. Since $\|Z_iZ_j\|=1$, the two-body terms satisfy $\|h_{ij}\|=J_0r_{ij}^{-\alpha}$. 
By analogous LRB arguments for spin systems, one expects a larger advantage for stronger decay rates $\alpha>1$ and more pronounced benefits for weaker interactions at moderate times.
Physically, for example, $\alpha\in(0,3)$ is the experimentally accessible range for trapped-ion simulators, while dressed Rydberg-atom arrays realize the short-range-like $\alpha=6$ regime.

A finite-size numerical illustration is given in Fig.~\ref{fig:num-app-ising}, where an $11$-site retained region accurately reproduces the early-time dense output of the $N=16$ chain. 

We next consider the lattice Schwinger model. Physically, the Schwinger model is the standard testbed for gauge-theory simulation because of its simplicity. Despite living in one spatial dimension, the model exhibits qualitative features of quantum chromodynamics that are hard to access with classical methods~\cite{schwinger1962gauge,banks1976strong}. We should mention that similar proposals based on local truncation with LRB were mentioned before~\cite{shaw2020quantum}, which built on
the Kogut-Susskind Hamiltonian formulation~\cite{kogut1975hamiltonian}
of Wilson's lattice gauge theory~\cite{wilson1974confinement}.

The model comes with $N$ fermionic sites $\{\psi_r\}_{r=1}^N$ and $N-1$ intervening gauge links carrying
a compact $U(1)$ field $U_r=e^{iaA_r}$. To write out the Hamiltonian $H = H_E + H_I + H_M$ explicitly~\cite{kogut1975hamiltonian,banks1976strong},
\begin{equation}
    H_E=\sum_r E_r^2,
    \quad
    H_I = x\sum_r\bigl[U_r\psi_r^\dagger\psi_{r+1}-U_r^\dagger\psi_r\psi_{r+1}^\dagger\bigr],
    \quad
    H_M=\mu\sum_r(-1)^r\psi_r^\dagger\psi_r,
    \label{eq:sch_hamiltonian}
\end{equation}
with $x=(ag)^{-2}$ the dimensionless hopping strength, $\mu=2m/(ag^2)$ the mass parameter, $a$ the lattice spacing, and $g$ the gauge coupling. After the Jordan-Wigner transformation, $H_I$ becomes an exactly nearest-neighbour fermion-gauge coupling,
\begin{equation}
    H_I = \tfrac14 x\sum_r
    \Bigl[(U_r+U_r^\dagger)(X_rX_{r+1}+Y_rY_{r+1})
    + i(U_r-U_r^\dagger)(X_rY_{r+1}-Y_rX_{r+1})\Bigr],
    \label{eq:sch_HI_pauli}
\end{equation}
so that, writing $T_r$ for the two-site hopping term (fermion pair
and the intervening link) and $D_r=D_r^{(M)}+D_r^{(E)}$ for the
one-site mass-plus-electric diagonal term,
\begin{equation}
    H(t) = \sum_{r=1}^{N}\bigl(T_r + D_r\bigr),
    \label{eq:sch_TD_decomp}
\end{equation}
with $T_r$ supported only on $\{r,r+1\}$ and $D_r$ supported only on $\{r\}$ with $D_r:=D_r^{(M)}+D_r^{(E)}:=-\tfrac{\mu}{2}(-1)^rZ_r+E_r^2$. Note that $D_r$ is diagonal in the computational basis: $D_r^{(M)}\propto Z_r$ and $D_r^{(E)}=E_r^2$ is diagonal in the electric eigenbasis by construction. Each operator is diagonal and local on the site-link unit cell. Shaw et al.\cite{shaw2020quantum} analyzed the locality of~\eqref{eq:sch_TD_decomp} using the finite-range LRB of Haah et al.~\cite{haah2021quantum} to make dynamical simulation algorithm perform better, such as trotterization. Two consequences follow: First, there is no gain for the global propagator, which means that locality alone does not reduce the cost of simulating the full unitary. However, if instead the target is a local observable, applying the finite-range LRB directly to
$O$ rather than to the propagator allow one to significantly improve the cost of simulation by removing the dependence on N. Unfortunately this is at the expense of increasing the dependence on time quadratically. More specifically, \cite{shaw2020quantum} truncated the electric field to a cutoff and apply
the Jordan-Wigner transformation which preserves non-locality in 1D.

Now, let $O=\hat n_k$ and let $J_H(O)=\int_0^T\langle\psi(t)|\hat n_k(t)|\psi(t)\rangle\dd t$. From~\cite{shaw2020quantum} we have,
\begin{equation}
    \bigl\|O(t) - O_R(t)\bigr\|
    \;\le\;
    \Theta\!\Bigl(\tfrac{\Gamma t}{\sqrt\kappa}\Bigr)\bigl(e^{\Gamma t\sqrt{8\kappa}}-1\bigr)\,
    e^{-d(0,X_R^c)},
    \label{eq:sch_local_lrb}
\end{equation}
where $X_R\subset\Lambda$ is a sublattice of radius $R$ centerd at $0$, and
\begin{equation}
    \Gamma=\Theta\bigl(x+\mu+\Lambda_E^2\bigr),
    \qquad
    \kappa=\Theta\!\Bigl(\tfrac{\mu+\Lambda_E}{\mu+\Lambda_E^2}\Bigr)
    \label{eq:sch_Gamma_kappa}
\end{equation}
are rate constants set by the local energy scales of $H$. One can easily integrate~\eqref{eq:sch_local_lrb} over $t\in[0,T]$ for the dense output framework. Before we give our scaling, we will derive several important lemmas to improve the previous bound. The standard finite-range LRB we want to invoke needs every term of the Hamiltonian to obey the same fixed norm bound, \textit{e.g.,} $\norm{H_r(t)}\le C$, one constant $C$ that works for every site $r$ and every time $t$, chosen once and for all before any lattice or cutoff parameters are sent to their physical limits. We call such a bound uniform. But the scaling can be improved if the bound itself grows as we push $\Lambda$ (or $N$) larger. This matters because the light cone could get worse and worse in exactly the regime we actually care about. The raw Hamiltonian $H=\sum_r(T_r+D_r)$ mixes a term that has this property with one that does not. The hopping term satisfies $\norm{T_r}\le2x$ for every $r$ which is genuinely uniform. The on-site term $D_r$ contains $E_r^2$, whose norm grows like $\Lambda^2$; there is no single constant bounding $\norm{D_r}$ that stays fixed as we increase $\Lambda$. Applying a LRB directly to the sum $H=T+D$, as in \cite{shaw2020quantum} do, forces one to use a $J$ large enough to cover $D_r$ as well, so $D_r$'s $\Lambda^2$ growth leaks into the light-cone rate. Even though $D_r$ is diagonal thus, physically, cannot transport anything from one site to another. Before invoking the LRB, we therefore first remove $D_r$ from the problem entirely by passing to the interaction picture. 
\begin{lemma}[$D_r$ is exactly on-site and diagonal]
\label{lem:diag}
$[D_r,D_{r'}]=0$ for all $r,r'$, and $[Z_k,H_0]=0$ for every $k$, where $H_0:=\sum_rD_r$.
\end{lemma}
\begin{proof}
$D_r^{(M)}\propto Z_r$ commutes with $Z_{r'}$ for every $r,r'$ (same or different qubit; Pauli-$Z$'s always
commute). $D_r^{(E)}=E_r^2$ acts on the link-$r$ register, a tensor factor disjoint from every qubit site, so it
commutes with every $Z_k$ and with every $D_{r'}^{(M)}$. Hence every term of $H_0$ commutes with $Z_k$.
\end{proof}

\begin{corollary}[$Z_k$ Evolution]
\label{cor:frozen}
Let $U_0(t):=\mathcal T\exp(-i\int_0^tH_0(s)\,ds)$. Then $U_0(t)^\dagger Z_kU_0(t)=Z_k$ for all $t$ and all $k$.
\end{corollary}
This is stronger than what we mentioned before as the generic `on-site backgrounds don't spread observables' fact. An on-site $H_0$ leaves an observable's support fixed but its matrix elements still rotate under the internal dynamics. Here $H_0$ is not merely on-site but diagonal in the same basis as $Z_k$.
\begin{proposition}[Reduction to the residual hopping]
\label{prop:reduction}
Let $U(t)$ solve $i\dot U(t)=HU(t)$, and factor $U(t)=U_0(t)\tilde U(t)$. Then $\tilde U(t)$ solves
$i\dot{\tilde U}(t)=V_I(t)\tilde U(t)$ with $V_I(t):=U_0(t)^\dagger V\,U_0(t)$, $V:=\sum_rT_r$, and
\begin{equation}
O(t):=U(t)^\dagger Z_kU(t) = \tilde U(t)^\dagger Z_k\tilde U(t).
\label{eq:exact-reduction}
\end{equation}
Moreover $V_I(t)=\sum_rT_r^I(t)$ with $T_r^I(t):=U_0(t)^\dagger T_rU_0(t)$ supported exactly on $\{r,r+1\}$ for every $t$, and $\norm{T_r^I(t)}=\norm{T_r}\le2x$.
\end{proposition}
\begin{proof}
The setup basically is in the interaction picture frame: $U(t)=U_0(t)\tilde U(t)$ and Corollary~\ref{cor:frozen} give $O(t)=\tilde U(t)^\dagger U_0(t)^\dagger
Z_kU_0(t)\tilde U(t)=\tilde U(t)^\dagger Z_k\tilde U(t)$. Since $H_0=\sum_xD_x(t)$
is a sum of commuting terms on disjoint tensor factors, $U_0(t)=\bigotimes_xu_x(t)$; conjugating $T_r$ by this product, only the three adjacent factors act nontrivially, so
$T_r^I(t)=u_r(t)^\dagger u_{r+1}(t)^\dagger u_{\mathrm{link},r}(t)^\dagger\,T_r\,u_r(t)u_{r+1}(t)u_{\mathrm{link},r}(t)$
remains supported on $\{r,r+1\}$, and unitary conjugation preserves the norm.
\end{proof}

\begin{theorem}[Dense-output truncation for the lattice Schwinger model]
\label{thm:main}
Let $O=Z_k$, $\norm{T_r^I(t)}\le2x$ (from Proposition~\ref{prop:reduction}), $X_R$ the ball of radius $R$ around $k$, $V_{I,X_R}(t):=\sum_{r: \{r,r+1\}\subset X_R}T_r^I(t)$, and $\tilde U_R(t)$ its propagator. Define
\begin{equation}
J_H(O):=\int_0^T\langle\psi_0|e^{itH}Z_ke^{-itH}|\psi_0\rangle\,dt,\qquad
J_{H_R}(O):=\int_0^T\langle\psi_0|\tilde U_R(t)^\dagger Z_k\tilde U_R(t)|\psi_0\rangle\,dt.
\end{equation}
Then
\begin{equation}
\big|J_{H}(O)-J_{H_R}(O)\big| \;\le\; \frac{8xT^2(8xT)^R}{R!}\,e^{8xT}.
\label{eq:main-bound}
\end{equation}
To achieve $|J-J_R|\le\epsilon$ it suffices to take
\begin{equation}
R_\epsilon(T) = \Theta\big(xT+\log(1/\epsilon)\big).
\label{eq:radius}
\end{equation}
\end{theorem}
\begin{proof}
The proof follows the previous Duhamel identity and we the standard LRB~\cite{lieb1972finite} to each bond $r$ at distance $d\ge R$ from $k$ with some constant $C$,
\begin{equation}
\norm{[\delta V_I(s),O_R(s)]}\le C\,\frac{(4s)^R}{R!}\,e^{8xs}.
\end{equation}
Integrating over $s\in[0,t]$ and we use the trick to bound the integrand by its value at $s=t$ times $t$, since it is non-negative and non-decreasing,
\begin{equation}
\norm{O(t)-O_R(t)}\le \frac{Ct(8xt)^R}{R!}e^{4Ct}.
\end{equation}
Integrating again over $t\in[0,T]$ the same way gives the appropriate form. Using Stirling formula and setting this $\le\epsilon$ gives $R_\epsilon(T)=\Theta(xT+\log(1/\epsilon))$.
\end{proof}

Integrating Eq.~\eqref{eq:sch_local_lrb} directly over $t\in[0,T]$, exactly as we did above, gives a radius $R_\epsilon(T)=\Theta\big(\Gamma\sqrt{8\kappa} T+\log(1/\epsilon)\big)$. Holding $x,\mu$ fixed and sending $\Lambda_E\to\infty$, where one gets the continuum-limit regime of practical interest, one obtains
\begin{equation}
\Gamma\sqrt{8\kappa}=\Theta\big(\Lambda_E^{3/2}\big).
\end{equation}
The required truncation radius, and hence the qubit and gate count of is
therefore improved from $\Theta(\Lambda_E^{3/2}T+\log(1/\epsilon))$ to $\Theta(xT+\log(1/\epsilon))$.

A finite-size numerical illustration of this truncation is provided in Fig.~\ref{fig:num-app-schwinger}.

\section{Proof of Theorem~\ref{thm:free_fermion_main}}
\label{app:slow}
In this appendix, we will give derivation for the dense-output scaling for the free fermion model with slow decay rate $\alpha \le D$. Recall that we work with Hamiltonian,
\begin{equation}
H(t) \;=\; \sum_{i,j\in\Lambda} h_{ij}(t)\, c_i^\dagger c_j ,
\end{equation}
where $c_i^\dagger$ and $c_i$ are fermionic creation and annihilation
operators satisfying the canonical anticommutation relations
$\{c_i, c_j^\dagger\} = \delta_{ij}$. And we have the truncated Hamiltonian,
\begin{equation}
    H_R(t) \sum_{i,j\in X_R \subset \Lambda} h_{ij}(t)\, c_i^\dagger c_j,
    \label{eq:truncated_hamiltonian}
\end{equation}
with propagator $U_R(t) = \mathcal{T}\exp(-i\int_0^t H_R(s)\dd s)$.
The truncated Heisenberg observable is
\begin{equation}
    O_R(t) = U_R(t)^\dagger O U_R(t),
    \label{eq:truncated_observable}
\end{equation}
and for an initial state $\ket{\psi_0}$, the truncated dense output follows
\begin{equation}
    J_{H_R}(O) = \int_0^T \langle\psi_0|O_R(t)|\psi_0\rangle \dd t.
    \label{eq:truncated_dense_output}
\end{equation}
We take $O$ to be in the general form,
\begin{equation}
    O = \sum_{i,j\in X_0} M_{ij}c_i^\dagger c_j,
    \qquad \|M\|\le 1,
    \label{eq:O_singleBody}
\end{equation}
supported on a finite region $X_0\subset\Lambda$. Under the full dynamics, the Heisenberg operator is
\begin{equation}
    O(t) = U(t)^\dagger OU(t)
    = \sum_{i,j\in X_0} M_{ij}
    c_i^\dagger(t)c_j(t) = \sum_{i,j\in X_0} M_{ij}
    \sum_{k,l\in\Lambda}
    G_{ki}(t)^\dagger G_{lj}(t)c_k^\dagger c_l.
    \label{eq:O_heisenberg}
\end{equation}
where $c_k(t) = U(t)^\dagger c_k U(t)$ and substituting in the single-particle propagator $G(t)$. Similarly, the truncated operator $O_R(t) = U_R(t)^\dagger OU_R(t)$ uses the truncated propagator $G^R(t)$, supported on $X_R$:
\begin{equation}
    O_R(t) = \sum_{i,j\in X_0} M_{ij}
    \sum_{k,l\in X_R}
    G^R_{ki}(t)^\dagger G^R_{lj}(t)c_k^\dagger c_l.
    \label{eq:OR_expanded}
\end{equation}
We want to bound the difference $O(t) - O_R(t)$:
\begin{equation}
    \|O(t)-O_R(t)\|
    \le
    \sum_{i,j\in X_0}|M_{ij}|
    \left|\sum_{\substack{k,l\in\Lambda \\ \max(d(r_{k},X_0),d(r_{l}-X_0))>R}}
    G_{ki}(t)^\dagger G_{lj}(t)c_k^\dagger c_l\right| \le |X_0|\sqrt{\max_{i\in X_0}
    \sum_{k:d(r_{k},X_0)>R}|G_{ki}(t)|^2},
    \label{eq:norm_diff_step1}
\end{equation}
where the last step is obtained by Cauchy-Schwarz. For power-law hopping $|J_{ij}(t)|\le J_0 r_{ij}^{-\alpha}$ for distance $r$ and $N$ being the number of lattice sites, we have~\cite{guo2020signaling}
\begin{equation}
    \Bigl\|\bigl\{c_X(t),c_Y^\dagger\bigr\}\Bigr\|
    \le t\times
    \begin{cases}
        \calO(1), & \alpha > D/2,\\[4pt]
        \calO\!\left(N^{\frac{1}{2}-\frac{\alpha}{D}}\right), & 0\le\alpha\le D/2.
    \end{cases}
    \label{eq:lrb_free_simplified}
\end{equation}
Since \(c_k(t)=\sum_lG_{lk}(t)c_l\), the canonical anticommutation relations give
\begin{equation}
    |G_{ki}(t)|
    =
    \|\{c_i(t),c_k^\dagger\}\|
    \le
    \int_0^t\dd\tau
    \sqrt{\sum_{m\in\Lambda}|J_{mk}(\tau)|^2}.
    \label{eq:G_is_commutator_norm}
\end{equation}
Take the case with $\alpha > D/2$:
\begin{equation}
\|O(t)-O_R(t)\|^2\le 
    \sum_{k:\,r_{kX_0}>R}(\int_0^t\dd\tau\sqrt{\sum_{m\in\Lambda}|J_{mk}(\tau)|^2})^2
    \le t\sum_{m\ne k, r_{kX_0}>R} r_{mk}^{-2\alpha}
    \le t
    \int_R^\infty r^{D-1-2\alpha}\dd r
    = \frac{t^2}{R^{2\alpha-D}},
    \label{eq:pointwise_final}
\end{equation}
where we have omitted some integration constant for simplicity.
Integrating~\eqref{eq:pointwise_final} over $t\in[0,T]$:
\begin{equation}
    |J_H(O) - J_{H_R}(O)|\le
    \int_0^T \|O(t)-O_R(t)\|\,\dd t \le
    \frac{T^2}{R^{\alpha-D/2}},
    \label{eq:J_error_integral}
\end{equation}
which proves Theorem~\ref{thm:free_fermion_main}.

\section{Spatial truncation for slow-decay free fermions}
\label{app:slow_decay}

We prove a truncation bound for translation-invariant hopping.
The finite Hamiltonians keep all couplings within their respective regions
unchanged. The proof compares their actual time evolutions, not a spatial
projection of the full propagator. We then give a qualitative convergence
result that does not require translation invariance.

\subsection{Setup and truncation bound}

Consider one fermionic mode at each site of $\mathbb Z^D$, with hopping
\begin{equation}
    h_{xy}=j(x-y),\qquad
    j(-r)=\overline{j(r)},\qquad
    |j(r)|\le J_0|r|_\infty^{-\alpha}\quad(r\ne0),
    \qquad \alpha>\frac D2,
    \label{eq:slow-hopping}
\end{equation}
where $|r|_\infty=\max_\ell|r_\ell|$ and $j(0)$ is real.
Other lattice norms only change constants.
Let $X_0$ be a fixed finite observable support and define
\begin{equation}
    X_R=\{x\in\mathbb Z^D:d_\infty(x,X_0)\le R\},
    \qquad R\in\mathbb N,\quad R\ge1.
    \label{eq:slow-region}
\end{equation}
The full finite region $\Lambda$ contains $X_R$.
For a finite set $A\subseteq\Lambda$, write
\begin{equation}
    H_A=\sum_{x,y\in A}j(x-y)c_x^\dagger c_y,
    \qquad
    \tau_t^A(O)=e^{itH_A}Oe^{-itH_A}.
    \label{eq:slow-finite-H}
\end{equation}
We use $H_R=H_{X_R}$.
Each $H_A$ acts on the fermionic Fock space over $\Lambda$, with all
terms involving sites outside $A$ omitted.

For a density matrix $\rho$ on this full Fock space, define
\begin{equation}
    J_{H_A}(O)=\int_0^T
    \operatorname{Tr}\!\left[\rho\,\tau_t^A(O)\right]\,dt.
    \label{eq:slow-output}
\end{equation}
When $A=X_R$, this expression depends only on the restriction of $\rho$
to the retained region.

\begin{theorem}[Translation-invariant truncation]
\label{thm:slow-truncation}
Assume \eqref{eq:slow-hopping} and choose
\begin{equation}
    0<s<1,\qquad s\le\alpha-\frac D2.
    \label{eq:slow-exponent}
\end{equation}
There is a constant $C=C(D,s,X_0)$ such that every Hermitian observable
$O$ supported on $X_0$ satisfies
\begin{align}
    \|\tau_t^\Lambda(O)-\tau_t^{X_R}(O)\|
    &\le C\|O\|J_0\,\frac{|t|}{R^s},
    \qquad t\in\mathbb R,
    \label{eq:slow-pointwise}\\
    |J_{H_\Lambda}(O)-J_{H_R}(O)|
    &\le C\|O\|J_0\,\frac{T^2}{R^s}.
    \label{eq:slow-dense}
\end{align}
The constant is independent of $\Lambda$, $R$, $T$, $j(0)$, and $\rho$.
The bounds hold on the full Fock space and require no restriction on the
initial particle number or its spatial distribution.
\end{theorem}

It therefore suffices to choose an integer radius
\begin{equation}
    R\ge
    \max\left\{1,\left(
    \frac{C\|O\|J_0T^2}{\epsilon}\right)^{1/s}\right\}.
    \label{eq:slow-radius}
\end{equation}
This is a sufficient radius, not a claim about the smallest possible radius.
In one dimension, the whole slow-decay range $1/2<\alpha\le1$ permits
$s=\alpha-1/2$. In particular, for $O=n_0$,
\begin{equation}
    |J_{H_\Lambda}(n_0)-J_{H_R}(n_0)|
    \le C_\alpha\,\frac{J_0T^2}{R^{\alpha-1/2}}.
    \label{eq:slow-one-dimensional}
\end{equation}
In higher dimensions the exponent $s=\alpha-D/2$ is available whenever
$0<\alpha-D/2<1$. If $\alpha-D/2\ge1$, the theorem gives any fixed
$s<1$.

\subsection{From one-particle dynamics to local observables}

For $f\in\ell^2(\mathbb Z^D)$, use the convention
\begin{equation}
    c(f)=\sum_x\overline{f_x}c_x,\qquad
    \|c(f)\|=\|f\|_2.
    \label{eq:slow-CAR}
\end{equation}
The norm identity follows from the canonical anticommutation relations.
A one-particle unitary $u$ induces the fermionic evolution
$\tau_u(c(f))=c(u^\dagger f)$.
For a finite Hamiltonian, $u=e^{-ith_A}$ gives $\tau_u=\tau_t^A$.
Here and below, the one-particle matrix $h_A$ is extended by zero outside
$A$, so $e^{-ith_A}$ is the identity on the exterior.

\begin{lemma}[Local observable comparison]
\label{lem:slow-CAR}
For any two one-particle unitaries $u,v$ and any observable supported on
the fixed finite set $X_0$,
\begin{equation}
    \|\tau_u(O)-\tau_v(O)\|
    \le C_{X_0}\|O\|
    \max_{q\in X_0}\|(u^\dagger-v^\dagger)|q\rangle\|_2.
    \label{eq:slow-CAR-comparison}
\end{equation}
For $O=n_q$, one may take $C_{X_0}=2$.
\end{lemma}
\begin{proof}
For a single annihilation operator, \eqref{eq:slow-CAR} gives
\[
    \|\tau_u(c_q)-\tau_v(c_q)\|
    =\|(u^\dagger-v^\dagger)|q\rangle\|_2.
\]
The same holds for its adjoint. Adding and subtracting a cross term
in $c_q^\dagger c_q$ proves the density bound.

For a general local observable, use a fixed normal-ordered monomial
basis of the algebra on $X_0$. Each monomial has finitely many factors
$c_q$ or $c_q^\dagger$, each of norm one.
A telescoping expansion bounds its evolution difference by the sum of
the single-factor differences. Since the algebra on $X_0$ is
finite-dimensional, the sum of the absolute expansion coefficients,
weighted by monomial length, is at most $C_{X_0}\|O\|$.
This proves \eqref{eq:slow-CAR-comparison}.
\end{proof}

\subsection{A one-particle truncation estimate}

All torus integrals and $L^2$ norms in this subsection use normalized
Haar measure. Thus $\int_{\mathbb T^D}1\,dk=1$.
We use the Fourier transform
$\widehat f(k)=\sum_x f_xe^{-ik\cdot x}$.

Since $\alpha>D/2$, the sequence $j$ belongs to $\ell^2(\mathbb Z^D)$.
Its Fourier transform
\begin{equation}
    \omega(k)=\sum_{r\in\mathbb Z^D}j(r)e^{-ik\cdot r}
    \label{eq:slow-dispersion}
\end{equation}
is understood in $L^2(\mathbb T^D)$ and is real almost everywhere.
Define the infinite-volume operator $h$ by multiplication by $\omega$
in Fourier space, with domain
\[
    D(h)=\{f\in\ell^2:\omega\widehat f\in L^2\}.
\]
This is a self-adjoint operator. Its matrix elements are $j(x-y)$,
and every finite-support vector belongs to its domain. Consequently,
the finite matrices used above are exactly $h_A=P_AhP_A$, where $P_A$
projects onto $A$.

\paragraph{Translation of the dispersion.}
Under \eqref{eq:slow-exponent}, the hopping also obeys
$|j(r)|\le J_0|r|_\infty^{-D/2-s}$ for $r\ne0$.
Parseval's identity gives
\begin{equation}
    \|\omega(\cdot-p)-\omega\|_2^2
    =\sum_r |j(r)|^2|e^{ir\cdot p}-1|^2
    \le C_DJ_0^2
    \sum_{m\ge1}m^{-1-2s}
    \min\{1,m^2|p|_\infty^2\}.
    \label{eq:slow-dispersion-sum}
\end{equation}
The power $m^{-1-2s}$ includes the number of sites in a shell.
For $0<|p|_\infty\le1$, split the sum at $m=|p|_\infty^{-1}$:
\[
    |p|_\infty^2
    \sum_{m\le|p|_\infty^{-1}}m^{1-2s}
    +\sum_{m>|p|_\infty^{-1}}m^{-1-2s}
    \le C_s|p|_\infty^{2s}.
\]
The estimate extends to all $p\in\mathbb T^D$ by increasing the constant.
Hence
\begin{equation}
    \|\omega(\cdot-p)-\omega\|_2
    \le C_{D,s}J_0|p|_\infty^s.
    \label{eq:slow-dispersion-shift}
\end{equation}
The on-site coefficient $j(0)$ cancels in this difference.

\paragraph{A finite-support window.}
Let $W_R$ be multiplication by
\begin{equation}
    w_R(x)=\prod_{\ell=1}^D
    \left(1-\frac{|x_\ell|}{R+1}\right)_+,
    \qquad (a)_+=\max\{a,0\}.
    \label{eq:slow-window}
\end{equation}
This window is supported on the cube $B_R(0)=\{x:|x|_\infty\le R\}$
and satisfies $w_R(0)=1$.
Its Fourier kernel is
\begin{equation}
    K_R(p)=\prod_{\ell=1}^D F_R(p_\ell),
    \qquad
    F_R(p)=\frac{1}{R+1}
    \left|\sum_{m=0}^R e^{imp}\right|^2.
    \label{eq:slow-Fejer}
\end{equation}
Thus $K_R\ge0$, $\int_{\mathbb T^D}K_R(p)\,dp=1$, and
$\widehat{W_Rf}=K_R*\widehat f$.
The geometric-series formula implies
\[
    F_R(p)\le C\min\{R,(Rp^2)^{-1}\},
    \qquad 0<|p|\le\pi.
\]
For $0<s<1$, splitting the integral at $|p|=R^{-1}$ gives
\[
    \int_{\mathbb T}F_R(p)|p|^s\,dp
    \le C_sR^{-s}.
\]
Taking products and using
$|p|_\infty^s\le\sum_{\ell=1}^D|p_\ell|^s$, we obtain
\begin{equation}
    \int_{\mathbb T^D}K_R(p)|p|_\infty^s\,dp
    \le C_{D,s}R^{-s}.
    \label{eq:slow-window-moment}
\end{equation}
In particular,
\begin{equation}
    a_R:=
    \int_{\mathbb T^D}K_R(p)
    \|\omega(\cdot-p)-\omega\|_2\,dp
    \le C_{D,s}J_0R^{-s}.
    \label{eq:slow-window-error}
\end{equation}
The window is only a proof device. It does not change the Hamiltonian
used in the local simulation.

\begin{lemma}[Finite-section evolution]
\label{lem:slow-one-particle}
Let $A$ be any finite set containing the cube $B_R(q)$.
Under \eqref{eq:slow-hopping} and \eqref{eq:slow-exponent},
\begin{equation}
    \|(e^{-ith}-e^{-ith_A})|q\rangle\|_2
    \le C_{D,s}J_0\,\frac{|t|}{R^s},
    \qquad t\in\mathbb R.
    \label{eq:slow-one-particle}
\end{equation}
\end{lemma}
\begin{proof}
By translation invariance we may take $q=0$.
Set
\[
    \psi(t)=e^{-ith}|0\rangle,\qquad
    \psi_A(t)=e^{-ith_A}|0\rangle,\qquad
    v_R(t)=W_R\psi(t).
\]
In Fourier space, $\widehat\psi(t,k)=e^{-it\omega(k)}$.
Positivity and unit integral of $K_R$, together with
$|e^{-ita}-e^{-itb}|\le |t||a-b|$, give
\begin{align}
    \|\psi(t)-v_R(t)\|_2
    &\le
    \int_{\mathbb T^D}K_R(p)
    \|e^{-it\omega(\cdot)}-e^{-it\omega(\cdot-p)}\|_2\,dp
    \nonumber\\
    &\le |t|a_R.
    \label{eq:slow-window-tail}
\end{align}

We next compare $v_R(t)$ with the actual finite-system solution.
Both $\psi(t)$ and $v_R(t)$ lie in $D(h)$: the Fourier transform of
$\psi(t)$ has modulus one, while $v_R(t)$ has finite support.
Since $P_AW_R=W_R$, the residual in the finite Schr\"odinger equation is
\begin{equation}
    i\partial_t v_R(t)-h_Av_R(t)
    =P_A[W_R,h]\psi(t).
    \label{eq:slow-residual}
\end{equation}
Its Fourier representation before applying $P_A$ is
\[
    \widehat{[W_R,h]\psi(t)}(k)
    =
    \int_{\mathbb T^D}K_R(p)
    [\omega(k-p)-\omega(k)]
    e^{-it\omega(k-p)}\,dp.
\]
The phase has modulus one, so
\begin{equation}
    \|P_A[W_R,h]\psi(t)\|_2\le a_R.
    \label{eq:slow-residual-bound}
\end{equation}
Also $v_R(0)=\psi_A(0)=|0\rangle$.
For $t\ge0$, Duhamel's formula and unitarity of $e^{-ith_A}$ yield
\begin{equation}
    \|v_R(t)-\psi_A(t)\|_2
    \le\int_0^t\|P_A[W_R,h]\psi(\tau)\|_2\,d\tau
    \le t a_R.
    \label{eq:slow-stability}
\end{equation}
Combining \eqref{eq:slow-window-tail} and \eqref{eq:slow-stability},
\[
    \|\psi(t)-\psi_A(t)\|_2
    \le 2t a_R
    \le C_{D,s}J_0tR^{-s}.
\]
Applying the same argument to $-h$ covers negative times.
For a general $q$, translate the window to $q$; its support is contained
in $A$ by assumption.
\end{proof}

\begin{proof}[Proof of Theorem~\ref{thm:slow-truncation}]
For each $q\in X_0$, both $X_R$ and $\Lambda$ contain $B_R(q)$.
Apply Lemma~\ref{lem:slow-one-particle} to each finite region and use the
infinite-volume evolution as an intermediate term. At time $-t$ this gives
\[
    \|(e^{ith_\Lambda}-e^{ith_R})|q\rangle\|_2
    \le C_{D,s}J_0|t|R^{-s}.
\]
Lemma~\ref{lem:slow-CAR} proves \eqref{eq:slow-pointwise}.
For every density matrix $\rho$,
\[
    |J_{H_\Lambda}(O)-J_{H_R}(O)|
    \le\int_0^T
    \|\tau_t^\Lambda(O)-\tau_t^{X_R}(O)\|\,dt.
\]
Integration proves \eqref{eq:slow-dense}.
\end{proof}

\subsection{Time-dependent translation-invariant hopping}

The same bound holds when
$h_{xy}(t)=j(x-y,t)$ is translation invariant at every time.
Assume the coefficients are measurable, obey
$j(-r,t)=\overline{j(r,t)}$ and the same uniform off-site bound
\eqref{eq:slow-hopping}, and satisfy $j(0,\cdot)\in L^1([0,T])$.
Every finite Hamiltonian retains the original time-dependent coefficients.

The infinite-volume dispersion $\omega(k,t)$ is integrable in time with
values in $L^2(\mathbb T^D)$. Its propagator is the multiplier
\begin{equation}
    \widehat{u(t,0)f}(k)=
    \exp\left(-i\int_0^t\omega(k,\tau)\,d\tau\right)\widehat f(k).
    \label{eq:slow-driven-propagator}
\end{equation}
For the local source $|0\rangle$, this solution is absolutely continuous
in $\ell^2$ and satisfies the Schr\"odinger equation almost everywhere.

Define $a_R(\tau)$ by \eqref{eq:slow-window-error} with
$\omega(k)$ replaced by $\omega(k,\tau)$.
The uniform hopping bound gives $a_R(\tau)\le C_{D,s}J_0R^{-s}$.
The two estimates in the proof become
\[
    \|\psi(t)-W_R\psi(t)\|_2
    \le\int_0^t a_R(\tau)\,d\tau,
    \qquad
    \|P_A[W_R,h(\tau)]\psi(\tau)\|_2\le a_R(\tau).
\]
Duhamel's formula for the finite time-dependent propagator therefore gives
\begin{equation}
    \|(u(t,0)-u_A(t,0))|q\rangle\|_2
    \le C_{D,s}J_0tR^{-s}
    \quad\text{if }B_R(q)\subseteq A.
    \label{eq:slow-driven-comparison}
\end{equation}
To control the adjoints required by Lemma~\ref{lem:slow-CAR}, fix $t$ and
apply the same argument on $[0,t]$ to the reversed Hamiltonian
$-h(t-\tau)$. Its final propagators are $u(t,0)^\dagger$ and
$u_A(t,0)^\dagger$, respectively.
Thus \eqref{eq:slow-pointwise} and \eqref{eq:slow-dense} hold with the
time-ordered evolutions as well. No commutation assumption is made about
the finite Hamiltonians at different times.

A measurable time-dependent observable supported on the same fixed
$X_0$ obeys the integrated bound with $\|O\|$ replaced by
$\sup_{0\le t\le T}\|O(t)\|$.

\subsection{A convergence criterion without translation invariance}

For static hopping, a more general result gives convergence without
an explicit rate. Let $c_{00}(\mathbb Z^D)$ denote the finite-support
vectors. Saying that this space is a core of $h$ means that it is dense
in $D(h)$ for the norm
$(\|f\|_2^2+\|hf\|_2^2)^{1/2}$.

\begin{proposition}[Local convergence from a core]
\label{prop:slow-core}
Let $h$ be self-adjoint on $\ell^2(\mathbb Z^D)$ and suppose
$c_{00}(\mathbb Z^D)$ is a core.
Let $A_n$ be increasing finite sets with union $\mathbb Z^D$, and let
$h_n=P_{A_n}hP_{A_n}$, extended by zero on the exterior.
For every fixed local fermionic observable $O$ and every finite $T$,
\begin{equation}
    \lim_{n\to\infty}
    \sup_{|t|\le T}\|\tau_t^{A_n}(O)-\tau_t(O)\|=0,
    \label{eq:slow-core-convergence}
\end{equation}
where $\tau_t(c(f))=c(e^{ith}f)$.
Consequently, the dense-output difference also tends to zero uniformly
over states.
\end{proposition}
\begin{proof}
For $\phi\in c_{00}(\mathbb Z^D)$ and all sufficiently large $n$,
\[
    h_n\phi=P_{A_n}h\phi,\qquad
    \|h_n\phi-h\phi\|_2\longrightarrow0.
\]
Convergence on a core implies strong-resolvent convergence, hence
$e^{-ith_n}f\to e^{-ith}f$ for each fixed $t$ and $f$
\cite[Lemma~6.36 and Corollary~6.33]{teschl2014mathematical}.
To see uniformity on compact time intervals, first take
$\phi\in c_{00}$. The bound
\[
    \|(e^{-ith_n}-e^{-it'h_n})\phi\|_2
    \le |t-t'|\,\|h_n\phi\|_2
\]
is uniform in $n$, because $h_n\phi\to h\phi$.
Pointwise convergence on a finite time grid therefore gives uniform
convergence on $[-T,T]$. Density of $c_{00}$ and unitarity extend it
to every $f\in\ell^2$.
Apply Lemma~\ref{lem:slow-CAR} to the finitely many sources $|q\rangle$
with $q\in X_0$ to obtain \eqref{eq:slow-core-convergence}.
Finally,
\[
    \left\|\int_0^T
    [\tau_t^{A_n}(O)-\tau_t(O)]\,dt\right\|
    \le T\sup_{0\le t\le T}
    \|\tau_t^{A_n}(O)-\tau_t(O)\|\longrightarrow0.
\]
This norm bound is independent of the state.
\end{proof}

For translation-invariant hopping with $j\in\ell^2$, the core condition
is automatic. In Fourier space the graph norm is
\[
    \int_{\mathbb T^D}(1+|\omega(k)|^2)|\widehat f(k)|^2\,dk.
\]
The measure $(1+|\omega|^2)\,dk$ is finite.
Continuous functions are dense in its $L^2$ space, and trigonometric
polynomials are uniformly dense in the continuous functions.
Thus trigonometric polynomials form a core for multiplication by
$\omega$. Their inverse Fourier transforms are precisely the
finite-support vectors.

Proposition~\ref{prop:slow-core} therefore applies to every Hermitian
translation-invariant $\ell^2$ hopping, even without a specified
power-law envelope. It also applies to non-translation-invariant
self-adjoint operators for which the stated core condition holds.
For a fixed observable, time interval, and accuracy, the resulting
operator-norm convergence implies that one finite section approximates
all sufficiently large sections of the same exhaustion, uniformly over
their initial states. Under the power-law assumptions, the earlier
window argument supplies the explicit sufficient radius
\eqref{eq:slow-radius}.

\section{Optimization in the interaction picture}
\label{app:int_pic}
There are many ways to improve the observable-level simulation accuracy~\cite{yu2025observable, yang2025circuit}. In this appendix we show that passing to an interaction picture and by carefully making the choice of the observable, one can yield substantially improved error bounds in addition to the Hamiltonian truncation. The essential motivation is that once we move the observable into the Heisenberg picture, one can make truncation for the observable as well since typically the propagator will be a decaying function.

It is shown that in the free-fermion truncation analysis of the slow-decay model, one eliminated the on-site potential by going into the interaction picture with respect to 
$H_0(t) = \sum_i h_i(t)\,c_i^\dagger c_i$~\cite{guo2020signaling}. Since the particle propagator only gains a phase in the interaction picture, it is easy to check that the truncation error bound is identical to the Schr\"{o}dinger-picture result. This means that the on-site potential is not responsible for the spreading of the observable. More specifically, for example, we defined a Hamiltonian with $H=H_0+H_{\rm hop}$ the on-site term and the hopping term. Going into the interaction picture with respect to $H_\mathrm{hop}(t)$ and let $U_\mathrm{hop}(t) = \mathcal{T}\exp\!\Bigl(-i\int_0^t H_\mathrm{hop}(s)\,ds\Bigr)$ one gets,
\begin{equation}
  c_i(t) := U_\mathrm{hop}(t)^\dagger\,c_i\,U_\mathrm{hop}(t) = \sum_{j \in \Lambda} G_{ji}(t)\,c_j,
  \label{eq:ci-spread}
\end{equation}
where $G(t)$ is the single-particle propagator of $H_\mathrm{hop}$. The on-site Hamiltonian in the interaction picture becomes
\begin{equation}
  H_{I,\mathrm{onsite}}(t)
  = U_\mathrm{hop}(t)^\dagger\,H_\mathrm{onsite}(t)\,U_\mathrm{hop}(t)
  = \sum_{i \in \Lambda} B_i(t)\,c_i^\dagger(t)\,c_i(t),
  \label{eq:HI-onsite}
\end{equation}
where $c_i^\dagger(t)c_i(t)$ is a density operator that spreads spatially under $H_\mathrm{hop}$. For the single-site density $O = n_q = c_q^\dagger c_q$, the interaction-picture operator is
\begin{equation}
  O(t) = U_\mathrm{hop}(t)^\dagger\,n_q\,U_\mathrm{hop}(t)
  = c_q^\dagger(t)\,c_q(t)
  = \sum_{i,j \in \Lambda} G_{lq}^\dagger (t)\,G_{jq}(t)\,c_i^\dagger c_j,
  \label{eq:OI}
\end{equation}
with coefficients $w_{ij}(t) = G_{lq}^\dagger (t)G_{jq}(t)$.  The truncation error in the dense output is now controlled by how rapidly these coefficients decay for sites $i$ or $j$ far from $q$.

Therefore, to obtain improvement, one must take the interaction picture with respect to the part of the Hamiltonian that does spread the observable, namely the non-local terms. This also suggest that we should not only truncate over the Hamiltonian itself, but also the observable, if suitable, associated with the system. We partially employed this trick in the Schwinger model above, where we see advantage over truncating the Hamiltonian only by moving into the Heisenberg picture of the observable.

\subsection{Interacting system}
\label{sec:IP-interacting}

Without loss of generality, we turn to the case where we add an interacting perturbation term $g V(t)$ in the original Hamiltonian with $\|V(t)\| \leq 1$, controlled by a coupling $g \geq 0$.  We work in the regime where $g$ is very small.
Taking the interaction picture with respect to $H_\mathrm{hop}(t)$, the full propagator factors as $U(t) = U_\mathrm{hop}(t)\,U_I(t)$, where
\begin{equation}
  U_I(t) = \mathcal{T}\exp\!\Bigl(-ig\int_0^t V_I(s)\,ds\Bigr),
  \qquad V_I(s) = U_\mathrm{hop}(s)^\dagger\,V(s)\,U_\mathrm{hop}(s).
  \label{eq:UI}
\end{equation}
The dressed observable $O(t) = U_\mathrm{hop}(t)^\dagger O\,U_\mathrm{hop}(t)$ spreads under $H_\mathrm{hop}$ exactly as in~\eqref{eq:OI}, independent of $g$. The interaction $g V_I$ then drives further evolution of the dressed observable. We truncate both $O(t)$ and $V_I(t)$ to $B_R(q)$, and decompose the error as
\begin{equation}
  |J_H(O) - J_{H_R}(\tilde{O}_R)|
  \leq \int_0^T \|O(t) - O_{I,R}(t)\|\,dt
  + \int_0^T \|O_{I,R}(t) - \tilde{O}_{R}(t)\|\,dt,
  \label{eq:error-split}
\end{equation}
where $\tilde{O}_R(t) = U_{I,R}(t)^\dagger O_{I,R}(t) U_{I,R}(t)$ is the Heisenberg picture observable after truncation.
The first term is the observable truncation error.  For the second term, write $U_I(t) = U_{I,R}(t) + \Delta(t)$, where $\Delta(0) = 0$.  Subtracting the two equations of motion of the propagator gives
\begin{equation}
    i\dot{\Delta}(t)=
    g V_I(t)\,U_I(t)-
    g V_{I,R}(t)\,U_{I,R}(t).
    \label{eq:delta_eom_raw}
\end{equation}
Adding and subtracting $g V_{I,R}(t)\,\Delta(t)$ on the right-hand side:
\begin{equation}
    i\dot{\Delta}(t)=
    g V_{I,R}(t)\,\Delta(t)+
    g\delta V_I(t)\,U_I(t),
    \label{eq:delta_eom}
\end{equation}
where $\delta V_I(t) = V_I(t) - V_{I,R}(t)$. This is an inhomogeneous first-order ODE for $\Delta(t)$ with homogeneous part driven by $g V_{I,R}$ and inhomogeneous source $g\delta V_I(t)\,U_I(t)$. Applying the variation of parameters principle
with the integrating factor $U_{I,R}(t)^\dagger$:
\begin{equation}
    \Delta(t)=
    -ig\int_0^t
    U_{I,R}(t,s)\,\delta V_I(s)\,U_I(s)\,\dd s,
    \label{eq:delta_duhamel}
\end{equation}
where $U_{I,R}(t,s) = U_{I,R}(t)\,U_{I,R}(s)^\dagger$ is the
two-time propagator of the truncated system. Define
\begin{equation}
    F(t):= U_I(t)^\dagger O_{I,R}(t)\,U_I(t)-
    U_{I,R}(t)^\dagger O_{I,R}(t)\,U_{I,R}(t).
    \label{eq:F_definition}
\end{equation}
Substituting $U_I(t) = U_{I,R}(t) + \Delta(t)$:
\begin{align}
    F(t)
    &\;=\;
    \bigl(U_{I,R}(t) + \Delta(t)\bigr)^\dagger
    O_{I,R}(t)
    \bigl(U_{I,R}(t) + \Delta(t)\bigr)
    \;-\;
    U_{I,R}(t)^\dagger O_{I,R}(t)\,U_{I,R}(t)
    \nonumber\\[4pt]
    &\;=\;
    U_{I,R}(t)^\dagger O_{I,R}(t)\,\Delta(t)
    \;+\;
    \Delta(t)^\dagger O_{I,R}(t)\,U_{I,R}(t)
    \;+\;
    \Delta(t)^\dagger O_{I,R}(t)\,\Delta(t).
    \label{eq:F_expanded}
\end{align}
We work to first order in $\delta V_I$, 
\begin{equation}
    F(t)
    \;\approx\;
    U_{I,R}(t)^\dagger O_{I,R}(t)\,\Delta(t)
    \;+\;
    \Delta(t)^\dagger O_{I,R}(t)\,U_{I,R}(t).
    \label{eq:F_first_order}
\end{equation}
Insert~\eqref{eq:delta_duhamel} into~\eqref{eq:F_first_order},
\begin{align}
    F(t)
    &\;=\;
    -ig\int_0^t
    U_{I,R}(t)^\dagger O_{I,R}(t)\,
    U_{I,R}(t,s)\,\delta V_I(s)\,U_{I,R}(s)
    \,\dd s
    \nonumber\\
    &\quad
    +\,
    ig\int_0^t
    U_{I,R}(s)^\dagger\,\delta V_I(s)^\dagger\,
    U_{I,R}(t,s)^\dagger\,O_{I,R}(t)\,U_{I,R}(t)
    \,\dd s.
    \label{eq:F_substituted}
\end{align}
Now use $U_{I,R}(t,s) = U_{I,R}(t)\,U_{I,R}(s)^\dagger$ to simplify:
\begin{equation}
    U_{I,R}(t)^\dagger O_{I,R}(t)\,U_{I,R}(t,s)
    \;=\;
    U_{I,R}(t)^\dagger O_{I,R}(t)\,U_{I,R}(t)\,U_{I,R}(s)^\dagger
    \;=\;
    \widetilde{O}_R(t)\,U_{I,R}(s)^\dagger,
    \label{eq:simplify_left}
\end{equation}
and similarly $U_{I,R}(t,s)^\dagger\,O_{I,R}(t)\,U_{I,R}(t)
= U_{I,R}(s)\,\widetilde{O}_R(t)$.  Conjugating both sides of~\eqref{eq:F_substituted} by $U_{I,R}(t)$ and combining the two integrals, which differ only in the ordering of $\delta V_I(s)$ and $\widetilde{O}_R(s)$:
\begin{align}
    F(t)
    &\;=\;
    ig\int_0^t
    U_{I,R}(t)^\dagger
    \Bigl(
        \widetilde{O}_R(t)\,\delta V_{I,R}(s)
        -
        \delta V_{I,R}(s)\,\widetilde{O}_R(t)
    \Bigr)
    U_{I,R}(t)
    \,\dd s
    \nonumber\\
    &\;=\;
    ig\int_0^t
    U_{I,R}(t)^\dagger
    \bigl[\widetilde{O}_R(t),\,\delta V_{I,R}(s)\bigr]
    U_{I,R}(t)
    \,\dd s,
    \label{eq:F_final}
\end{align}
Now, since $U_{I,R}(t)$ is unitary, $\|U_{I,R}(t)\| = \|U_{I,R}(t)^\dagger\| = 1$,
\begin{equation}
    \|O_{I,R}(t) - \widetilde{O}_R(t)\|
    \;\le\;
    g\int_0^t
    \bigl\|[\delta V_{I,R}(s),\,\widetilde{O}_R(t)]\bigr\|
    \,\dd s,
    \label{eq:norm_step2}
\end{equation}
where one can obtain the error scaling for the second term in \eqref{eq:error-split} based on the LRB.  In the previous section we showed that for a general interacting system the time dependence is typically exponential.  Here the same estimate is applied only to the residual interaction in the hopping frame, so the effective prefactor carries the weak coupling \(g\).  We next explain how the observable choice controls the first term in~\eqref{eq:error-split} with examples.

\subsection{Enhanced truncation via observable choice}
\label{sec:observable-choice}

In the interaction picture, $O(t)$ could be a sum of bilinears $\sum_{ij} G_{lq}^\dagger G_{jq}\,c_i^\dagger c_j$, whose truncation error is controlled by the product of two propagator factors $|G_{lq}||G_{jq}|$.  One natural question to ask is: does the choice of observable affect the rate of decay of its interaction-picture representation?

For a general single-body observable $O = \sum_{q,r\in X_0} M_{qr}\,c_q^\dagger c_r$, the dressed observable is
\begin{equation}
  O(t) = \sum_{q,r \in X_0} M_{qr} \sum_{i,j \in \Lambda} G_{lq}^\dagger (t)\,G_{jr}(t)\,c_i^\dagger c_j.
\end{equation}
The truncation error for a site at distance $d$ from $X_0$ is proportional to $\sum_q |G_{lq}(t)|^2$. For multi-body observables, the situation is richer. Consider the $k$-body string observable
\begin{equation}
  O_\mathrm{string}^{(k)} = c_{q_1}^\dagger c_{q_2}^\dagger \cdots c_{q_k}^\dagger c_{q_k} \cdots c_{q_2} c_{q_1},
\end{equation}
where $q_1,\ldots,q_k$ are distinct sites near $q$. In the interaction picture, each creation (annihilation) operator picks up a propagator factor, so the coefficient for a configuration with
indices $\{i_1,\ldots,i_k,j_1,\ldots,j_k\}$ is
\begin{equation}
  \prod_{a=1}^k G_{i_a q_a}^\dagger (t) \prod_{b=1}^k G_{j_b q_b}(t).
  \label{eq:string-coeff}
\end{equation}
If any index $i_a$ lies at distance $d \gg 1$ from $q_a$, its propagator factor is small; crucially, the overall coefficient decays as the product of $2k$ propagator factors, so the tail falls
off faster than for a single-body observable. 

\subsection{Example: tight-binding dispersion}
\label{sec:tight-binding}

We work through the tight-binding model as an example, deriving the propagator from first principles before analyzing
the truncation error. This example illustrates that the spatial decay of the propagator $G_{lq}(t)$ directly controls the truncation error.

Consider the 1D tight-binding Hamiltonian with nearest-neighbour
hopping strength $J > 0$:
\begin{equation}
    H=
    -J\sum_i \bigl(c_i^\dagger c_{i+1} + c_{i+1}^\dagger c_i\bigr),
    \label{eq:TB_hamiltonian}
\end{equation}
defined on a ring of $L$ sites with periodic boundary conditions.
The single-particle Hamiltonian matrix is $\mathcal{H}_{ij} = -J(\delta_{i,j+1} + \delta_{i,j-1})$, which is diagonalized by the discrete Fourier transform: the eigenstates are plane waves $|k\rangle$ with momenta $k = 2\pi n/L$, $n = 0,1,\ldots,L-1$, and
eigenvalues $\epsilon_k = -2J\cos k$. The single-particle propagator $G_{lq}(t)$ satisfies $i\dot{G}(t) = \mathcal{H}\,G(t)$, $G(0) = I$, with formal solution $G(t) = e^{-i\mathcal{H}t}$.
The matrix element between sites $i$ and $q$ is
\begin{equation}
    G_{lq}(t)=\langle l|\,e^{-i\mathcal{H}t}\,|q\rangle=\frac{1}{L}\sum_{k} e^{-i\epsilon_k t}\,e^{ik(l-q)},
    \label{eq:G_fourier}
\end{equation}
where we inserted a complete set of momentum eigenstates
$|k\rangle = L^{-1/2}\sum_j e^{ikj}|j\rangle$.
In the thermodynamic limit $L\to\infty$, the sum becomes
an integral over $k\in[-\pi,\pi]$:
\begin{equation}
    G_{lq}(t)
    \;=\;
    \frac{1}{2\pi}
    \int_{-\pi}^{\pi} e^{-i\epsilon_k t}\,e^{ik(l-q)}\,\dd k
    \;=\;
    \frac{1}{2\pi}
    \int_{-\pi}^{\pi} e^{2iJt\cos k}\,e^{ik(l-q)}\,\dd k.
    \label{eq:G_integral}
\end{equation}
This integral is a standard representation of the Bessel function,
\begin{equation}
    G_{lq}(t)=
    e^{i(l-q)\pi/2}\,
    \frac{1}{2\pi}
    \int_{-\pi}^{\pi}
    e^{i(2Jt)\sin\theta}\,e^{-i(l-q)\theta}\,\dd\theta
    \;=\;
    i^{l-q}\,\mathcal{J}_{l-q}(2Jt),
    \label{eq:G_bessel}
\end{equation}
where $i^{l-q} = e^{i(l-q)\pi/2}$. For fixed $t$ and large $|l-q|$, the Bessel function satisfies the asymptotic bound
\begin{equation}
    |\mathcal{J}_n(x)|
    \;\le\;
    \left(\frac{ex}{2n}\right)^n,
    \label{eq:bessel_asymptotic}
\end{equation}
which follows from Stirling's approximation applied to the
series representation of $\mathcal{J}_n$.
Applying this with $x = 2Jt$ and $n = |l-q|$:
\begin{equation}
    |G_{lq}(t)|
    \;=\;
    |\mathcal{J}_{|l-q|}(2Jt)|
    \le
    \left(\frac{eJt}{|l-q|}\right)^{|l-q|},
    \label{eq:G_TB_asymptotic}
\end{equation}
valid for $|l-q| > eJt$.  Physically, this reflects the fact that the tight-binding dispersion has a finite maximum group velocity,
and the amplitude for a particle to propagate beyond the light cone $|l-q| > 2Jt$ is exponentially suppressed in the `distance beyond the light cone'. Then for $R \gg Jt$ one has
\begin{align}
    \|O(t)-O_R(t)\|
    &\;=\;
    \sum_{|k-q|>R} |G_{kq}(t)|^2
    \le
    \sum_{n > R}\left(\frac{eJt}{n}\right)^{2n}.
    \label{eq:tail_TB}
\end{align}
For $R > eJt$, the terms in~\eqref{eq:tail_TB} decrease
rapidly in $n$. Define
\begin{equation}
    a_n:=
    \left(\frac{eJt}{n}\right)^{2n},
\end{equation}
so that the tail is bounded by $\mathcal{T}_R(t)=\sum_{n=R+1}^{\infty} a_n$. For \(n\ge R+1\), the ratio of successive terms is
\begin{equation}
    \frac{a_{n+1}}{a_n}=
    \left(\frac{n}{n+1}\right)^{2n}
    \left(\frac{eJt}{n+1}\right)^2 \le
    \left(\frac{eJt}{n+1}\right)^2
    \le
    \left(\frac{eJt}{R+2}\right)^2
    =:\rho_R.
\end{equation}
If \(R+2>eJt\), then \(\rho_R<1\), and therefore
\begin{equation}
    a_n\le a_{R+1}\,\rho_R^{\,n-(R+1)}.
\end{equation}
Summing the resulting geometric series yields
\begin{equation}
        \|O(t)-O_R(t)\|\le a_{R+1} 
    \sum_{m=0}^{\infty}\rho_R^m =\frac{a_{R+1}}{1-\rho_R}.
\end{equation}
Substituting the definitions of \(a_{R+1}\) and \(\rho_R\),
\begin{equation}
        \|O(t)-O_R(t)\|
    \le  \left(\frac{eJt}{R+1}\right)^{2(R+1)}
    \frac{1}{1-\left(\dfrac{eJt}{R+2}\right)^2}.
    \label{eq:tail_TB_rigorous}
\end{equation}
In particular, if the radius is large enough, \textit{i.e.,} $R+2\ge \sqrt{2}eJT$, then
\begin{equation}
       \|O(t)-O_R(t)\|
    \le
    \left(\frac{eJt}{R+1}\right)^{2(R+1)}.
    \label{eq:tail_TB_simple}
\end{equation}
Integrating over $t\in[0,T]$ yields
\begin{align}
    \int_0^T     \|O(t)-O_R(t)\|\,dt &\le
    \left(
        \frac{eJ}{R+1}
    \right)^{2(R+1)}
    \int_0^T t^{2(R+1)}\,dt \\
    &=\frac{T}{(2R+3)}
    \left(
        \frac{eJT}{R+1}
    \right)^{2(R+1)}.
\end{align}
For large $R$, with Cauchy-Schwartz inequality, the dense-output error simplifies to
\begin{equation}
    |J_H(O)-J_{H_R}(O)|
    \le
    \frac{T}{\sqrt{R}}\left(\frac{eJT}{R}\right)^R.
    \label{eq:dense_output_TB}
\end{equation}
To achieve accuracy $\epsilon$, we set the right-hand side of
\eqref{eq:dense_output_TB} equal to $\epsilon$:
\begin{equation}
    \frac{T}{\sqrt{R}}
    \left(\frac{eJT}{R}
    \right)^R\sim\epsilon.
\end{equation}
Taking logarithms gives
\begin{equation}
    R\log\!\left(\frac{R}{eJT}\right)=\log\!\left(\frac{T}{\epsilon \sqrt{R}}\right).
\end{equation}
After rearrangement we get the scaling $R_\epsilon(T)=\Theta\!\left(JT+\frac{\log(1/\epsilon)}{\log\log(1/\epsilon)}\right)$.

As we can see in the tight-binding model, the decay of the propagator $|G_{lq}(t)|$ outside the light cone could make the simulation more efficiently by bounding a smaller truncation radius. If the propagator decays faster than the tight-binding Bessel asymptotics, for instance, as $|G_{lq}(t)| \sim e^{-\kappa|l-q|^2/t}$ for a Gaussian dispersion, the tail decays faster and the required truncation radius is correspondingly even smaller. Conversely, slower propagator decay yields a larger required radius. This observation also suggests a general strategy for engineering the truncation: choose the observable $O$ so that its interaction-picture representative $O(t)$ has
a rapidly decaying propagator. The hierarchy of truncation errors going from the free fermion to the interacting system suggests that the optimal strategy for computing dense outputs of long-time quantum lattice simulations is: (i) identify the dominant hopping in $H(t)$ and pass to the interaction picture with respect to it; (ii) choose an observable $O$ whose interaction-picture representative $O(t)$ has the fastest possible propagator decay; (iii) truncate $O(t)$ rather than the Hamiltonian. When the residual interaction-picture Hamiltonian $\lambda\delta V_I$ is weak, the Hamiltonian-truncation
error is correspondingly small.

\subsection{Complexity analysis}
\label{subsec:baseline_interaction_picture}

We now argue, for a general system, the strategy sketched at the end of Section~\ref{sec:tight-binding}: split the Hamiltonian as $H(t)+gV(t)$ into an exactly-solvable reference part $H(t)$ and a weak residual interaction $gV(t)$, and truncate the observable $O$ only after passing to its Heisenberg picture with respect to $H(t)$, rather than truncating the Schr\"odinger-picture Hamiltonian $H(t)$ directly. The tight-binding chain of Section~\ref{sec:tight-binding} was a single instance of $H$, used to demonstrate the mechanism of the propagator's effect concretely; here we isolate what property of $H$ (equivalently, of the chosen observable frame) is actually responsible for the complexity gain, and bound it for general reference dynamics. For example, a simple extension to interacting system could be adding the on-site repulsion $\sum U n_in_i$.

Let $H_0(t)$ be any quadratic free-fermion Hamiltonian with hopping amplitudes exhibiting power-law decay, and let $gV(t)$, $\norm{V(t)}\le1$, be a residual interaction. Let $U_0(t)$ be the propagator of $H_0(t)$ alone, so $O(t)=U_0(t)^\dagger O\,U_0(t)$ and $V_I(t)=U_0(t)^\dagger V(t)\,U_0(t)$ play the roles of $O(t)$ and $V_I(t)$. The first term of the error decomposition~\eqref{eq:error-split} is, the observable truncation error. The second term of~\eqref{eq:error-split}, bounded in~\eqref{eq:norm_step2} by $g\int_0^t\|[\delta V_I(s),\widetilde O_R(s)]\|\dd s$, has exactly the structure of the interacting-system bound of Theorem~\ref{subsec:int_ferm}, with the weak residual interaction $gV$. This is the Hamiltonian-truncation error, which competes directly against the light cone generated by the observable truncation rather than setting the complexity on its own. Hereafter, we denote $R_{\rm obs}(T,\epsilon)$ as the observble-level truncation and $R_{\rm int}(T,\epsilon)$ as the normal interacting Hamiltonian truncation radius.

\noindent \textbf{Generic power-law decay.} Assuming a power-law interacting function, one naturally obtain~\ref{subsec:int_ferm},
\begin{equation}
    R_{\rm \epsilon}(T,\epsilon) \;=\; \Theta\!\left(e^{\Theta(gT)}/\mathrm{poly}(\epsilon)\right).
    \label{eq:R_res_poly}
\end{equation}
Since the observable light cone in the previous section grows only polynomially in $T$ while here the light cone grows exponentially, the total truncation radius is, for any fixed constant $g>0$, eventually set by the interaction alone. That is saying, passing to the interaction picture buys nothing here: the bottleneck is truncating the interaction itself, not the observable.

\noindent \textbf{Exponentially local residual interaction ($F(r)=e^{-ar}(1+r)^{-(D+\eta)}$).} If $V$ instead decays exponentially rather than as a power law, repeating~\eqref{eq:fermion-Reps-exp} with $H\to gV$ gives a \emph{ballistic} residual radius,
\begin{equation}
    R_{\epsilon}(T) \;=\; \Theta\!\left(\frac{2g}{a}\,T \;+\; \frac{1}{a}\log\frac{1}{g\epsilon}\right).
    \label{eq:R_res_exp}
\end{equation}
When $H_0$ is likewise short-ranged (e.g., the tight-binding of section~\ref{sec:tight-binding}, for which the radius scales ballistically), we now have two competing linear light cones. One is set by the observable-level truncation, and one by the residual interaction. The total radius is
\begin{equation}
    R_{\rm total}(T,\epsilon) \;=\; \max\big(R_{\rm obs}(T,\epsilon),\,R_{\rm int}(T,\epsilon)\big),
    \label{eq:R_total_ballistic}
\end{equation}
and complexity must be bounded using the tighter, i.e. binding, cone, which is quantifiable by the parameters in the system.

\noindent \textbf{Slow decay.} Finally, if the residual interaction decays slow enough, $\alpha<D$, that its norm is not extensively bounded, the Lieb-Robinson radius growth rate itself becomes system-size dependent~\cite{guo2020signaling},
\begin{equation}
    R_{\epsilon}(T,N) = \Theta\!\left(\exp\!\Bigl(g\,N^{1-\alpha/D}\,T\Bigr)/\text{poly}(\epsilon)\right)
    \label{eq:R_res}
\end{equation}
to make the second term of~\eqref{eq:error-split} at most $\epsilon$. Here, unlike the two cases above, whether the interaction picture helps is not fixed by the decay class alone but crosses over as a function of the coupling $g$:
\begin{proposition}[Threshold system size for interaction-picture advantage]
\label{prop:threshold_N}
Let $R_{\rm total}(T,\epsilon,N):=\max\big(R_{\rm obs}(T,\epsilon),R_{\rm int}(T,\epsilon,N)\big)$. Then
\begin{equation}
    N = \Theta\!\left(
        \left(
            \frac{\log(1/\epsilon) + D\log R_{\rm obs}(T,\epsilon)}{g\,T}
        \right)^{\frac{D}{D-\alpha}}
    \right)
    \label{eq:N_max}
\end{equation}
marks the system size at which $R_{\rm obs}(T,\epsilon)=R_{\rm int}(T,\epsilon,N)$, assuming $T$ is long enough that $R_{\rm obs}(T,\epsilon)<R_{\rm int}(T,\epsilon,N)$, so that the exponential dependence in~\eqref{eq:R_res} dominates.
\end{proposition}
Direct Schr\"odinger-picture truncation of the full $H$ fails for every fixed coupling once $N$ exceeds a threshold that does not improve as the interaction weakens, since the untruncated growth rate is a property of the full-strength interaction and carries no small parameter to tune. Proposition~\ref{prop:threshold_N} shows that passing to the interaction picture pushes this threshold out to $N_{\max}(T,\epsilon,g)=\Theta\big(g^{-D/(D-\alpha)}\big)$, which diverges as $g\to0$: the interaction picture converts an $\mathcal O(1)$ obstruction into a $g$-tunable one. For any fixed system size $N$, choosing
\begin{equation}
    g \;\lesssim\; \frac{\log(1/\epsilon)}{N^{1-\alpha/D}\,T}
\end{equation}
suffices to restore $N$-independent complexity. Note that here, unlike the ballistic case above, the advantage comes entirely from the weak coupling $g$, not from the observable-truncation light cone $R_{\rm obs}$ itself: an interaction picture with a weak, slowly decaying residual interaction and a localized dressed observable can beat direct Lieb-Robinson truncation of the full interacting Hamiltonian only because $g$ can be tuned small.

\section{Numerical methods and supplementary results}
\label{app:numerical_exploration}

This appendix specifies the numerical evaluation and radius-selection rules used in Sec.~\ref{sec:numerical}, followed by robustness checks and additional model examples.

\subsection{Numerical methods and region selection}
\label{app:num_methods}

Free-fermion densities are evaluated from the one-particle propagator,
$f(t)=\sum_{j\in S_0}|(e^{-iht})_{0j}|^2$, where $S_0$ is the initially occupied set. For a spectral decomposition $H=V\operatorname{diag}(E_a)V^\dagger$, with $c=V^\dagger\psi_0$ and $O^E=V^\dagger OV$, the dense output is
\begin{equation}
J_H(O)=\sum_{a,b}c_a^*O^E_{ab}c_b\,
T e^{i(E_a-E_b)T/2}
\operatorname{sinc}\!\left(\frac{(E_a-E_b)T}{2}\right),
\label{eq:num_spectral_integral}
\end{equation}
where $\operatorname{sinc}(x)=\sin(x)/x$ and $\operatorname{sinc}(0)=1$. We use this analytic integration for free systems, bosonic pairs, and the residual-interaction example. Interacting fermions, four-particle bosons, Ising spins, and the Schwinger model use sparse matrix-exponential actions followed by trapezoidal integration. Each full/restricted comparison uses the same solver. The bosonic calculations use the complete fixed-$M$ sector of dimension $\binom{N+M-1}{M}$.

Table~\ref{tab:num-grids} gives the sampled ranges. Numerical integrals use nested coarse and fine grids, with the fine step half the listed coarse step and every target time included as an endpoint. $E_{\mathrm{abs}}$ and the integral of pointwise operator norms use quadrature even when the dense output is evaluated spectrally. Grid refinement and reference-size comparisons provide numerical consistency checks, not rigorous error bounds. For $E_{\mathrm{DO}}$, the coarse--fine check compares the signed integral $J_H-J_{H_R}$ on the two grids before taking its absolute value.

\begin{table}[htbp]
\centering
\small
\caption{Reference sizes and numerical grids. Radii are integers in lattice-spacing units. Retained Hamiltonian coefficients keep their full-system values.}
\label{tab:num-grids}
\begin{tabular}{lllll}
\toprule
Experiment & $N$ & Target times $T$ & Radii $R$ & Coarse step \\
\midrule
Free, fast decay & $401,441$ & $0.5,1,\ldots,164$ & $2,\ldots,180$ & $0.01$ \\
Free, slow decay & $513,1025,2049$ & $1,2,4$ & $2,\ldots,192$ & $0.01$ \\
Error hierarchy & $251,301$ & $0.5,1,\ldots,12$ & $2,\ldots,80$ & $0.01$ \\
Interacting fermions & $121,161$ & $0.5,1,\ldots,40$ & $1,\ldots,56$ & $0.005$ \\
Interacting bosons & $21,25$ & $0.5,1,\ldots,6$ & $0,\ldots,7$ & $0.01$ \\
Residual interaction & $41,46$ & $0.5,1,\ldots,12$ & $1,\ldots,19$ & $0.0025$ \\
Ising & $14,16$ & $0.25,0.5,\ldots,2$ & $0,\ldots,5$ & $0.01$ \\
Schwinger & $16,18$ & $0.5,1,\ldots,6$ & $0,\ldots,6$ & $0.01$ \\
\bottomrule
\end{tabular}
\end{table}

For any plotted error $E$, the selected numerical radius is
\begin{equation}
\widehat R_\epsilon(T)=\min\left\{R\in\mathcal R:
E(T,R')\le\epsilon\ \text{for all }R'\in\mathcal R\text{ with }R'\ge R\right\},
\label{eq:num_selected_radius}
\end{equation}
where $\mathcal R$ is the set of tested radii. Requiring all larger tested radii to pass avoids selecting an isolated minimum of an oscillatory error.

For the equal-error curves, we interpolate $\log_{10}E$ between neighbouring radii and use the outermost crossing. Missing crossings are not extrapolated; triangles mark unresolved thresholds. A cross is shown only when the corresponding coarse--fine change exceeds $0.1\epsilon$ near or beyond the selected boundary. The curves use interpolation only; no growth-law fit is imposed.

\subsection{Supplementary robustness checks}

\subsubsection{Free-system robustness checks}
\label{app:num_free_supp}

Figure~\ref{fig:num-app-free-fast}(a) compares the $10^{-7}$ boundary at $N=401$ and $441$, with an additional $N=881$ reference evaluated at eight late target times using the same local integrals. Strict-threshold crossings remain size-sensitive: at $T=130$, the selected integer radius changes from $146$ at $N=441$ to $156$ at $N=881$. These are finite-reference results, not a converged infinite-chain boundary. Increasing the hopping exponent from $\alpha=3$ to $8$ gives a more nearly linear boundary and a smaller early- and intermediate-time radius [panel (b)]. The exponent $\alpha=8$ also satisfies $\alpha-D>5$ in Corollary~\ref{cor:flc-radius}. No asymptotic exponent is fitted.

\begin{figure}[htbp]
\centering
\includegraphics[width=\linewidth]{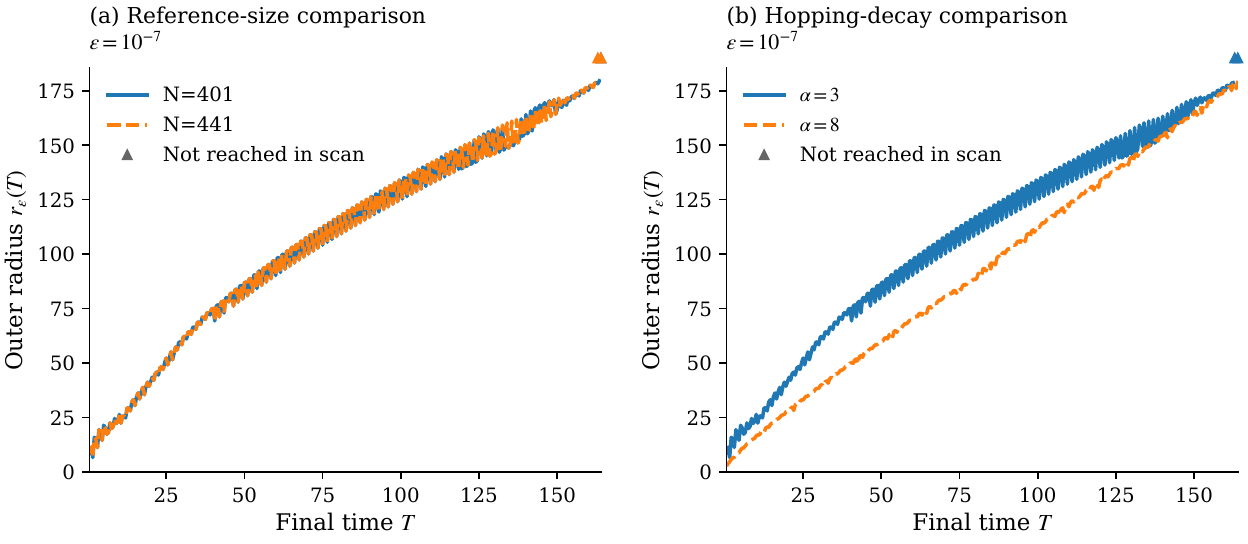}
\caption{Single-particle free-system checks at $\epsilon=10^{-7}$. (a) Reference sizes $N=401,441$ at $\alpha=3$, with open diamonds for $N=881$ at $T=100,120,130,134,140,150,160,164$; unresolved points are shown as triangles rather than diamonds. (b) Hopping exponents $\alpha=3,8$ at $N=441$. Curves connect interpolated outermost crossings; no line is drawn between the $N=881$ spot checks.}
\label{fig:num-app-free-fast}
\end{figure}

The four-particle slow-decay results in Fig.~\ref{fig:num-app-free-slow} retain a clear separation between $E_{\mathrm{DO}}$ and $E_{\mathrm{abs}}$. At $T=4$ and $10^{-3}$, a few hundred sites suffice for the $2049$-site reference. The selected region for $\alpha=0.75$ still increases slightly from $N=1025$ to $2049$, while the $\alpha=1$ region is nearly unchanged. The corresponding local norms vary much less than the full-chain norms. These norms are evaluated on $\mathcal H_4$ using Eq.~\eqref{eq:num_sector_norm_v3}.

\begin{figure}[htbp]
\centering
\includegraphics[width=0.90\linewidth]{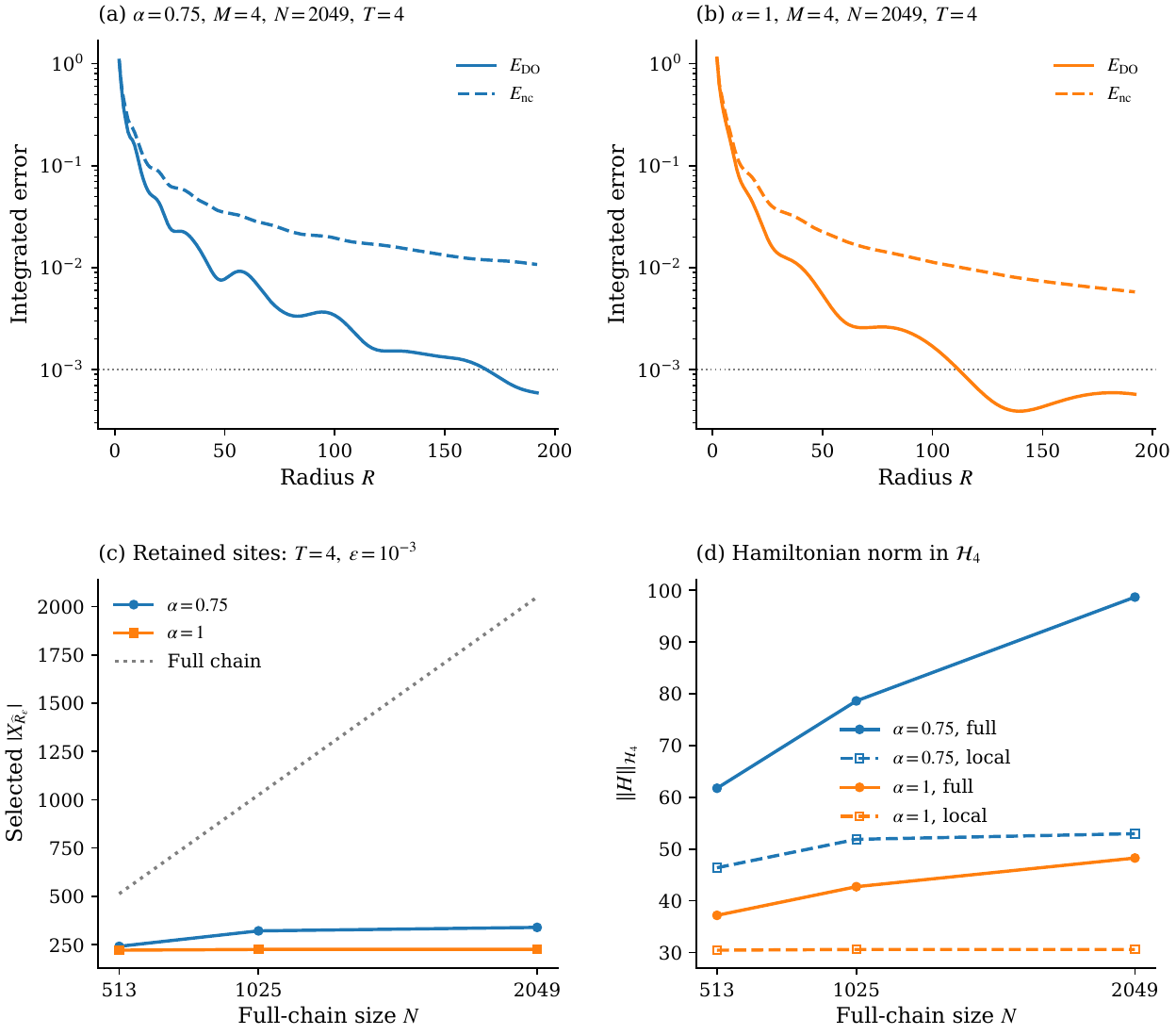}
\caption{Four-particle counterpart of Fig.~\ref{fig:num-free-slow}, with occupied sites $\{-2,-1,0,1\}$. (a,b) Error profiles at $N=2049$ and $T=4$. (c) Selected site counts at $\epsilon=10^{-3}$ as $N$ increases. (d) Full and selected local Hamiltonian norms on the complete four-fermion sector.}
\label{fig:num-app-free-slow}
\end{figure}

\subsubsection{Operator-level comparison}
\label{app:num_hierarchy}

Embed the restricted Hamiltonian in the full chain, acting trivially on its exterior, and define
\begin{equation}
D_R(t)=e^{itH}Oe^{-itH}-e^{itH_R}Oe^{-itH_R},
\qquad K_R(T)=\int_0^T D_R(t)\,\dd t.
\label{eq:num_operator_difference}
\end{equation}
The four error requirements satisfy
\begin{align}
E_{\mathrm{DO}}&\le\|K_R(T)\|_{\mathcal H_M}
 \le\int_0^T\|D_R(t)\|_{\mathcal H_M}\,\dd t,
\nonumber\\
E_{\mathrm{DO}}&\le E_{\mathrm{abs}}
 \le\int_0^T\|D_R(t)\|_{\mathcal H_M}\,\dd t.
\label{eq:num_hierarchy_v3}
\end{align}
There is no general ordering between $E_{\mathrm{abs}}$ and $\|K_R(T)\|_{\mathcal H_M}$. The operator norms control all initial states in the full-chain $M$-particle sector.

For a quadratic Hermitian operator $A=\sum_{ij}a_{ij}c_i^\dagger c_j$, with one-particle eigenvalues $\lambda_1\le\cdots\le\lambda_N$, its exact fixed-sector norm is
\begin{equation}
\|A\|_{\mathcal H_M}=\max\left\{
\left|\sum_{j=1}^{M}\lambda_j\right|,
\left|\sum_{j=N-M+1}^{N}\lambda_j\right|\right\}.
\label{eq:num_sector_norm_v3}
\end{equation}
We apply this formula to $H$ and the spectrally integrated $K_R$. The one-particle kernel of $D_R(t)$ is a difference of two normalized rank-one projectors with nonzero eigenvalues $\pm s(t)$, so its norm is $s(t)$ on every sector $1\le M<N$.

Figure~\ref{fig:num-app-hierarchy}(a) shows more visible time cancellation for one particle than for four particles. Panel (b) compares the operator quantities at $N=251$ and $301$ at $T=8$. Over all saved times through $T=12$ and all tested radii, their changes are at most $8.1\times10^{-6}$ and $2.7\times10^{-5}$, respectively; at $\epsilon=10^{-3}$, the corresponding selected radii change by zero and at most one lattice spacing.

\begin{figure}[htbp]
\centering
\includegraphics[width=\linewidth]{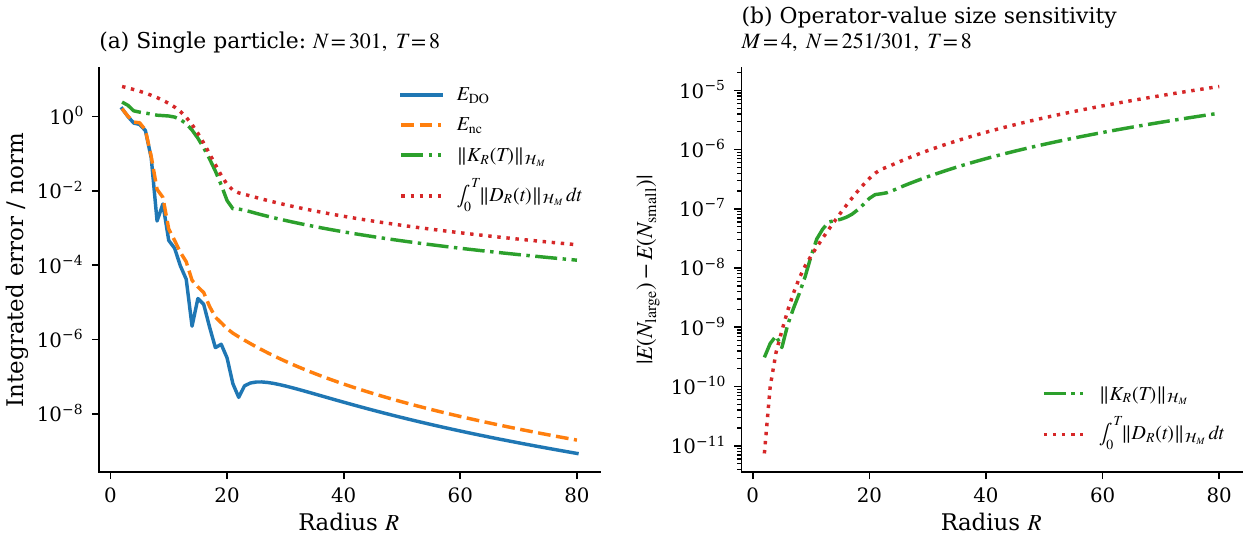}
\caption{Supplementary error-hierarchy checks at $\alpha=3$ and $T=8$. (a) All four error criteria for one particle at $N=301$. (b) Absolute changes of the two four-particle operator quantities between $N=251$ and $301$.}
\label{fig:num-app-hierarchy}
\end{figure}

\subsubsection{Interacting-system robustness checks}
\label{app:num_interacting_supp}

\paragraph{Fermionic interaction tails.}
At fixed $V_0=2$, we compare $W(r)=V_0r^{-\alpha_V}$ with $\alpha_V=3,1.5$ and $W(r)=V_0e^{-(r-1)}$, matching the nearest-neighbour interaction. Figure~\ref{fig:num-app-fermions} shows a larger resolved error tail for the more slowly decaying power law and a modestly larger selected radius at $10^{-5}$. The farthest tail approaches the numerical floor, so we use the resolved points only.

For all three profiles, the coarse--fine change of the signed integral is below $1.4\times10^{-8}$ near and beyond the $10^{-5}$ boundary, and the selected radii are unchanged.

\begin{figure}[htbp]
\centering
\includegraphics[width=\linewidth]{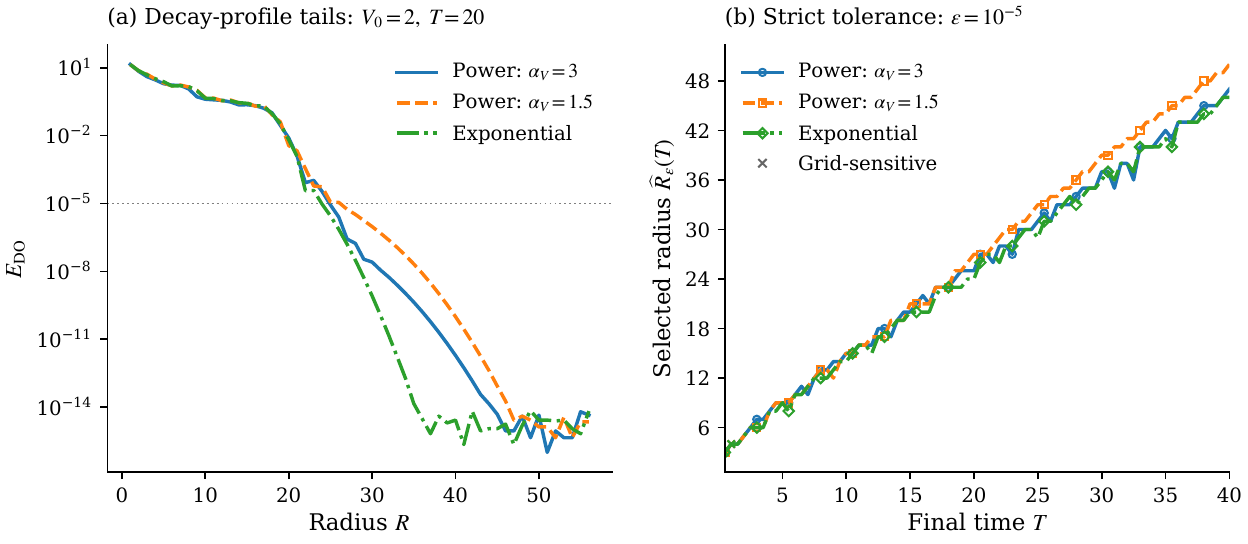}
\caption{Interaction-decay comparison for two fermions at $N=161$ and $V_0=2$. (a) Dense-output errors at $T=20$, with a $10^{-5}$ reference line. (b) Selected radii at that tolerance; the coarse and fine grids select the same radii.}
\label{fig:num-app-fermions}
\end{figure}

\paragraph{Bosonic occupations and decay.}
\label{app:num_boson_supp}

The nonzero probability of $n_0\ge3$ in Fig.~\ref{fig:num-app-bosons}(a) shows that higher local occupations participate in the four-particle dynamics. Varying $\alpha=5,6,8$ changes both hopping and density-interaction tails; the selected regions remain comparable at $10^{-3}$ [panel (b)].

Theorem~\ref{thm:LRZ} applies when $R\ge\max(2,\operatorname{diam}X_0)$ and $T<R/v_{\rm LR}$, with the density and particle-free-shell assumptions satisfied here uniformly in $N$. The numerical scan also includes points outside this analytical validity window.

For the four-particle cases with $U=4$, $V_0=1$, and $\alpha=5,6,8$, the coarse--fine change of the signed integral is below $2.1\times10^{-6}$ near and beyond the $10^{-3}$ boundary, with unchanged selected radii.

\begin{figure}[htbp]
\centering
\includegraphics[width=\linewidth]{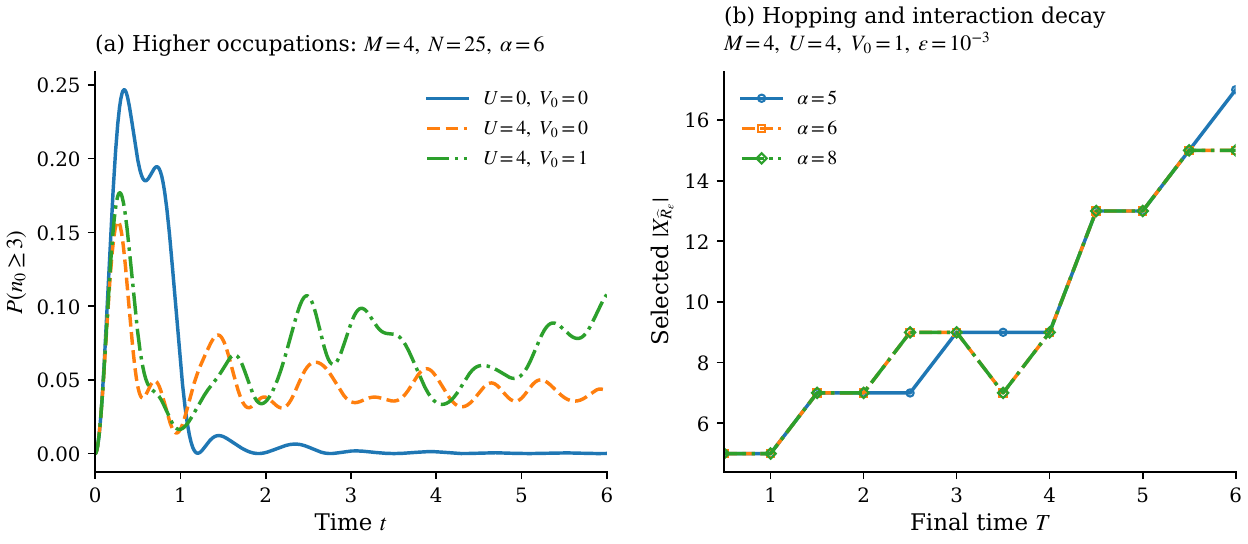}
\caption{Four-particle bosonic checks at $N=25$. (a) Probability of $n_0\ge3$ at $\alpha=6$ for the free, on-site-only, and long-range interacting models. (b) Selected site counts for $U=4$, $V_0=1$, and $\epsilon=10^{-3}$ as the common hopping and interaction exponent varies.}
\label{fig:num-app-bosons}
\end{figure}

\subsection{Additional numerical examples}

\subsubsection{Weak residual interactions with global hopping}
\label{subsec:num_residual_exploration}

To isolate the weak-interaction mechanism discussed in Appendix~\ref{app:int_pic}, we set $W(r)=gr^{-3}$ in Eq.~\eqref{eq:num_tV}, write $H=H_{\mathrm{hop}}+gV$, and compare
\begin{equation}
H_R^{\mathrm{direct}}=H_{\mathrm{hop},R}+gV_R,
\qquad
H_R^{\mathrm{res}}=H_{\mathrm{hop}}+gV_R.
\label{eq:num_residual_two_approximations}
\end{equation}
The initial state remains the adjacent fermion pair and $O=n_0$. Only interaction pairs wholly inside $X_R$ enter $V_R$; the residual-only calculation keeps the hopping background global.

At $T=12$, $g=0.2$, and $10^{-3}$, the selected radii are $14$ for direct truncation and $7$ for residual-only truncation [Fig.~\ref{fig:num-app-residual}]. For $g=0.05$, residual-only truncation passes at the smallest tested radius $R=1$ throughout the sampled times. The smaller residual-only radius measures the truncated interaction region, while the hopping background remains global.

\begin{figure}[htbp]
\centering
\includegraphics[width=\linewidth]{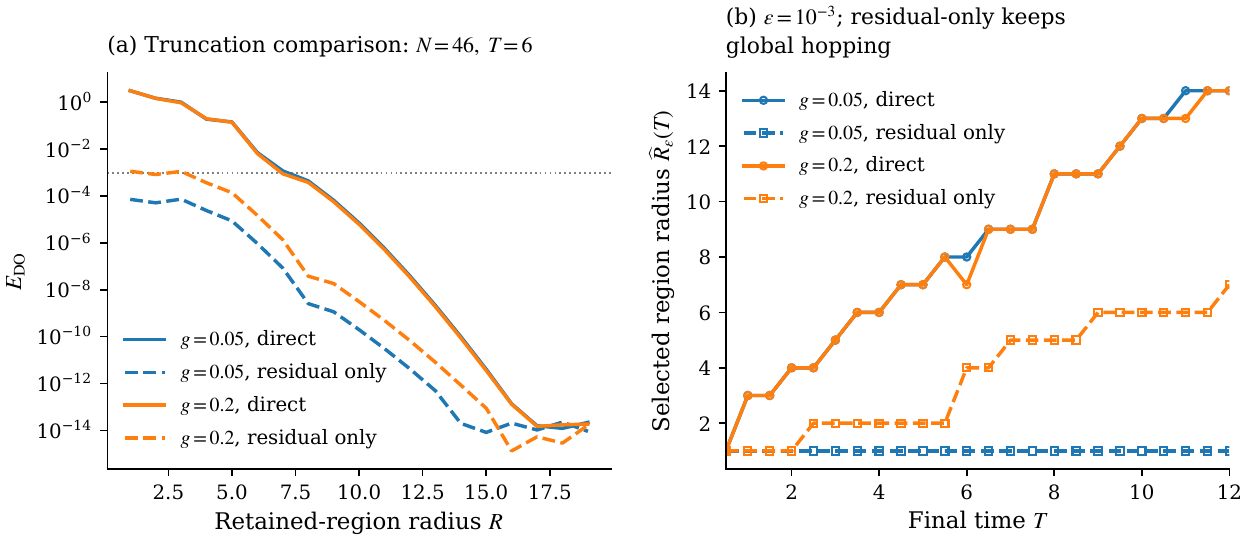}
\caption{Direct and residual-only truncation at $N=46$. (a) Dense-output errors at $T=6$ for $g=0.05,0.2$. (b) Selected radii at $\epsilon=10^{-3}$.}
\label{fig:num-app-residual}
\end{figure}

\subsubsection{Local dense outputs in the Ising chain}
\label{subsec:num_ising_exploration}

We use the static Ising Hamiltonian~\eqref{eq:longrange_ising} with $\alpha=3$, $h_i=0.7$, initial state $|+\rangle^{\otimes N}$, and observable $O=\sigma_0^x$. The periodic distance is $d_N(i,j)=\min(|i-j|,N-|i-j|)$. Retained couplings keep their full-ring values, and the local region is treated as an open subregion of the original ring. The calculation uses the full spin Hilbert space.

For $N=16$, the eleven-spin region $R=5$ reproduces the initial decay and revival of the transverse spin [Fig.~\ref{fig:num-app-ising}]. Its dense-output error is below $10^{-4}$ at $T=1$ and reaches several $10^{-3}$ at $T=2$. The error panel therefore gives a more precise accuracy window than visual agreement of the signals.

\begin{figure}[htbp]
\centering
\includegraphics[width=\linewidth]{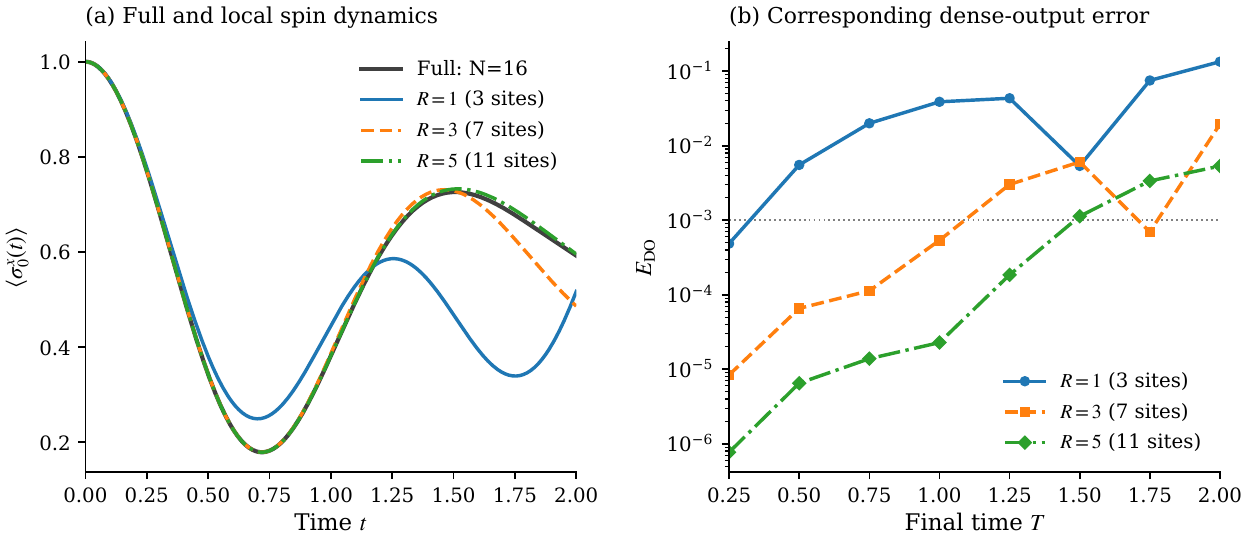}
\caption{Ising dynamics with $N=16$, $\alpha=3$, and transverse field $0.7$. (a) Full-chain and local transverse-spin signals for $R=1,3,5$. (b) Corresponding dense-output errors, with a $10^{-3}$ reference line.}
\label{fig:num-app-ising}
\end{figure}

\subsubsection{Spatial and electric-field truncation in the Schwinger model}
\label{subsec:num_schwinger_exploration}

For the Schwinger Hamiltonian~\eqref{eq:sch_hamiltonian}, we choose $x=1$, $\mu=0.5$, and $O=n_0$ at an even site. The initial state is the staggered vacuum, $n_r=[1-(-1)^r]/2$, with zero electric fields. Calculations are restricted to the physical Gauss-law sector,
$E_r-E_{r-1}=n_r-[1-(-1)^r]/2$, with zero boundary flux and a hard cutoff $|E_r|\le\Lambda_E$. Local intervals keep the global staggering and fix both cut fluxes to zero; transitions beyond the electric cutoff are discarded.

Figure~\ref{fig:num-app-schwinger} separates spatial and electric-field truncation. At $N=18$ and $T=3$, the spatial errors fall below $10^{-3}$ at $R=6$ for all three cutoffs. The full-chain integrals at $\Lambda_E=1$ and $3$ differ by order $10^{-3}$ at several later times, while the stored outputs at $2$ and $3$ differ by about $10^{-11}$ or less. The change from cutoff $2$ to $3$ is therefore much smaller than the spatial truncation error on this finite chain; this comparison does not establish accuracy relative to the model without an electric cutoff.

\begin{figure}[htbp]
\centering
\includegraphics[width=\linewidth]{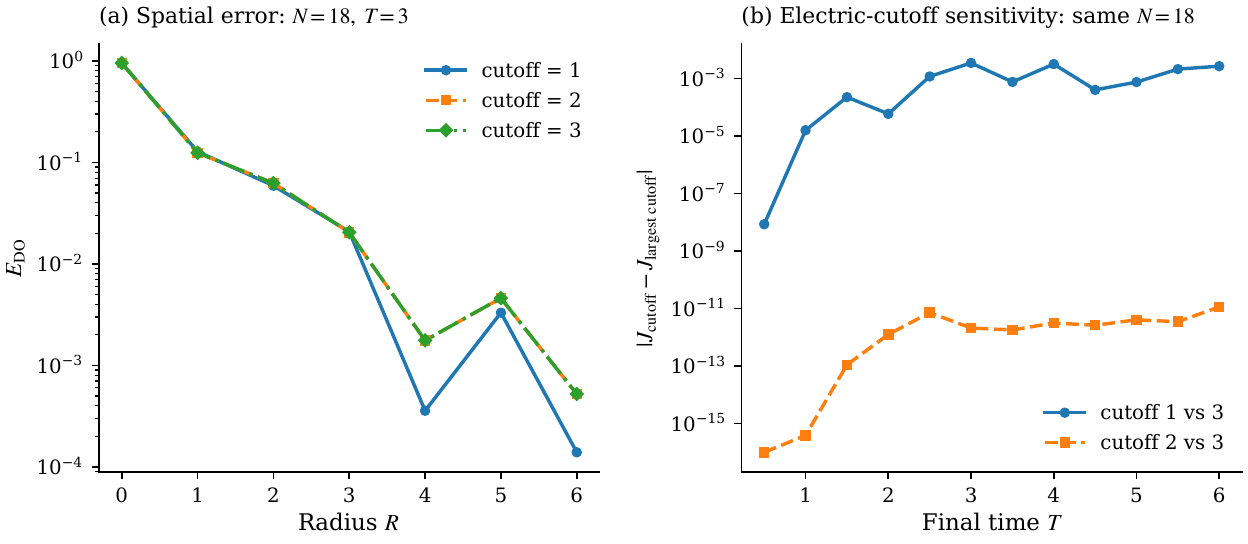}
\caption{Spatial and electric-field truncation in the Schwinger model at $N=18$, $x=1$, and $\mu=0.5$. (a) Full-minus-local dense-output errors at $T=3$, evaluated separately at each electric cutoff. (b) Differences between full-chain integrals at different cutoffs.}
\label{fig:num-app-schwinger}
\end{figure}

\section{Resource complexity of local dense-output estimation}
\label{app:do-resources}

The spatial bounds derived above determine how large a region is sufficient for approximating the dense output. 
Here we translate such a bound into a query-complexity guarantee. 
The procedure has two steps: choose a region $X_R$ for which replacing $H$ by $H_R$ introduces a controlled error, and estimate the local integral $J_{H_R}(O)$. 
We use the same estimator for the full and local problems, so the comparison isolates the saving due to spatial truncation.

The main resource statement is Theorem~\ref{thm:do-local-cost}. We first
separate spatial and estimation errors and prove the coherent estimator.
We then give compatible block encodings for two lattice settings, followed
by finite-precision accounting and an optional error-budget refinement.

Throughout this section, $H$ and $O$ are time independent and $\|O\|\leq1$.
Our target remains the dense output $J_H(O)$ defined in
Eq.~\eqref{eq:dense_output}. The spatial truncation bounds themselves may
also apply to time-dependent dynamics under the corresponding assumptions.
The resource analysis below is restricted to time-independent $H$ and $O$;
the step where time independence is used is identified in the estimator
proof.

\subsection{Local approximation, error budget, and input access}
\label{subsec:do-access}

Let $X_R$ be the radius-$R$ neighbourhood of the support $X_0$ of $O$, and
let $H_R$ contain the interaction terms supported entirely in $X_R$, with
their original coefficients. Let $\rho_R$ be the restriction of the initial
state $\rho_0$ to this region. In a tensor-product representation,
$\rho_R=\operatorname{Tr}_{\Lambda\setminus X_R}\rho_0$; for fermions we
use the restriction to the local even algebra. Recalling
Eq.~\eqref{eq:trun_dense_output}, in the time-independent setting
\begin{equation}
 J_{H_R}(O)=\int_0^T
 \operatorname{Tr}\!\left[\rho_R e^{iH_Rt}Oe^{-iH_Rt}\right] \,\mathrm{d}t.
 \label{eq:do-local-functional}
\end{equation}

Suppose the spatial analysis gives
\begin{equation}
 |J_H(O)-J_{H_R}(O)|\leq \epsilon_{\rm trunc}(R,T),
 \label{eq:do-certificate}
\end{equation}
for the chosen initial state, or for a state class containing it. We set
$\epsilon_{\rm trunc}(R,T)=0$ when $X_R$ is the full lattice. If an estimate $\widehat J$ of
the local integral has error at most $\epsilon_{\mathrm{est}}$, then
\begin{equation}
 |\widehat J-J_H(O)|
 \leq |\widehat J-J_{H_R}(O)|+|J_{H_R}(O)-J_H(O)|
 \leq \epsilon_{\mathrm{est}}+\epsilon_{\rm trunc}(R,T).
 \label{eq:do-error-split}
\end{equation}
Thus we require $\epsilon_{\mathrm{est}}+\epsilon_{\rm trunc}(R,T)\leq\epsilon$. We first use
$\epsilon_{\rm trunc}(R,T)\leq\epsilon/2$ and $\epsilon_{\mathrm{est}}=\epsilon/2$.

We use the block-encoding access model of Sec.~\ref{sec:query-model} and
Eq.~\eqref{eq:block-encoding}. For $H_R$, let $U_{H_R}$ have normalization
$\alpha_H(R)$:
\begin{equation}
 (\langle0^{m_H}|\otimes I)U_{H_R}(|0^{m_H}\rangle\otimes I)
 =\frac{H_R}{\alpha_H(R)},
 \qquad \alpha_H(R)\geq\|H_R\|.
 \label{eq:do-block-encoding}
\end{equation}
We also assume a unitary $U_O$ that block encodes $O$ with unit normalization,
\begin{equation}
 (\langle0^{m_O}|\otimes I)U_O(|0^{m_O}\rangle\otimes I)=O,
 \label{eq:do-observable-block}
\end{equation}
and a state-preparation unitary $U_{\rho_R}$ that prepares a purification
$|\Phi_R\rangle$ of $\rho_R$.

Following Sec.~\ref{sec:query-model}, controlled and inverse uses are counted
as queries. We denote the Hamiltonian-block, state-preparation, and
observable-block query counts by $Q_H$, $Q_\psi$, and $Q_O$. For a fair
full-versus-local comparison, the two Hamiltonian block encodings use the
same underlying coefficient or local-term access. A smaller local
normalization must therefore come from an explicit local construction. If a
local preparation of $\rho_R$ is unavailable, the cost of the supplied
global preparation must instead be included.

\subsection{Coherent dense-output estimation}
\label{subsec:do-transfer}
\label{app:do-readout}

We estimate the integral by placing midpoint evolution times in coherent superposition. 
The observable expectation on this state is the time-grid average, and a Hadamard-test construction maps that expectation to a probability $p$ for which $T(2p-1)$ is the corresponding quadrature approximation. 
Amplitude estimation then estimates $p$. 
The proof separates the total estimation error into quadrature, simulation, and probability-estimation errors.

\begin{lemma}[Coherent dense-output estimation]
\label{lem:do-readout}
Let $0<\epsilon_{\mathrm{est}}<T$ and $0<\delta<1/3$. Suppose the Hamiltonian, initial state,
and observable are supplied through the access unitaries described above, with
Hamiltonian normalization $\alpha_H$. The dense output can be
estimated to additive error $\epsilon_{\mathrm{est}}$ with probability at least $1-\delta$
using
\begin{align}
 Q_H&=\widetilde{\mathcal O}\!\left(
       \frac{(1+\alpha_H T)T}{\epsilon_{\mathrm{est}}}
       \log\frac1\delta\right),
       \label{eq:do-readout-H}\\
 Q_\psi,\ Q_O&=\mathcal O\!\left(
       \frac{T}{\epsilon_{\mathrm{est}}}\log\frac1\delta\right).
       \label{eq:do-readout-other}
\end{align}
Here $Q_H$, $Q_\psi$, and $Q_O$ count Hamiltonian-block, state-preparation,
and observable-block queries, including controlled and inverse calls.
\end{lemma}

\subsection{Proof of Lemma~\ref{lem:do-readout}}

For $f(t)=\langle\Phi|e^{iHt}Oe^{-iHt}|\Phi\rangle$,
\begin{equation}
 |f'(t)|
 =\left|\langle\Phi|e^{iHt}i[H,O]e^{-iHt}|\Phi\rangle\right|
 \leq2\alpha_H.
 \label{eq:do-Lipschitz}
\end{equation}
Let $K=2^m$ be the smallest power of two satisfying
\begin{equation}
 K\geq\max\!\left\{1,\frac{2\alpha_HT^2}{\epsilon_{\mathrm{est}}}\right\},
 \qquad t_k=\frac{(k+1/2)T}{K},\quad 0\leq k<K.
 \label{eq:do-clock-grid}
\end{equation}
Writing $J^{(K)}=(T/K)\sum_k f(t_k)$, the midpoint error obeys
\begin{align}
 |J_H(O)-J^{(K)}|
 &\leq 2\alpha_H
   \sum_{k=0}^{K-1}\int_{kT/K}^{(k+1)T/K}|t-t_k|\,\mathrm{d}t
   \nonumber\\
 &=\frac{\alpha_HT^2}{2K}\leq\frac{\epsilon_{\mathrm{est}}}{4}.
 \label{eq:do-quadrature-error}
\end{align}
This first-derivative estimate requires no analyticity assumption, and increasing
$K$ changes the clock width only logarithmically. Now, we prepare a uniform clock using $m$ Hadamard gates and apply
\begin{equation}
 W=\sum_{k=0}^{K-1}|k\rangle\langle k|\otimes e^{-iHt_k}.
 \label{eq:do-controlled-propagator}
\end{equation}
For $k=\sum_{j=0}^{m-1}k_j2^j$, this is implemented by an unconditional
evolution for $T/(2K)$ and evolutions for $2^jT/K$ controlled on $k_j$.
The sum of all scheduled simulation times is
\begin{equation}
 \frac{T}{2K}+\sum_{j=0}^{m-1}\frac{2^jT}{K}
 =T-\frac{T}{2K}<T.
 \label{eq:do-scheduled-time}
\end{equation}
This binary decomposition is where time independence is used. A time-dependent
extension would require coherent controlled time-ordered propagators.

Block-encoded Hamiltonian simulation provides a coherent, clean-ancilla
implementation of evolution for duration $s$, to isometry error $\xi$,
using $\mathcal O(\alpha_H|s|+\log(1/\xi))$ oracle calls;
this convenient precision bound follows from qubitization or quantum
singular value transformation~\cite{low2019hamiltonian,gilyen2019qsvt}.
Using a sufficiently accurate block approximation also bounds any residual
ancilla leakage, without changing this logarithmic dependence.
Set
\begin{equation}
 \zeta=\frac{\epsilon_{\mathrm{est}}}{8T}
 \label{eq:do-state-simulation-error}
\end{equation}
and implement each of the $m+1$ scheduled evolutions to isometry error at
most $\zeta/(m+1)$. A telescoping estimate, applied on the ideal
clean-ancilla subspace at each step, bounds the final state-vector error
by $\zeta$. The Hamiltonian query cost for one preparation of this time-superposition state is
\begin{equation}
 \mathcal O\!\left[
  \alpha_HT+(m+1)
   \log\!\left(\frac{8(m+1)T}{\epsilon_{\mathrm{est}}}\right)\right].
 \label{eq:do-cost-per-readout}
\end{equation}
Controlled and inverse versions have the same asymptotic query cost in
our access model. Ideally the resulting state is
\begin{equation}
 |\Psi\rangle=\frac{1}{\sqrt K}
   \sum_{k=0}^{K-1}|k\rangle e^{-iHt_k}|\Phi\rangle.
 \label{eq:do-history-state}
\end{equation}
Only one state-preparation call is needed per preparation of this state. Append the zero ancillas for the block encoding $U_O$ of $O$. A Hadamard
test with controlled $U_O$, acting on the system but not the clock, has
ideal success probability
\begin{equation}
 p_{\mathrm{ideal}}
 =\frac12\!\left(1+\langle\Psi|I\otimes O|\Psi\rangle\right)
 =\frac12\!\left(1+\frac{J^{(K)}}{T}\right).
 \label{eq:do-mean-probability}
\end{equation}
The Hermiticity of $O$ makes the relevant expectation real. No
postselection on the block-encoding ancillas is performed.
Let $p$ be the probability of the implemented estimation unitary. The state-vector
error bound and $\|O\|\leq1$ imply
\begin{equation}
 \left|T(2p-1)-J^{(K)}\right|\leq2T\zeta\leq\frac{\epsilon_{\mathrm{est}}}{4}.
 \label{eq:do-implemented-bias}
\end{equation}
Amplitude estimation~\cite{brassard2002amplitude}, followed by median
amplification, estimates this implemented probability to
$|\widehat p-p|\leq\epsilon_{\mathrm{est}}/(4T)$ with failure probability at most $\delta$,
using $\mathcal O((T/\epsilon_{\mathrm{est}})\log(1/\delta))$ calls to the estimation unitary
and its inverse. The output $\widehat J=T(2\widehat p-1)$ therefore obeys
\begin{equation}
 |\widehat J-J_H(O)|
 \leq 2T|\widehat p-p|
      +|T(2p-1)-J^{(K)}|
      +|J^{(K)}-J_H(O)|
 \leq\epsilon_{\mathrm{est}}.
 \label{eq:do-readout-total-error}
\end{equation}
Multiplying Eq.~\eqref{eq:do-cost-per-readout} by the amplitude-estimation
call count proves Eq.~\eqref{eq:do-readout-H}. Each call uses a constant
number of state-preparation and observable-block queries, proving
Eq.~\eqref{eq:do-readout-other}.

The approximate simulator is a fixed unitary and its adjoint supplies the inverse
used by amplitude estimation. Its simulation bias is already included in
Eq.~\eqref{eq:do-implemented-bias}.

Combining this lemma with the spatial error bound gives a guarantee for
the original, full-system integral.
\begin{theorem}[Local dense-output resource bound]
\label{thm:do-local-cost}
Suppose Eq.~\eqref{eq:do-certificate} and the stated input-access assumptions
hold. Let $0<\epsilon<T$ and $0<\delta<1/3$, and choose $R_\epsilon$ such
that $B(R_\epsilon,T)\leq\epsilon/2$. There is an estimate $\widehat J$ with
\begin{equation}
 \Pr\!\left[|\widehat J-J_H(O)|\leq\epsilon\right]\geq1-\delta
 \label{eq:do-total-guarantee}
\end{equation}
using
\begin{align}
 Q_H&=\widetilde{\mathcal O}\!\left[
   \frac{\bigl(1+\alpha_H(R_\epsilon)T\bigr)T}{\epsilon}
   \log\frac1\delta\right],
   \label{eq:do-local-cost-H}\\
 Q_\psi,\ Q_O&=\mathcal O\!\left(
   \frac{T}{\epsilon}\log\frac1\delta\right).
   \label{eq:do-local-cost-other}
\end{align}
\end{theorem}
\begin{proof}
Apply Lemma~\ref{lem:do-readout} to $H_{R_\epsilon}$ with
$\epsilon_{\mathrm{est}}=\epsilon/2$. On the estimator's success event,
Eq.~\eqref{eq:do-error-split} bounds the total error by $\epsilon$.
The spatial bound is deterministic and adds no failure probability.
\end{proof}

This estimator differs from the $T^3/\epsilon$ baseline in
Eq.~\eqref{eq:liu_cost}. Here the same estimator is applied to the full and
local problems, so the spatial saving is the change from
$\alpha_H(\Lambda)$ to $\alpha_H(R_\epsilon)$.

\subsection{Applications and compatible encodings}
\label{subsec:do-encodings}
\label{app:do-encodings}

To apply Theorem~\ref{thm:do-local-cost}, we need a spatial bound that fixes
$R_\epsilon$ and a compatible block encoding that fixes $\alpha_H(R)$. We
keep these two ingredients together for each model. Both constructions use
the standard linear-combination rule~\cite{gilyen2019qsvt}: if $V_\ell$
block encodes a contraction $A_\ell$ with unit normalization and
\begin{equation}
 \operatorname{PREP}|0\rangle=(\sum_\ell w_\ell)^{-1/2}
 \sum_\ell\sqrt{w_\ell}|\ell\rangle,
\end{equation}
then
\begin{equation}
 (\operatorname{PREP}^{\dagger}\otimes I)
 \left(\sum_\ell |\ell\rangle\langle\ell|\otimes V_\ell\right)
 (\operatorname{PREP}\otimes I)
 \label{eq:do-LCU-unitary}
\end{equation}
block encodes $\sum_\ell w_\ell A_\ell$ with normalization
$\sum_\ell w_\ell$.

\subsection{Finite-range many-body systems}
\label{app:do-finite-range}

Consider a $D$-dimensional lattice with fixed local dimension, bounded
coordination, and a fixed number of finite-range interaction templates per
site, each of norm at most $J$. The compatible local-term construction gives
\begin{equation}
 \alpha_H(R)=\Theta(J|X_R|),
 \qquad
 \alpha_H(\Lambda)=\Theta(JN).
 \label{eq:do-finite-range-normalization}
\end{equation}

Index each allowed interaction template by an anchor site and one of a
fixed number of template labels. Each physical term is assigned to a
unique slot. For a region $X_R$, use a power-of-two number of slots
$m_R=\Theta(|X_R|)$ and pad unused slots with the zero operator. Terms
whose support leaves $X_R$ also occupy zero slots. Every nonzero slot has
norm at most $J$, so the normalized slot operator has a
unit-normalization block encoding. Uniform PREP uses only Hadamard gates,
and Eq.~\eqref{eq:do-LCU-unitary} gives normalization $Jm_R$.

For fixed local dimension and fixed interaction range, a coherent
local-term value oracle supplies a constant-size matrix using a constant
number of data queries. A unitary dilation of this normalized matrix
provides the slot block encoding. Computing the dilation and addressing the site registers contribute to the gate
cost but not to the stated number of data queries. The same construction on all $N$ sites has
$m_\Lambda=\Theta(N)$ slots. This proves
Eq.~\eqref{eq:do-finite-range-normalization} and the compatibility of the two encodings. For the actual restricted dynamics, a finite-range Lieb--Robinson estimate
gives
\begin{equation}
 \epsilon_{\rm trunc}(R,T)\leq CJT^2e^{-\mu(R-v_{\mathrm{LR}}T)}
 \label{eq:do-finite-range-error}
\end{equation}
above a fixed microscopic radius, with constants independent of $N$.

Let $H=\sum_Z h_Z$ have interaction range
$a$, with $\|h_Z\|\leq J$, bounded support sizes, and bounded coordination.
The support of $O$ is a fixed region $X_0$. Use regular regions $X_R$ for
which the number of interaction terms crossing the boundary is at most
$C_0(1+R)^{D-1}$. All constants below are independent of $N$ and of the
initial state.

Extend $H_R$ by the identity action on exterior registers and write
$\tau_t^H(O)=e^{iHt}Oe^{-iHt}$. Differentiating
$\tau_s^H(\tau_{t-s}^{H_R}(O))$ gives the exact identity
\begin{equation}
 \tau_t^H(O)-\tau_t^{H_R}(O)
 =i\int_0^t\tau_s^H\!\left(
   [H-H_R,\tau_{t-s}^{H_R}(O)]\right)\,\mathrm{d}s.
 \label{eq:do-boundary-Duhamel}
\end{equation}
Every term of $H-H_R$ supported entirely outside $X_R$ commutes with
$\tau_u^{H_R}(O)$. Thus only terms crossing the boundary remain, rather
than an extensive sum over the entire discarded region.

A finite-range Lieb--Robinson estimate, uniform for the restricted
Hamiltonians, has the form~\cite{nachtergaele2019quasi}
\begin{equation}
 \|[h_Z,\tau_u^{H_R}(O)]\|
 \leq C_1 J
 e^{-\mu_0(d(X_0,Z)-v_0u)}.
 \label{eq:do-boundary-LR}
\end{equation}
The support-size factors have been absorbed into $C_1$. A boundary-crossing
term has $d(X_0,Z)\geq R-a$ above a fixed microscopic radius. By unitarity,
Eq.~\eqref{eq:do-boundary-Duhamel} and then time integration give
\begin{align}
 |J_H(O)-J_{H_R}(O)|
 &\leq C_2J(1+R)^{D-1}e^{-\mu_0R}
       \int_0^T(T-u)e^{\mu_0v_0u}\,\mathrm{d}u
       \nonumber\\
 &\leq C_3JT^2(1+R)^{D-1}
       e^{-\mu_0R+\mu_0v_0T}.
 \label{eq:do-finite-range-raw}
\end{align}
For fixed $D$ and $\mu_0>0$, the polynomial in $R$ is bounded by a constant
times $e^{\mu_0R/2}$. Setting $\mu=\mu_0/2$ and $v_{\mathrm{LR}}=2v_0$ gives
Eq.~\eqref{eq:do-finite-range-error} after redefining the constants.

For example, with $r_0$ a fixed radius above which the preceding geometry
applies, choose
\begin{equation}
 R_\epsilon=
 \left\lceil r_0+v_{\mathrm{LR}}T+
  \frac1\mu\log\!\left(1+\frac{2CJT^2}{\epsilon}\right)
 \right\rceil.
 \label{eq:do-finite-range-explicit-radius}
\end{equation}
If this radius reaches the full lattice, use the full Hamiltonian, whose
spatial error is zero. Otherwise the spatial error is at most $\epsilon/2$. 

Together with $|X_R|=\mathcal O((1+R)^D)$ and the local-term normalization, this gives the following resource bound.

\begin{corollary}[Replacing global volume by local volume]
\label{cor:do-finite-range}
Under these assumptions and the required initial-state access, a sufficient
radius is
\begin{equation}
 R_\epsilon=\mathcal O\!\left(
  1+v_{\mathrm{LR}}T+\log\!\left(1+\frac{JT^2}{\epsilon}\right)\right),
 \label{eq:do-finite-range-radius}
\end{equation}
with the region capped at the full lattice. Let $X_{R_\epsilon}$ denote the
corresponding local region. When $\alpha_H(R_\epsilon)T\geq1$,
\begin{equation}
 Q_{H,\mathrm{full}}=\widetilde{\mathcal O}\!\left(
   \frac{JNT^2}{\epsilon}\log\frac1\delta\right),
 \qquad
 Q_{H,\mathrm{local}}=\widetilde{\mathcal O}\!\left(
   \frac{J|X_{R_\epsilon}|T^2}{\epsilon}\log\frac1\delta\right).
 \label{eq:do-volume-costs}
\end{equation}
\end{corollary}
At fixed $T,\epsilon,J$ and locality constants, the local Hamiltonian-block
query bound is independent of $N$.

\subsection{Free fermions at fixed particle number}
\label{app:do-displacement}

For the fixed-particle-number free-fermion setting considered above, assume
$M$ particles initially lie inside $X_R$ and
$|h_{ii}|\leq J_{\mathrm{diag}}$, $|h_{ij}|\leq J_0|i-j|^{-\alpha}$ for
$i\ne j$. A displacement-based block encoding gives
\begin{equation}
 \alpha_H(R)=M\!\left(J_{\mathrm{diag}}+
  \sum_{0<|\mathbf r|\leq\operatorname{diam}(X_R)}
       J_0|\mathbf r|^{-\alpha}\right).
 \label{eq:do-free-normalization}
\end{equation}

Let $X_R$ be an open box, or a regular finite region of an open lattice
box, and embed its one-particle register into a padded coordinate grid.
The padding is chosen large enough to implement every required displacement
as reversible modular addition. A separate signed-coordinate test,
performed before accepting a hopping, rejects any displacement that leaves
the original region.

Let the displacement labels include $\mathbf r=0$ when
$J_{\mathrm{diag}}>0$ and all nonzero lattice vectors with
$|\mathbf r|\leq\operatorname{diam}(X_R)$. Set
\begin{equation}
 w_{\mathbf 0}=J_{\mathrm{diag}},\qquad
 w_{\mathbf r}=J_0|\mathbf r|^{-\alpha}\quad(\mathbf r\ne0),
 \label{eq:do-envelope-weights}
\end{equation}
omitting zero-weight labels. Define a diagonal contraction $D_{\mathbf r}$
whose entry at $x$ is
\begin{equation}
 b_{\mathbf r}(x)=
 \begin{cases}
  h_{x+\mathbf r,x}/w_{\mathbf r},
    &x,x+\mathbf r\in X_R,\\
  0,&\text{otherwise}.
 \end{cases}
 \label{eq:do-diagonal-coefficient}
\end{equation}
For $\mathbf r=0$, the numerator is $h_{xx}$. Let $S_{\mathbf r}$ denote
reversible displacement on the padded register. Since the boundary test
uses unwrapped coordinates, the one-particle target, extended by zero on
padding states, is exactly
\begin{equation}
 h_R=\sum_{\mathbf r}w_{\mathbf r}S_{\mathbf r}D_{\mathbf r}.
 \label{eq:do-displacement-decomposition}
\end{equation}
A diagonal contraction with entry $b$, $|b|\leq1$, has the one-ancilla
unitary dilation
\begin{equation}
 \begin{pmatrix}
  b&-\sqrt{1-|b|^2}\\
  \sqrt{1-|b|^2}&b^*
 \end{pmatrix}.
 \label{eq:do-diagonal-dilation}
\end{equation}
Compute the coefficient and boundary flag coherently, apply this rotation,
and uncompute the coefficient data. Follow it by $S_{\mathbf r}$.
The resulting unitary block encodes $S_{\mathbf r}D_{\mathbf r}$ with
normalization one, using a constant number of coefficient-oracle queries.
Preparing the displacement register with amplitudes proportional to
$\sqrt{w_{\mathbf r}}$ and using Eq.~\eqref{eq:do-LCU-unitary} therefore
block encodes $h_R$ with normalization $\sum_{\mathbf r}w_{\mathbf r}$.

Preparing the known envelope weights contributes preprocessing and gate overhead;
the primitive-query comparison uses a constant number of coefficient-oracle
queries per block.

On the antisymmetric $M$-particle subspace,
\begin{equation}
 H_R=\left.\sum_{a=1}^{M}h_R^{(a)}\right|_{\wedge^M\mathbb C^{|X_R|}}.
 \label{eq:do-first-quantized-H}
\end{equation}
A uniform selection of the particle register contributes a factor $M$ to
the normalization, proving Eq.~\eqref{eq:do-free-normalization}. The exact Hamiltonian evolution preserves the antisymmetric subspace, and the initial particles in this example lie inside the local region.

For the single-site density $n_q$, a unit-normalization observable encoding
can use the predicate that at least one particle register equals $q$.
On the antisymmetric subspace this predicate is exactly $n_q$ and remains a bounded projector on the computational register.

Finally, the lattice shell count gives
\begin{equation}
 \sum_{0<|\mathbf r|\leq L}|\mathbf r|^{-\alpha}
 \asymp\sum_{r=1}^{\lfloor L\rfloor}r^{D-1-\alpha}
 =
 \begin{cases}
  \Theta(1),&\alpha>D,\\
  \Theta(\log(2+L)),&\alpha=D,\\
  \Theta(L^{D-\alpha}),&0<\alpha<D.
 \end{cases}
 \label{eq:do-shell-sum}
\end{equation}

For regular regions, Eq.~\eqref{eq:do-shell-sum} gives
\begin{equation}
 \alpha_H(R)=
 \begin{cases}
  \Theta(1),&\alpha>D,\\
  \Theta(\log(2+R)),&\alpha=D,\\
  \Theta(R^{D-\alpha}),&0<\alpha<D,
 \end{cases}
 \label{eq:do-free-regimes}
\end{equation}
with fixed particle number and energy scales absorbed into the constants.
For the full lattice, replace $R$ by $L_N=\Theta(N^{1/D})$.

This is the resource counterpart of the slow-decay free-fermion truncation
analysis above; see Theorem~\ref{thm:free_fermion_main}. For the explicit
certificate used here, Theorem~\ref{thm:slow-truncation} gives, for any
$0<s<1$ with $s\leq\alpha-D/2$,
\begin{equation}
 \epsilon_{\rm trunc}(R,T)\leq C J_0T^2R^{-s}.
 \label{eq:do-slow-application-error}
\end{equation}
Thus
$R\geq\max\{1,(2CJ_0T^2/\epsilon)^{1/s}\}$ controls the spatial error by $\epsilon/2$. 
When this radius is independent of $N$, Eq.~\eqref{eq:do-free-regimes} replaces the full normalization $\Theta(N^{1-\alpha/D})$ by $\Theta(R_\epsilon^{D-\alpha})$ for $D/2<\alpha<D$, or $\Theta(\log N)$ by $\Theta(\log(2+R_\epsilon))$ for $\alpha=D$. 
Substitution into Eq.~\eqref{eq:do-local-cost-H} gives the corresponding query bounds. 
For $\alpha>D$, this fixed-$M$ encoding has no leading-order system-size saving in its block-query normalization.

\subsection{Finite-precision inputs and implementation resources}
\label{app:do-implementation}

The query bounds assume exact input oracles. 
To connect them to finite-precision constructions, suppose the implemented inputs are Hermitian $H_R',O'$ and a state $\rho_R'$, with
\begin{equation}
 \|H_R'-H_R\|\leq\kappa_H,\qquad
 \|O'-O\|\leq\kappa_O,\qquad
 \|\rho_R'-\rho_R\|_1\leq\kappa_\rho.
 \label{eq:do-input-errors}
\end{equation}
Assume $\|O'\|\leq1$ and a valid normalization for $H_R'$.
Duhamel's formula gives $\|e^{-iH_R't}-e^{-iH_Rt}\|\leq t\kappa_H$. 
Splitting the three input
perturbations and integrating yields
\begin{equation}
 \left|\int_0^T\operatorname{Tr}
  [\rho_R'e^{iH_R't}O'e^{-iH_R't}]\,\mathrm{d}t-J_{H_R}(O)\right|
 \leq T^2\kappa_H+T\kappa_O+T\kappa_\rho.
 \label{eq:do-input-bias}
\end{equation}
For example, keep $\epsilon_{\rm trunc}(R,T)\leq\epsilon/2$, reserve $\epsilon/4$ for Eq.~\eqref{eq:do-input-bias}, and apply Lemma~\ref{lem:do-readout} with estimation error $\epsilon/4$. 
Sufficient choices are $\kappa_H=\mathcal O(\epsilon/T^2)$ and $\kappa_O,\kappa_\rho=\mathcal O(\epsilon/T)$ with suitable constants.
Equivalently, the normalized Hamiltonian-block error is $\mathcal O(\epsilon/(\alpha_H(R)T^2))$. 
These constant-factor budget changes preserve the main query scalings.

Primitive-query or gate accounting is obtained by multiplying each oracle count by the cost of implementing that oracle and adding preprocessing when relevant.

For an explicit gate estimate, let $g_H(R),g_\psi(R),g_O(R)$ be the gate costs of the respective oracles, including the controlled
versions used above, at the required precision. 
Let $w_R$ be the complete working width, including purification and oracle scratch registers.
Standard linear-size all-zero reflections and the simulation construction give the sufficient bound
\begin{equation}
 G=\widetilde{\mathcal O}\!\left[
  \frac{T}{\epsilon}\log\frac1\delta
  \left\{
   \bigl(1+\alpha_H(R)T\bigr)
       \bigl(g_H(R)+w_R\bigr)
   +g_\psi(R)+g_O(R)
  \right\}\right].
 \label{eq:do-gate-cost}
\end{equation}
Here $w_R$ accounts for ancillary reflections and control work. 
The system register uses $\mathcal O(|X_R|\log d_{\mathrm{loc}})$ qubits for local dimension $d_{\mathrm{loc}}$, or $\mathcal O(M\log|X_R|)$ position qubits in the fixed-$M$ encoding, in addition to purification and oracle workspaces and logarithmic clock registers.

\subsection{Choosing the radius and allocating the error}
\label{subsec:do-budget}
\label{app:do-budget}

The equal split is sufficient but need not minimize the cost. More generally, for any radius with $\epsilon_{\rm trunc}(R,T)<\epsilon$, one may allocate
\begin{equation}
 \epsilon_{\mathrm{est}}=\epsilon-\epsilon_{\rm trunc}(R,T).
\end{equation}
A larger region can cost more to simulate, but it also reduces the spatial error and leaves more error available for estimation. This tradeoff can be optimized explicitly in a simple power-law model.

Consider the continuous-radius scaling model $\epsilon_{\rm trunc}(R,T)=A(T)R^{-p}$ and $c_R\alpha_H(R)=C(T)R^s$, where $p,s>0$. Suppose also that $\alpha_H(R)T\geq1$ and that the Hamiltonian contribution dominates. 
Ignoring the logarithmic overhead of the implementation, the quantity being optimized is proportional to
\begin{equation}
 \frac{R^s}{\epsilon-A(T)R^{-p}}.
 \label{eq:do-budget-objective}
\end{equation}
For $x=A(T)/(\epsilon R^p)\in(0,1)$ its $x$ dependence is $x^{-s/p}/(1-x)$. 
Consequently,
\begin{equation}
 \frac{\mathrm{d}}{\mathrm{d}x}
 \log\!\left(\frac{x^{-s/p}}{1-x}\right)
 =-\frac{s}{px}+\frac{1}{1-x},
 \label{eq:do-budget-derivative}
\end{equation}
which has the unique minimum
\begin{equation}
 x_* =\frac{s}{p+s},\qquad
 R_* =\left(\frac{A(T)(p+s)}{s\epsilon}\right)^{1/p}.
 \label{eq:do-budget-minimizer}
\end{equation}
Thus the corresponding spatial and estimation budgets are $s\epsilon/(p+s)$ and $p\epsilon/(p+s)$. In applying this formula, enforce the spatial bound's validity condition, compare neighbouring feasible integer radii, and include the full-system option. 
If $s=0$ or other costs dominate, this simple power-law balance no longer determines the optimum.

A fair comparison uses the same estimation procedure and includes the cost of both encoding constructions. 
In the simulation-dominated regime, the leading resource expressions for a local and full implementation differ by
\begin{equation}
 \frac{c_R\alpha_H(R)}
      {c_\Lambda\alpha_H(\Lambda)}
 \frac{\epsilon}{\epsilon-\epsilon_{\rm trunc}(R,T)},
 \label{eq:do-fair-ratio}
\end{equation}
up to the logarithmic factors already displayed in the query bounds.
With an equal error split the second factor is two, so a small normalization reduction need not improve the total finite-size resource estimate.

Numerically selected radii can illustrate this tradeoff, but a rigorous resource guarantee still requires an a priori spatial error bound.

\section{Worst-case lower bounds for dense-output resources}
\label{app:do-lower}
\label{subsec:do-lower}

This section establishes two complementary worst-case lower bounds: one for the oracle complexity of dense-output estimation and one for the spatial information required in time-dependent finite-range systems. 

\subsection{Oracle lower bounds for estimation}

The following lower bounds show that the estimator's dependence on evolution
time and accuracy is necessary for some inputs, up to logarithmic factors.

\begin{proposition}[Worst-case oracle lower bounds for estimation]
\label{prop:do-lower}
For $\alpha_H T\geq1$ and $0<\epsilon\leq T/64$, there are
time-independent one-particle Hamiltonians on three modes, with
$\|H\|=\alpha_H/2$, a known initial state, and a known
single-site density observable, for which estimating the dense output
to error $\epsilon$ with success probability at least $2/3$ requires
\begin{equation}
 Q_H=\Omega\!\left(\frac{\alpha_H T^2}{\epsilon}\right)
 \label{eq:do-H-lower}
\end{equation}
queries to the supplied Hamiltonian block encoding. Separately, even
when $H=0$ and the observable is a known projector, some state-preparation
oracles require
\begin{equation}
 Q_\psi=\Omega(T/\epsilon).
 \label{eq:do-state-lower}
\end{equation}
Both statements allow controlled and inverse queries.
\end{proposition}

The Hamiltonian and state-preparation bounds use different input families.
They are oracle-model lower bounds and do not imply optimality of the chosen
radius or encoding. Suppose two possible supplied oracles are $V_0$ and $V_1$, and the rest of
the input and every oracle-independent operation are identical. For an
algorithm with worst-case query budget $q$, replace its oracle calls one
at a time. Unitarity and the triangle inequality bound the Euclidean
distance between its purified final states by
\begin{equation}
 q\|V_1-V_0\|.
 \label{eq:do-hybrid}
\end{equation}
The same estimate holds for inverse calls because
$\|V_1^\dagger-V_0^\dagger\|=\|V_1-V_0\|$, and for controlled calls by
their block-diagonal form. Intermediate measurements and classical
adaptivity can be included in the purification and controls.
If the two correct-answer intervals are disjoint, an estimator succeeding
with probability at least $2/3$ on each input also distinguishes the two
oracles with that success probability. Its output distributions then
have total variation distance at least $1/3$, so
$q\|V_1-V_0\|\geq1/3$ is necessary. Fix $\alpha_HT\geq1$ and
$0<\epsilon\leq T/64$. On the three-mode one-particle space, define
\begin{equation}
 H_\theta=\theta\bigl(|0\rangle\langle1|+|1\rangle\langle0|\bigr)
             +\frac{\alpha_H}{2}|2\rangle\langle2|,
 \quad
 |\psi\rangle=\frac{|0\rangle+i|1\rangle}{\sqrt2},
 \quad O=|0\rangle\langle0|.
 \label{eq:do-hard-H}
\end{equation}
The unknown parameter is either $0$ or
$\theta_*=8\epsilon/T^2$. Both the state preparation and the observable
are known and independent of that choice. The assumptions imply
$\theta_*T\leq1/8$ and
$\theta_*/\alpha_H\leq1/8$, and hence
$\|H_\theta\|=\alpha_H/2$ on both inputs.
The spectator mode $|2\rangle$ remains unoccupied.

The active two-mode dynamics gives exactly
\begin{equation}
 f_\theta(t)=\frac{1+\sin(2\theta t)}{2},\qquad
 J_\theta(T)=\frac T2+\frac{1-\cos(2\theta T)}{4\theta},
 \label{eq:do-hard-output}
\end{equation}
with the second expression interpreted continuously at $\theta=0$.
For $|x|\leq1$, $1-\cos x\geq x^2/3$, so
\begin{equation}
 J_{\theta_*}(T)-J_0(T)
 \geq\frac{\theta_*T^2}{3}
 =\frac{8\epsilon}{3}>2\epsilon.
 \label{eq:do-hard-gap}
\end{equation}

We must specify the entire supplied oracle, not just its upper-left
block. Set $A_\theta=H_\theta/\alpha_H$ and supply the
Hermitian unitary
\begin{equation}
 V_\theta=
 \begin{pmatrix}
  A_\theta&\sqrt{I-A_\theta^2}\\
  \sqrt{I-A_\theta^2}&-A_\theta
 \end{pmatrix}.
 \label{eq:do-hard-block}
\end{equation}
Since $A_\theta$ commutes with its square-root function, $V_\theta^2=I$,
and its upper-left block is exactly the promised Hamiltonian block.
The spectator blocks are independent of $\theta$. On the active subspace,
with $u=\theta/\alpha_H$,
\begin{equation}
 \|V_\theta-V_0\|
 =\sqrt{u^2+\bigl(\sqrt{1-u^2}-1\bigr)^2}
 \leq2|u|\qquad(|u|\leq1/8).
 \label{eq:do-hard-oracle-distance}
\end{equation}
Combining Eqs.~\eqref{eq:do-hybrid}, \eqref{eq:do-hard-gap}, and
\eqref{eq:do-hard-oracle-distance} gives
\begin{equation}
 q\geq\frac{\alpha_H}{6\theta_*}
 =\frac{\alpha_HT^2}{48\epsilon}.
 \label{eq:do-H-lower-explicit}
\end{equation}
The construction is also a number-conserving quadratic fermionic
Hamiltonian,
$\theta(c_0^\dagger c_1+c_1^\dagger c_0)
 +(\alpha_H/2)n_2$, restricted to one particle.
Padding to qubit registers introduces no parameter-dependent data. This lower bound applies to the supplied block oracle; stronger input access can change the problem. Now take $H=0$, $O=|1\rangle\langle1|$, and supply
\begin{equation}
 U_\theta=e^{-i\theta Y},\qquad
 U_\theta|0\rangle=\cos\theta|0\rangle+\sin\theta|1\rangle.
 \label{eq:do-hard-state}
\end{equation}
The dense output is $J_\theta=T\sin^2\theta$. Choose
$\theta_0=\pi/4$ and $\theta_1=\pi/4+8\epsilon/T$. Since
$16\epsilon/T\leq1/4$ and $\sin x\geq x/2$ for $0\leq x\leq1$,
\begin{equation}
 J_{\theta_1}-J_{\theta_0}
 =\frac T2\sin(16\epsilon/T)\geq4\epsilon,
 \qquad
 \|U_{\theta_1}-U_{\theta_0}\|\leq\frac{8\epsilon}{T}.
 \label{eq:do-state-gap}
\end{equation}
The same hybrid estimate gives $q\geq T/(24\epsilon)$, proving
Eq.~\eqref{eq:do-state-lower}. The Hamiltonian and state bounds use separate hard families.

\subsection{A spatial-information lower bound with a localized input}
\label{app:do-radius-lower}

This result concerns the larger class of time-dependent, finite-range
Hamiltonians and is separate from the time-independent estimation lower bound.

\begin{proposition}[Ballistic spatial access can be necessary]
\label{prop:do-radius-lower}
For each integer $R\geq1$ and $J>0$, there are two piecewise-constant,
nearest-neighbour quadratic fermionic Hamiltonians with local term norms
at most $J$, a particle initially at site $0$, and observable $n_0$, whose
restrictions to $X_R=\{0,\ldots,R\}$ agree at every time. If
\begin{equation}
 T\geq\frac{2\pi(R+3/2)}{J},
 \label{eq:do-radius-lower-window}
\end{equation}
their dense outputs differ by at least $T/2$. Hence a method using only
the restricted Hamiltonian data cannot guarantee additive error
$\epsilon\leq T/8$ for this class, regardless of the number of such local
queries. A uniform spatial-access guarantee therefore requires
$R=\Omega(JT)$ in the worst case.
\end{proposition}

\begin{proof}
Let $L=R+1$ and work on the one-particle subspace of sites $0,\ldots,L$.
Write $X_{j,j+1}=|j\rangle\langle j+1|+|j+1\rangle\langle j|$.
Start in $|0\rangle$. Apply a splitting pulse
$\exp[-i(\pi/4)X_{0,1}]$, using the Hamiltonian $JX_{0,1}$ for
$\pi/(4J)$. The state becomes
$(|0\rangle-i|1\rangle)/\sqrt2$.

Transport the mobile component from $1$ to $L$ using the pulses $\exp[-i(\pi/2)X_{j,j+1}]$,
$j=1,\ldots,L-1$, each of duration $\pi/(2J)$. Denote their product by
$V$. During a further interval $\pi/J$, the two inputs use respectively
zero Hamiltonian or $J|L\rangle\langle L|$. Label this choice by
$b\in\{0,1\}$. Only this on-site term, outside $X_R$, depends on $b$.
It multiplies the mobile amplitude by $(-1)^b$.

Apply $V^\dagger$ using the reverse pulse order and hopping $-J$, and
then the inverse splitting pulse. The state immediately before the last
pulse is
$(|0\rangle-i(-1)^b|1\rangle)/\sqrt2$. The final state is $|0\rangle$
for $b=0$ and $i|1\rangle$ for $b=1$. The complete protocol takes
\begin{equation}
 \tau=\frac{\pi}{2J}+
       2(L-1)\frac{\pi}{2J}+
       \frac{\pi}{J}
 =\frac{\pi(L+1/2)}{J}
 =\frac{\pi(R+3/2)}{J}.
 \label{eq:do-roundtrip-time}
\end{equation}
Set both Hamiltonians to zero during the remaining time $[\tau,T]$.

The central density is identical for the two inputs before the final
splitting pulse. During that pulse, at elapsed time
$0\leq s\leq\pi/(4J)$, their density difference is
$\sin(2Js)\geq0$. Afterwards it is one. Consequently
\begin{equation}
 J_{b=0}(n_0)-J_{b=1}(n_0)
 =T-\tau+\frac{1}{2J}\geq T-\tau.
 \label{eq:do-radius-output-gap}
\end{equation}
Under Eq.~\eqref{eq:do-radius-lower-window}, this is at least $T/2$.
Yet all local Hamiltonian data, the initial state, and the observable are identical
for the two inputs. Their $\epsilon$-accurate answer intervals are disjoint,
so no algorithm with access only to that local data can succeed with probability
at least $2/3$ on both. This proves the proposition. Decoupled sites can
be added to place the construction in the bulk of a larger lattice.
\end{proof}

\end{document}